\documentclass{uwstat572}
\usepackage{amsmath}

\usepackage{tikz}
\usepackage{bm}
\usepackage{graphicx}
\usepackage{colortbl}
\usepackage{amsmath}
\usepackage{subcaption}
\usepackage{multirow}
\usepackage{float}
\usepackage{setspace}
\usepackage{amsthm}
\usepackage{caption}
\usepackage{amssymb}
\usepackage{todonotes}
\usepackage{hyperref}
\usepackage{multirow}

\usepackage[plain,noend]{algorithm2e}
\newtheorem{algo}{Algorithm}

\makeatletter
\renewcommand{\algocf@captiontext}[2]{#1\algocf@typo. \AlCapFnt{}#2} % text of caption
\def\@algocf@capt@plain{top}
\renewcommand{\algocf@makecaption}[2]{%
  \addtolength{\hsize}{\algomargin}%
  \sbox\@tempboxa{\algocf@captiontext{#1}{#2}}%
  \ifdim\wd\@tempboxa >\hsize%     % if caption is longer than a line
    \hskip .5\algomargin%
    \parbox[t]{\hsize}{\algocf@captiontext{#1}{#2}}% then caption is not centered
  \else%
    \global\@minipagefalse%
    \hbox to\hsize{\box\@tempboxa}% else caption is centered
  \fi%
  \addtolength{\hsize}{-\algomargin}%
}
\makeatother

\usepackage[labelfont=bf]{caption}
\hypersetup{citecolor=blue}
\DeclareMathOperator*{\argmax}{arg\,max}
\DeclareMathOperator{\tree}{\textsc{tree}}
\DeclareMathOperator{\desc}{\textsc{desc}}
\DeclareMathOperator{\branch}{\textsc{branch}}
\DeclareMathOperator{\gain}{\textsc{gain}}
\DeclareMathOperator{\term}{\textsc{term}}

\def\T{{ \mathrm{\scriptscriptstyle T} }}

\newtheorem{proposition}{Proposition}

\newtheorem{lemma}{Lemma}
\newtheorem{definition}{Definition}
\newtheorem{condition}{Condition}
\newtheorem{theorem}{Theorem}

\usepackage{ bbold }
\usepackage{authblk}
  
 \usepackage{import}
  \usepackage{ dsfont }
  \usepackage{amsmath}
  \usepackage{xr}

\title{\emph{Tree-Values:} selective inference for regression trees}

\author[a]{Anna C. Neufeld}
\author[b]{Lucy L. Gao}
\author[a,c]{Daniela M. Witten}
\affil[a]{Department of Statistics, University of Washington}
\affil[b]{Department of Statistics, University of British Columbia}
\affil[c]{Department of Biostatistics, University of Washington}

\begin{document}
\maketitle

\begin{abstract}%
We consider conducting inference on the output of the Classification and Regression Tree (CART) \citep{breiman1984classification} algorithm. A naive approach to inference that does not account for the fact that the tree was estimated from the data will not achieve standard  guarantees, 
such as Type 1 error rate control and nominal coverage. 
 Thus,  we propose a 
selective inference framework for conducting inference on a fitted CART tree. In a nutshell, we condition on the fact that the tree was estimated from the data.  We propose a test for the difference in the mean response between a pair of terminal nodes that controls the selective Type 1 error rate, and  a confidence interval for the mean response within a single terminal node that attains the nominal selective coverage. Efficient algorithms for computing the necessary conditioning sets are provided. We apply these methods in simulation and to a dataset involving the association between portion control interventions and caloric intake. 
 \end{abstract}%

 \newpage
 
\section{Introduction}
\label{section_intro}

Regression tree algorithms recursively partition  covariate space using binary splits to obtain regions that are maximally homogeneous with respect to a continuous response. 
 The Classification and Regression Tree (CART; \citealt{breiman1984classification}) proposal, which involves growing a large tree and then pruning it back, is by far the most popular of these algorithms. 

The regions defined by the splits in a fitted CART tree induce a piecewise constant regression model
 where the predicted response within each region is the mean of the observations in that region. CART is popular in large part because it is highly interpretable; someone without technical expertise can easily “read” the tree to make predictions, and to understand why a certain prediction is made. However, its interpretability belies the fact that CART trees are highly unstable: a small change to the training dataset can drastically change the structure of the fitted tree. In the absence of an established notion of statistical significance associated with a given split in the tree, it is hard for a practitioner to know whether they are interpreting signal or noise. In this paper, we use the framework of selective inference to fill this gap by providing a toolkit to conduct inference on hypotheses motivated by the output of the CART algorithm.

Given a CART tree,  consider testing for a difference in the mean response of the regions resulting from a binary split. A very naive approach, such as a two-sample $Z$-test, that   does not account for the fact that the regions were themselves estimated from the data will fail to control the selective Type 1 error rate: the probability of rejecting a true null hypothesis, given that we decided to test it \citep{fithian2014optimal}. Similarly, a naive $Z$-interval for the mean response in a region will not attain nominal selective coverage: the probability that the interval covers the parameter, given that we chose to construct it.

In fact, approaches for conducting inference on the output of a regression tree are quite limited. Sample splitting involves fitting a CART tree using a subset of the observations, which will naturally lead to an inferior tree to the one resulting from all of the observations, and thus is unsatisfactory in many applied settings; see \cite{athey2016recursive}. 
 \citet{wager2015adaptive} develop convergence guarantees for unpruned CART trees that can be leveraged to build confidence intervals for the mean response within a region; however, they do not provide finite-sample results and cannot accommodate pruning. \cite{loh2016identification} and \cite{loh2019subgroups} develop bootstrap calibration procedures that attempt to provide confidence intervals for the regions of a regression tree. In Appendix~\ref{appendix:loh}, we show that this bootstrap calibration approach fails to provide intervals that achieve nominal coverage for the parameters of interest in this paper. 
 
As an alternative to performing inference on a CART tree, one could turn to the conditional inference tree (CTree) framework of \cite{hothorn2006unbiased}. This framework uses a different tree-growing algorithm than CART, and at each split tests for linear association between the split covariate and the response. As summarized in \cite{loh2014fifty}, the CTree framework alleviates issues with instability and variable selection bias associated with CART. Despite these advantages,  CTree remains far less widely-used than CART. Furthermore, while CTree attaches a notion of statistical significance to each split in a tree, it does not directly allow for inference on the mean response within a region or the difference in mean response between two regions. Finally, while the CTree framework requires few assumptions, its inference is based on asymptotics. 

In this paper, we introduce a finite-sample selective inference \citep{fithian2014optimal} framework for the difference between the mean responses in two regions, and for the mean response in a single region, in a pruned or unpruned CART tree. 
We condition on the event  that CART yields a particular set of regions, and thereby achieve selective Type 1 error rate control as well as nominal selective coverage. 

The rest of this paper is organized as follows. In Section~\ref{section_Background}, we review the CART algorithm, and briefly define some key ideas in selective inference. In Section~\ref{section_frame}, we present our proposal for selective inference on the regions estimated via CART. We show that the necessary conditioning sets can be efficiently computed in Section~\ref{section_computing}. 
In Section~\ref{section_simstudy} we compare our framework to sample splitting and CTree via simulation. In Section~\ref{section_realData} we compare our framework to CTree on data from the Box Lunch Study. The discussion is in Section~\ref{section_disc}. Technical details are relegated to the supplementary materials.  

\section{Background}
\label{section_Background}

\subsection{Notation for Regression Trees}
\label{subsec_treenotation} 

Given $p$ covariates $(X_1,\ldots,X_p)$ measured on each of $n$ observations $(x_1,\ldots,x_n)$,  let $x_{j, (s)}$ denote the $s$th order statistic of the $j$th covariate, and define the half-spaces 
\begin{equation}
\label{eq:halfspace}	
\chi_{j,s,1} = \left\{z \in \mathbb{R}^p: z_{j} \leq x_{j, (s)} \right\}, \ \ \ \ \ \ \ \chi_{j,s,0} = \left\{z \in \mathbb{R}^p : z_{j} > x_{j, (s)} \right\}.
\end{equation}
The following definitions are illustrated in Figure~\ref{fig_simpleTree}. 

\begin{definition}[Tree and Region]
\label{def_tree}
Consider a set $\mathcal{S}$ such that $R \subseteq \mathbb{R}^p$ for all $R \in \mathcal{S}$. Then $\mathcal{S}$ is a \emph{tree} if and only if (i) $\mathbb{R}^p \in \mathcal{S}$; (ii) every element of $\mathcal{S} \setminus \mathbb{R}^p$ equals $R \cap \chi_{j,s,e}$ for some $R \in \mathcal{S}$,  
$j \in \{1,\ldots,p\}$, $s \in \{1,\ldots,n-1\}$, $e \in \{0,1\}$; (iii) $R \cap \chi_{j,s,e} \in \mathcal{S}$ implies that $R \cap \chi_{j,s,1-e} \in \mathcal{S}$ for $e \in \{0,1\}$; and (iv) for any $R, R' \in \mathcal{S}, R \cap R' \in \{\emptyset, R, R'\}$. If $R \in \mathcal{S}$ and $\mathcal{S}$ is a tree, then we refer to $R$ as a \emph{region}. 
\end{definition}
We use the notation $\tree$ to refer to a particular tree.  
Definition~\ref{def_tree} implies that any region $R \in \tree \setminus \{\mathbb{R}^p\}$ is of the form  $R = \cap_{l=1}^L \chi_{j_l,s_l,e_l}$, where for each $l=1,\ldots,L$, we have that $j_l \in \{1,\ldots,p\}$, $s_l \in \{1,\ldots,n-1\}$, and $e_l \in \{0,1\}$. We call $L$ the \emph{level} of the region, and use the convention that the level of $\mathbb{R}^p$ is $0$. 

\begin{definition}[Siblings and Children]
Suppose that $\{R, R \cap \chi_{j,s,1},R \cap \chi_{j,s,0}\} \subseteq \tree$. Then $R \cap \chi_{j,s,1}$ and $R \cap \chi_{j,s,0}$ are \emph{siblings}. Furthermore, they are the \emph{children} of $R$.  
\end{definition}

\begin{definition}[Descendant and Ancestor]
\label{def_descendant}
If $R,R' \in \tree$ and $R \subseteq R'$, then $R$ is a \emph{descendant} of $R'$, and $R'$ is an \emph{ancestor} of $R$. 
\end{definition}

\begin{definition}[Terminal Region]
\label{def_term}
A region $R \in \tree$ without descendants is a  \emph{terminal region}. 
\end{definition}
We let $\desc(R, \tree)$ denote the set of descendants of region $R$ in $\tree$, and we let $\term(R,\tree)$ denote the subset of $\desc(R, \tree)$ that are terminal regions.

 Given a response vector ${y} \in \mathbb{R}^n$, let $\bar{y}_R =  \left( \sum_{i: x_i \in R} y_i \right) / \left\{ \sum_{i=1}^n 1_{(x_i \in R)} \right\}$, where $1_{(A)}$ is an indicator variable that equals $1$ if the event $A$ holds, and $0$ otherwise. Then, a tree $\tree$ induces the regression model $\hat{\mu}(x) = \sum_{R \in \term(\mathbb{R}^p, \tree)} \bar{y}_R 1_{(x \in R)}$. In other words, it predicts the response within each terminal region to be the mean of the observations in that region. 

\begin{figure}
\centering
\includegraphics[width=6.2cm, height=4cm]{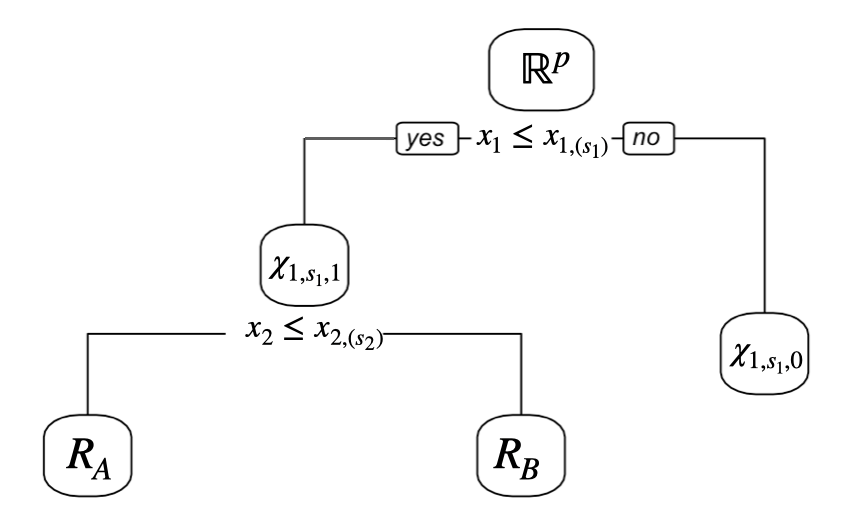}
\hspace{20mm}
\begin{tikzpicture}[scale=0.29, sibling distance=7em]
\node[fill, outer sep=0pt, inner sep = 0pt] (1) at (-5, -5) {};
\node[fill, outer sep=0pt, inner sep = 0pt]  (2) at (-5, 5) {};
\node[fill, outer sep=0pt, inner sep = 0pt]  (3) at (6.7, 5) {};
\node[fill, outer sep=0pt, inner sep = 0pt]  (4) at (6.7, -5) {} ;
\node[fill, outer sep=0pt, inner sep = 0pt]  (5) at (2.7, -5) {} ;
\node[fill, outer sep=0pt, inner sep = 0pt]  (6) at (2.7, 5) {} ;
\node (7) at (3.9, -6) {$x_{1,(s_1)}$} ;
\node[fill, outer sep=0pt, inner sep = 0pt]  (8) at (-5, 2.5) {} ;
\node[fill, outer sep=0pt, inner sep = 0pt]  (9) at (2.7, 2.5) {} ;
\node (10) at (-7, 2.5) {$x_{2,(s_2)}$} ;
\node[outer sep=0pt, inner sep = 0pt]  (11) at (2.5, -1) {} ;
\node[outer sep=0pt, inner sep = 0pt]  (12) at (5, -1) {} ;
\node (15) at (-2, -1.5) {$R_A$};
\node (16) at (-2, 3.5) {$R_B$} ;
\node (18) at (4.7, 0) {$\chi_{1,s_1,0}$} ;
\path[-,font=\scriptsize]
(1) edge[thick] node[left] {$\bm{X_2}$} (2)
(2) edge[thick] node[above] {}(3)
(3) edge[thick] node[above] {}(4)
(4) edge[thick] node[below] {$\bm{X_1}$}(1)
(8) edge[thick] node[below] {} (9)
(5) edge[thick] node[below] {} (6);
\end{tikzpicture}

\caption{The regression tree takes the form  $\tree = \{ \mathbb{R}^p, \chi_{1,s_1,1}, \chi_{1,s_1,0}, \chi_{1,s_1,1} \cap \chi_{2,s_2,1}, \chi_{1,s_1,1} \cap \chi_{2,s_2,0}\}$. The regions 
 $R_A = \chi_{1,s_1,1} \cap \chi_{2,s_2,1}$ and  $R_B = \chi_{1,s_1,1} \cap \chi_{2,s_2,0}$ are siblings, and are children, and therefore descendants, of the region $\chi_{1,s_1,1}$. The ancestors of $R_A$ and $R_B$ are $\mathbb{R}^p$ and $\chi_{1,s_1,1}$.  Furthermore, $R_A$, $R_B$, and $\chi_{1,s_1,0}$ are terminal regions. 
}  
\label{fig_simpleTree}
\end{figure}
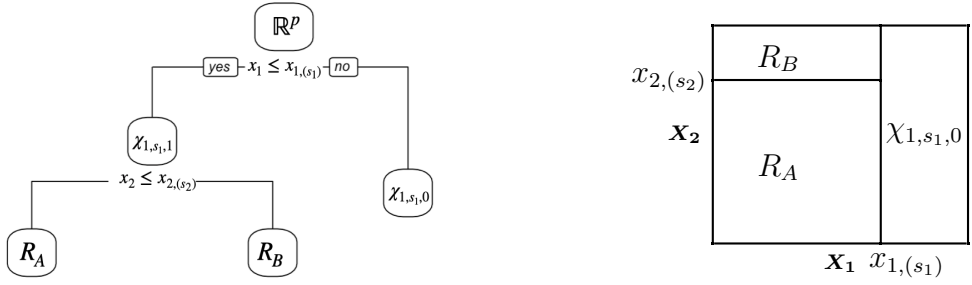

\subsection{A Review of the CART Algorithm \citep{breiman1984classification}}
\label{subsec_CARTreview} 

The CART algorithm \citep{breiman1984classification} greedily searches for a tree that minimizes the sum of squared errors  
$\sum_{R \in \term(\mathbb{R}^p, \tree)} \sum_{i: x_i \in R} (y_i - \bar{y}_R)^2$. It first grows a very large tree via recursive binary splits, starting with the full covariate space $\mathbb{R}^p$. To split a region $R$, it selects the covariate  $x_{j}$ and the split point $x_{j,(s)}$ to maximize the \emph{gain}, defined as 
\begin{equation}
\label{eq:gain}
\small
\textsc{gain}_R\left(y,j,s\right) \equiv \sum_{i \in R} \left(y_i - \bar{y}_R\right)^2 - \left\{ \sum_{i \in R \cap \chi_{j,s, 1}} \left(y_i - \bar{y}_{R \cap \chi_{j,s,1}}\right)^2 + \sum_{i \in R \cap \chi_{j,s,0}} \left(y_i - \bar{y}_{R \cap \chi_{j,s,0}}\right)^2 \right\}. 
\end{equation}
\normalsize
Details are provided in Algorithm~\ref{alg_growing}. 

Once a very large tree has been grown, cost-complexity pruning is applied. We define the average per-region gain in sum-of-squared errors provided by the descendants of a region $R$, 
\small
\begin{equation}
\label{def_g}	
g(R, \tree,y) = \frac{\underset{i: x_i \in R}{\sum} (y_i - \bar{y}_{R})^2 - \underset{r \in \term(R, \tree)}{\sum} \ \ \underset{i: x_i \in r}{\sum} (y_i - \bar{y}_r)^2}{ |\term(R, \tree)|-1}.
\end{equation}
\normalsize
Given a complexity parameter $\lambda \geq 0$, if $g(R, \tree,y) < \lambda$ for some $R \in \tree$, then cost-complexity pruning removes $R$'s descendants from $\tree$, turning $R$ into a terminal region. Details are  in Algorithm~\ref{alg_pruning}, which involves the notion of a \emph{bottom-up ordering}. 

\begin{definition}[Bottom-up ordering]
\label{def_bottomup}
Let $\tree = \{R_1,\ldots, R_K\}$. Let $\pi$ be a permutation of the integers $(1, \ldots, K)$. Then $\mathcal{O} = \left(R_{\pi(1)}, \ldots, R_{\pi(K)}\right)$ is a \emph{bottom-up ordering} of the regions in $\tree$ if, for all $k =1,\ldots,K$, $\pi(k) \leq \pi(j)$ if $R_k \in \desc(R_j, \tree)$.
\end{definition}
There are other equivalent formulations for cost-complexity pruning (see Proposition 7.2 in \cite{ripley1996pattern}); the formulation in Algorithm~\ref{alg_pruning} is  convenient for establishing the results in this paper.

To summarize, the CART algorithm 
first applies Algorithm~\ref{alg_growing} to the initial region $\mathbb{R}^p$ and the data $y$ to obtain an unpruned tree, 
 which we call $\tree^0(y)$. It then applies Algorithm~\ref{alg_pruning} to $\tree^0(y)$ to obtain an optimally-pruned tree using complexity parameter $\lambda$, which we call $\tree^\lambda(y)$. 

\SetKwFunction{printgrow}{\textsc{Grow}}
\SetKwFunction{printprune}{\textsc{Prune}}
 
\begin{algo}[Growing a tree] \hspace{5mm}\\
 \printgrow{$R$, $y$}
 \vspace{-3mm}
\begin{tabbing}
  \qquad 1. If a stopping condition is met, return $R$. \\
  \qquad 2. Else return $\{ R, $ \printgrow{$R \cap \chi_{\tilde{j}, \tilde{s}, 1}, y$}, \printgrow{$R \cap \chi_{\tilde{j}, \tilde{s}, 0}$, $y$}$\}$, where \\
 \qquad \qquad \qquad $
 (\tilde{j},\tilde{s}) \in 
\argmax_{(j,s):
s \in \{1,\ldots,n-1\}, j \in \{1,  \ldots, p\}}
\textsc{gain}_R\left(y,j,s\right).
 $
\end{tabbing}
\label{alg_growing}
\end{algo}

\begin{algo}[Cost-complexity pruning]
\label{alg_pruning}
Parameter $\mathcal{O}$ is a bottom-up ordering of the $K$ regions in $\tree$.\\
\printprune{$\tree$, $y$, $\lambda$, $\mathcal{O}$}
\begin{tabbing}
\qquad 1. Let $\tree_0=\tree$. Let $K$ be the number of regions in $\tree_0$. \\
\qquad 2. For $k=1,\ldots,K$: \\
 \qquad  \qquad (a) Let $R$ be the $k$th region in $\mathcal{O}$. \\
 \qquad  \qquad (b) Update $\tree_k$ as follows, where $g(\cdot)$ is defined in \eqref{def_g}: \\
\qquad \qquad \qquad 
$\tree_k \leftarrow 
\begin{cases}
\tree_{k-1}\setminus \desc(R, \tree_{k-1}) & \textrm{ if }  g(R, \tree_{k-1},y) < \lambda , \\
\tree_{k-1} & \textrm {otherwise}.
\end{cases}$ \\
\normalsize
\qquad 3. Return $\tree_K$.
\end{tabbing}
\end{algo}

\subsection{A Brief Overview of Selective Inference}
\label{subsection_selectivereview}

Here, we provide a very brief overview of selective inference; see \cite{fithian2014optimal} or \cite{taylor2015statistical} for a more detailed treatment. 

Consider conducting inference on a parameter $\theta$. Classical approaches assume that we were already interested in conducting inference on $\theta$ before looking at our data. If, instead, our interest in $\theta$ was sparked by looking at our data, then 
 inference must be performed with care: we must account for the fact that we ``selected" $\theta$ based on the data \citep{fithian2014optimal}. 
 In this setting, interest focuses on a p-value $p(Y)$ such that the test for $H_0: \theta = \theta_0$ based on $p(Y)$ controls the \emph{selective Type 1 error} rate, in the sense that 
\begin{equation}
\label{eq_st1e}
pr_{H_0: \theta = \theta_0} \left\{ p(Y) \leq \alpha  \mid \theta \text{ selected} \right\} \leq \alpha, \text{ for all } 0 \leq \alpha \leq 1.
\end{equation}
Also of interest are confidence intervals $[L(Y),U(Y)]$ that achieve \emph{$(1-\alpha)$-selective coverage} for the parameter $\theta$, meaning that
\begin{equation}
\label{eq_selcov}
pr\left\{ \theta \in [L(Y),U(Y)] \mid \theta \text{ selected} \right\} \geq 1-\alpha.
\end{equation}
Roughly speaking, the inferential guarantees in \eqref{eq_st1e} and \eqref{eq_selcov} can be  achieved by defining p-values and confidence intervals that condition on the aspect of the data that led to the selection of $\theta$.  In recent years, a number of papers have taken this approach to perform selective inference on parameters selected from the data in the  regression \citep{lee2016exact,  liu2018more, tian2018selective, tibshirani2016exact}, clustering \citep{gao2020selective}, and changepoint detection \citep{hyun2018post, jewell2019testing} settings.

In the next section, we propose p-values that satisfy~\eqref{eq_st1e} and confidence intervals that satisfy~\eqref{eq_selcov}
in the setting of CART, where the parameter of interest is either the mean response within a region, or the difference between the mean responses of two sibling regions.

\section{The Selective Inference Framework for CART}
\label{section_frame}

\subsection{Inference on a Pair of Sibling Regions}
\label{subsec_siblings}

Throughout this paper, we assume that $Y \sim N_n({\mu}, \sigma^2 I_n)$ with $\sigma > 0$ known. 

We let $X \in \mathbb{R}^{n \times p}$ denote a fixed covariate matrix. Suppose that we apply CART with complexity parameter $\lambda$ to a realization $y = (y_1,\ldots,y_n)^\T$ from $Y$ to obtain $\tree^{\lambda}(y)$. 
 Given sibling regions $R_A$ and $R_B$ in $\tree^\lambda(y)$, we define a contrast vector $\nu_{sib} \in \mathbb{R}^n$ such that 
 \begin{equation}
\label{eq_nusib}
\left(\nu_{sib}\right)_i = \frac{1_{(x_i \in R_A)}}{\sum_{i'=1}^n 1_{(x_{i'} \in R_A)}} - \frac{1_{(x_i \in R_B)}}{\sum_{i'=1}^n 1_{(x_{i'} \in R_B)}},
\end{equation} 
and $\nu_{sib}^\T \mu = \left( \sum_{i: x_i \in R_A} \mu_i\right) /   \left\{ \sum_{i=1}^n 1_{(x_i \in R_A)} \right\}       - \left( \sum_{i: x_i \in R_B} \mu_i \right) / \left\{ \sum_{i=1}^n 1_{(x_i \in R_B)} \right\}$.
 Now, consider testing the null hypothesis of no difference in means between $R_A$ and $R_B$, i.e.  $H_0: \nu_{sib}^\T{\mu}  = 0  \text{ versus } H_1: \nu_{sib}^\T{\mu} \neq 0$.  
 This null hypothesis is of interest because  $R_A$ and $R_B$ appeared as siblings in $\tree^\lambda(y)$. A test based on a p-value of the form $ pr_{H_0} \left( |\nu_{sib}^\T Y| \geq |\nu_{sib}^\T {y}| \right)$ that does not account for this will not control the selective Type 1 error rate   
 in \eqref{eq_st1e}.

To control the selective Type 1 error rate, we propose a p-value that conditions on the aspect of the data that led us to select $\nu_{sib}^\T \mu$,
 \begin{equation}
\label{eq_prePval}
 pr_{H_0} \left\{ |\nu_{sib}^\T Y| \geq |\nu_{sib}^\T {y}| \mid R_A,R_B \text{ are siblings in } \tree^\lambda(Y) \right\}.
\end{equation}
But  \eqref{eq_prePval}
 depends on a nuisance parameter, the portion of $\mu$ that is orthogonal to $\nu_{sib}$. To remove the dependence on this nuisance parameter, we condition on its sufficient statistic $\mathcal{P}_{\nu_{sib}}^\perp Y$, where $\mathcal{P}_{\nu}^\perp = I - \nu \nu^\T / \|\nu\|_2^2$. The resulting p-value, or ``tree-value", is 
defined as
\begin{equation}
\label{def_mainpval}
p_{sib}(y) = 
pr_{H_0} \left\{ |\nu_{sib}^\T Y| \geq |\nu_{sib}^\T {y}| \mid  R_A, R_B \text{ are siblings in } \tree^\lambda(Y), \mathcal{P}_{\nu_{sib}}^\perp Y = \mathcal{P}_{\nu_{sib}}^\perp {y} \right\}.
\end{equation}

Results similar to Theorem~\ref{theorem_1jewell} can be found in \cite{jewell2019testing, lee2016exact, liu2018more}, and \cite{tibshirani2016exact}. 

\begin{theorem}
\label{theorem_1jewell}
The test based on the p-value $p_{sib}(y)$ in \eqref{def_mainpval} controls the selective Type 1 error rate for $H_0: \nu_{sib}^\T \mu=0$, where $\nu_{sib}$ is defined in \eqref{eq_nusib}, in the sense that
\begin{equation}
pr_{H_0}\left\{ p_{sib}(Y) \leq \alpha \mid R_A, R_B \text{ are siblings in } \tree^\lambda(Y)  \right\} = \alpha, \text{ for all } 0 \leq \alpha \leq 1. \label{eq:type1}
\end{equation}
Furthermore, $p_{sib}(y) = pr\left\{|\phi| \geq |\nu_{sib}^\T {y}| \mid \phi \in S^\lambda_{sib}(\nu_{sib}) \right\},$
where $\phi \sim N(0, \|\nu_{sib}\|_2^2 \sigma^2)$, $y'(\phi,\nu) = \mathcal{P}_{\nu}^\perp {y} + \phi( \nu / \|\nu \|_2^2)$, and 
\begin{equation}
\label{def_S}
S_{sib}^\lambda(\nu_{sib}) = \{ \phi : R_A, R_B \text{ are siblings in } \tree^\lambda\{y'(\phi,\nu_{sib})\} \}. 
\end{equation}
\end{theorem}
Proofs of all theoretical results are provided in the appendix. Theorem~\ref{theorem_1jewell} says that given the set $S^\lambda_{sib}(\nu_{sib})$, we can compute the p-value in \eqref{def_mainpval} using
\begin{equation}
\label{truncNormCDFpval}
p_{sib}(y) = 1 - F\left\{|\nu_{sib}^\T {y}| ; 0, \|\nu_{sib}\|_2^2 \sigma^2, S^\lambda_{sib}(\nu_{sib}) \right\}  + F\left\{-|\nu_{sib}^\T {y}| ; 0, \|\nu_{sib}\|_2^2 \sigma^2, S^\lambda_{sib}(\nu_{sib}) \right\},
\end{equation}
where $F\left( ~\cdot~ ; 0, \|\nu\|^2 \sigma^2, S \right)$ denotes the cumulative distribution function of the $N(0, \|\nu \|_2^2 \sigma^2)$ distribution truncated to the set $S$. In Section~\ref{section_computing}, we provide an efficient approach for analytically characterizing the truncation 
set $S^\lambda_{sib}(\nu_{sib})$. To avoid numerical issues associated with the truncated normal distribution, we compute \eqref{truncNormCDFpval} using methods described in the supplement of
\cite{chen2019valid}. Note that the proof of Theorem~\ref{theorem_1jewell}, and consequently the efficient computation of $p_{sib}(y)$ discussed in Section~\ref{section_computing}, relies on the assumption that $Y\sim N_n(\mu, \sigma^2 I_n)$.

We now consider inverting the test proposed in  
\eqref{def_mainpval} to construct an equitailed confidence interval for $\nu_{sib}^\T{\mu}$ that has $(1-\alpha)$-selective coverage \eqref{eq_selcov}, in the sense that
\begin{equation}
pr\left\{\nu_{sib}^\T {\mu} \in \left[ L(Y),U(Y) \right] \mid R_A, R_B \text{ are siblings in } \tree^\lambda(Y) \right\} = 1-\alpha.
\label{eq:selcov}
\end{equation}
  
\begin{proposition}
 \label{prop_CI}
 For any $0 \leq \alpha \leq 1$ and any realization $y \in \mathbb{R}^n$,
the values $L(y)$ and $U(y)$ that satisfy
\small
\begin{equation}
\label{eq_mainCI}
F \{\nu_{sib}^\T y ; L(y), \sigma^2\|\nu_{sib}\|_2^2, S^\lambda_{sib}(\nu_{sib})\} = 1 - \alpha/2,   \;\;\; F \{\nu_{sib}^\T y ; U(y), \sigma^2\|\nu_{sib}\|_2^2, S^\lambda_{sib}(\nu_{sib})\}  =\alpha/2,
\end{equation}	
\normalsize
are unique, and $[L(Y), U(Y)]$ achieves $(1-\alpha)$-selective coverage for $\nu_{sib}^\T \mu$.
 \end{proposition}

\subsection{Inference on a Single Region}
\label{subsec_framingsingle}

Given a single region $R_A$ in a CART tree, we define the contrast vector $\nu_{reg}$ such that 
\begin{equation}
\label{eq_nureg}
(\nu_{reg})_i = 1_{(x_i \in R_A)}/\left\{ \sum_{i'=1}^n 1_{(x_{i'} \in R_A)}\right\}. 
\end{equation}
Then, $\nu_{reg}^\T \mu = \left( \sum_{i: x_i \in R_A} \mu_i \right)/ \left\{ \sum_{i=1}^n 1_{(x_i \in R_A)}\right\}$. We now consider testing the null hypothesis $H_0: \nu_{reg}^\T \mu=c$ for some fixed $c$. Because our interest in this null hypothesis results from the fact that 
$R_A  \in \tree^\lambda(y)$, we must condition on this event in defining the p-value. We define
\begin{equation}
p_{reg}(y)=pr_{H_0}\left\{ | \nu_{reg}^\T Y - c | \geq | \nu_{reg}^\T y - c | \mid R_A \in \tree^\lambda(Y), \mathcal{P}_{\nu_{reg}}^\perp Y = \mathcal{P}_{\nu_{reg}}^\perp {y} \right\},
\label{eq:pvalreg}
\end{equation}
and introduce the following theorem. 

\begin{theorem}
\label{theorem_1modified}
The test based on the p-value $p_{reg}(y)$ in \eqref{eq:pvalreg} controls the selective Type 1 error rate for $H_0: \nu_{reg}^\T \mu = c$, where $\nu_{reg}$ is defined in \eqref{eq_nureg}. Furthermore, 
$p_{reg}(y) = pr\left\{|\phi - c| \geq |\nu_{reg}^\T {y} - c| \mid \phi \in S_{reg}(\nu_{reg}) \right\},$
where $\phi \sim N(c, \|\nu_{reg}\|_2^2 \sigma^2)$ and, for $y'(\phi,\nu) = \mathcal{P}_{\nu}^\perp {y} + \phi(\nu / \|\nu \|_2^2)$, 
\begin{equation}
\label{def_S2}
S^\lambda_{reg}(\nu_{reg}) = \{ \phi : R_A \in \tree^\lambda\{y'(\phi,\nu_{reg})\} \}.
\end{equation}
\end{theorem}
Theorem 2 and the resulting efficient computations in Section~\ref{section_computing} rely on the assumption that $Y \sim N_n(\mu, \sigma^2 I_n)$.

We can also define a confidence interval for $\nu_{reg}^\T \mu$ that attains nominal selective coverage. 
\begin{proposition}
 \label{prop_CIreg} 
For any $0 \leq \alpha \leq 1$ and any realization $y \in \mathbb{R}^n$, the values $L(y)$ and $U(y)$ that satisfy
\small
\begin{equation}
\label{eq_mainCIreg}
F \{ \nu_{reg}^\T y ; L(y), \sigma^2\|\nu_{reg}\|_2^2, S^\lambda_{reg}(\nu_{reg})\} = 1 - \alpha/2, \;\;\; F \{\nu_{reg}^\T y ; U(y), \sigma^2\|\nu_{reg}\|_2^2, S^\lambda_{reg}(\nu_{reg})\}  = \alpha/2,
\end{equation}
\normalsize
are unique, and $[L(Y), U(Y)]$ achieves $(1-\alpha)$-selective coverage for $\nu_{reg}^\T \mu$. 
\end{proposition}
In Section~\ref{section_computing}, we propose an approach to analytically characterize the set $S^\lambda_{reg}(\nu_{reg})$ in \eqref{def_S2}.

\subsection{Intuition for the Conditioning Sets $S^\lambda_{sib}(\nu_{sib})$ and  $S^\lambda_{reg}(\nu_{reg})$}
\label{subsec_intuition}

We first develop intuition for the set $S^\lambda_{sib}(\nu_{sib})$ defined in \eqref{def_S}. 
From Theorem~\ref{theorem_1jewell}, 
$$
\left\{ y'(\phi,\nu_{sib}) \right\}_i =
y_i  + 
(\phi - \nu_{sib}^\T y)\left\{  \frac{\sum_{i'=1}^n 1_{(x_{i'} \in R_B)}}{\sum_{i'=1}^n 1_{(x_{i'} \in R_A \cup R_B)}}{1}_{(x_i \in R_A)} - \frac{\sum_{i'=1}^n 1_{(x_{i'} \in R_A)}}{\sum_{i'=1}^n 1_{(x_{i'} \in R_A \cup R_B)}}{1}_{(x_i \in R_B)}\right\}.
$$
Thus, $y'(\phi,\nu_{sib})$ is a perturbation of $y$ that exaggerates the difference between  the observed sample mean responses of   $R_A$ and $R_B$  if $|\phi| > |\nu_{sib}^\T y|$, and shrinks that difference  if $|\phi | < |\nu_{sib}^\T y|$. The set $S^\lambda_{sib}(\nu_{sib})$ quantifies the amount that we can shift the difference in sample mean responses between $R_A$ and $R_B$ while still producing a tree containing these sibling regions. 
The top row of Figure~\ref{fig_intuition} displays  $\tree^0\{y'(\phi, \nu_{sib})\}$, as a function of $\phi$, in an example where $S^0_{sib}(\nu_{sib}) = (-19.8, -1.8) \cup (0.9, 34.9)$. 

\begin{figure}[!h]
    \centering
    \includegraphics[width=13cm, height=6.5cm]
   {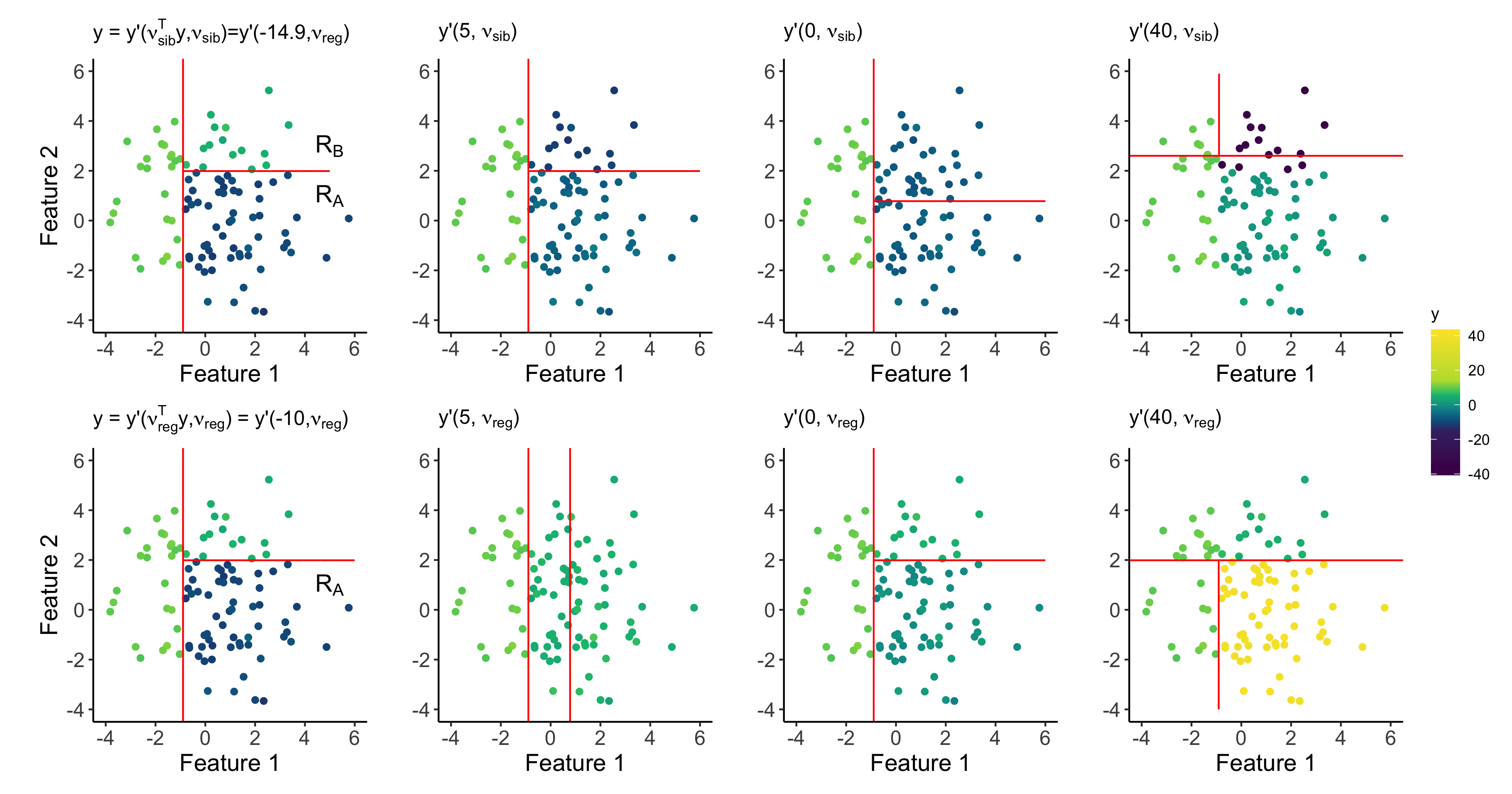}
    \caption{Data with $n=100$ and $p=2$. Regions resulting from CART ($\lambda=0$) are delineated using solid lines. Here, $R_A=\chi_{1, 26, 0} \cap \chi_{2, 72, 1}$ and $R_B = \chi_{1, 26, 0} \cap \chi_{2, 72, 0}$. \emph{Top:} Output of CART applied to $y'(\phi, \nu_{sib})$, where $\nu_{sib}$ in \eqref{eq_nusib} encodes the contrast between $R_A$ and $R_B$, for various values of $\phi$. The left-most panel displays $y=y'(\nu_{sib}^\T y, \nu_{sib})$. By inspection, we see that $-14.9 \in S_{sib}^0(\nu_{sib})$ and $5 \in S_{sib}^0(\nu_{sib})$, but $0 \not\in S_{sib}^0(\nu_{sib})$ and $40 \not\in S_{sib}^0(\nu_{sib})$. In fact, $S^0_{sib}(\nu_{sib}) = (-19.8, -1.8) \cup (0.9, 34.9)$. \emph{Bottom:} Output of CART applied to $y'(\phi, \nu_{reg})$, where $\nu_{reg}$ in \eqref{eq_nureg} encodes membership in $R_A$. The left-most panel displays $y=y'(\nu_{reg}^\T y, \nu_{reg})$.  Here, $S^0_{reg}(\nu_{reg}) = (-\infty, 3.1) \cup (5.8,8.8) \cup (14.1, \infty)$. }
 \label{fig_intuition}
\end{figure}

We next develop intuition for $S^\lambda_{reg}(\nu_{reg})$, defined in \eqref{def_S2}. Note that $\left\{ y'(\phi,\nu_{reg})\right\}_i = y_i  + (\phi - \nu_{reg}^\T y) {1}_{(x_i \in R_A)}$, where $y'(\phi, \nu_{reg})$ is defined in Theorem~\ref{theorem_1modified}. Thus, $y'(\phi, \nu_{reg})$ shifts the responses of the observations in $R_A$ so that their sample mean equals $\phi$, and leaves the others unchanged. The set  $S^\lambda_{reg}(\nu_{reg})$ quantifies the amount that we can exaggerate or shrink the sample mean response in region $R_A$ while still producing a tree that contains $R_A$. 
The bottom row of Figure~\ref{fig_intuition} displays   $y'(\phi, \nu_{reg})$ as $\phi$ is varied, in an example with 
 $S^0_{reg}(\nu_{reg}) = (-\infty, 3.1) \cup (5.8,8.8) \cup (14.1, \infty)$.

\section{Computing the conditioning sets $S_{sib}^\lambda(\nu_{sib})$ and $S_{reg}^\lambda(\nu_{reg})$}
\label{section_computing}

\subsection{Recharacterizing the conditioning sets in terms of branches}
\label{subsec_branches}

We begin by introducing the concept of a branch.  
\begin{definition}[Branch]
\label{def:branch}
A \emph{branch} is an ordered sequence of triples $\mathcal{B} = \left((j_1,s_1,e_1),\ldots,(j_L,s_L,e_L)\right)$ such that $j_l \in \{1,\ldots, p\}$, $s_l \in \{1,\ldots, n-1\}$, and $e_l \in \{0,1\}$ for $l=1,\ldots,L$. The branch $\mathcal{B}$ induces a nested set of regions $\mathcal{R}(\mathcal{B}) =\{R^{(0)}, R^{(1)}, \ldots, R^{(L)} \}$, where $R^{(l)} = \bigcap_{l'=1}^{l} \chi_{j_{l'},s_{l'},e_{l'}}$ for $l = 1,\ldots,L$, and $R^{(0)} = \mathbb{R}^p$.
\end{definition}

For a branch $\mathcal{B}$ and a vector $\mathcal{\nu}$, we define
\begin{equation}
 \label{eq:sbnu} 
S^\lambda(\mathcal{B}, \nu) = \left\{ \phi : \mathcal{R}(\mathcal{B}) \subseteq  \tree^\lambda\{y'(\phi,\nu)\}\right\}.
\end{equation}
For $R \in \tree$, we let $\branch(R, \tree)$ denote the branch such that 
$\mathcal{R}\{\branch(R, \tree)\}$ contains $R$ and all of its ancestors in $\tree$.  

\begin{lemma}
\label{lemma_orderedsplits}
Suppose that $R_A$ and $R_B$ are siblings in $\tree^\lambda(y)$. Then $R_A$ and $R_B$ are siblings in  $\tree^\lambda\{y'(\phi,\nu_{sib})\}$ if and only if $\mathcal{R}[\branch\{R_A, \tree^\lambda(y)\}] \subseteq \tree^\lambda\{y'(\phi,\nu_{sib})\}$. Therefore, $S_{sib}^\lambda(\nu_{sib})=S^\lambda[\branch\{R_A,\tree^\lambda(y)\}, {\nu_{sib}}]$,  defined in \eqref{def_S} and \eqref{eq:sbnu}. 
\end{lemma}

Lemma~\ref{lemma_orderedsplits} says that $\tree^\lambda\{y'(\phi, \nu_{sib})\}$ contains siblings $R_A$ and $R_B$ if and only if it contains the entire branch associated with $R_A$ in $\tree^\lambda(y)$. However, Lemma~\ref{lemma_orderedsplits} does not apply in the single region case: 
for $\nu_{reg}$ defined in \eqref{eq_nureg} and some $R_A \in \tree^\lambda(y)$, the fact that  $R_A \in \tree^\lambda\{y'(\phi, \nu_{reg})\}$ does not imply that $\mathcal{R}[\branch\{R_A, \tree^\lambda(y)\}] \subseteq \tree^\lambda\{y'(\phi, \nu_{reg})\}$. Instead, a result similar to Lemma~\ref{lemma_orderedsplits} holds, involving permutations of  the branch.  
\begin{definition}[Permutation of a branch]
\label{def:branch_perm}
Let $\Pi$ denote the set of all $L!$ 
permutations of
$(1,2,\ldots,L)$.  Given $\pi \in \Pi$ and a branch $\mathcal{B} = ((j_1,s_1,e_1), \ldots ,(j_L,s_L,e_L))$, we say that   $\pi\left( \mathcal{B}\right) = $$( (j_{\pi(1)},s_{\pi(1)},e_{\pi(1)}), \ldots ,$
$(j_{\pi(L)},s_{\pi(L)},e_{\pi(L)}) )$ is a \emph{permutation of the branch} $\mathcal{B}$.  
\end{definition}
Branch $\mathcal{B}$ and its permutation $\pi\left( \mathcal{B}\right)$ induce the same region $R^{(L)}$, but $\mathcal{R}\{\pi\left( \mathcal{B}\right)\} \neq \mathcal{R}(\mathcal{B})$.  

\begin{lemma}
\label{lemma_permutations}
Let $R_A \in \tree^{\lambda}(y)$. Then $R_A \in \tree^\lambda\{y'(\phi,\nu_{reg})\}$ if and only if there exists a $\pi \in \Pi$ such that $\mathcal{R}[\pi\left\{ \branch_{R_A}(y)\right\}] \subseteq \tree^\lambda\{y'(\phi, \nu_{reg})\}$. Thus, for $S^\lambda_{reg}(\nu_{reg})$ in \eqref{def_S2},
\begin{align}
\label{def_Sreg_permuted}
S^\lambda_{reg}(\nu_{reg}) 
&= \bigcup_{\pi \in \Pi} S^{\lambda}
\left(\pi\left[ \branch\{R_A,\tree^\lambda(y)\}\right] , \nu_{reg}\right).
\end{align}
\end{lemma}

Lemmas~\ref{lemma_orderedsplits} and \ref{lemma_permutations} reveal that computing  $S_{sib}^\lambda(\nu_{sib})$ and $S_{reg}^\lambda(\nu_{reg})$ requires characterizing sets of the form $S^\lambda(\mathcal{B},\nu)$, defined in \eqref{eq:sbnu}. 
To compute $S_{sib}^\lambda(\nu_{sib})$ we will only need to consider $S^\lambda(\mathcal{B},\nu)$ where $\mathcal{R}(\mathcal{B}) \subseteq \tree^\lambda(y)$. However, to compute $S_{reg}^\lambda(\nu_{reg})$, we will need to consider $S^\lambda\{\pi(\mathcal{B}),\nu\}$ where $\mathcal{R}(\mathcal{B}) \subseteq \tree^\lambda(y)$ but $\mathcal{R}\{\pi(\mathcal{B})\} \nsubseteq \tree^\lambda(y)$.

\subsection{Computing $S^\lambda(\mathcal{B},\nu)$ in \eqref{eq:sbnu}}
\label{subsec_growingcalculations}

Throughout this section, we consider a   vector $\nu \in \mathbb{R}^n$ and a  branch $\mathcal{B}=\left( (j_1,s_1,e_1), \ldots ,(j_L,s_L,e_L) \right)$, where $\mathcal{R}(\mathcal{B})$ may 
or may not be in $\tree^\lambda(y)$. 
Recall from Definition~\ref{def:branch} that $\mathcal{B}$ induces the nested regions $R^{(l)} = \bigcap_{l'=1}^l \chi_{j_{l'},s_{l'},e_{l'}}$ for $l=1,\ldots,L$, and $R^{(0)} = \mathbb{R}^p$. Throughout this section, our only requirement on $\mathcal{B}$ and $\nu$ is the following condition. 
\begin{condition}
\label{cond_nuprop}
 For $y'(\phi,\nu)$ defined in Theorem~\ref{theorem_1jewell},  
 $\mathcal{B}$ and $\nu$ satisfy $\left\{y'(\phi, \nu)\right\}_i = y_i + c_1 {1}_{\left\{ x_i \in R^{(L)} \right\} } + c_2 {1}_{\left[ x_i \in \left\{ R^{(L-1)} \cap \chi_{j_L,s_L,1-e_L}\right\} \right]}$ for $i=1,\ldots,n$ and for some constants $c_1$ and $c_2$. 
\end{condition}

To characterize  $S^\lambda(\mathcal{B},\nu)$ in \eqref{eq:sbnu}, recall that the CART algorithm  in Section~\ref{subsec_CARTreview} involves growing a very large tree $\tree^0(y)$, and then pruning it. We first characterize the set
\begin{align}
\label{eq:grow}
S_{grow}(\mathcal{B},\nu) &= \{ \phi : \mathcal{R}(\mathcal{B}) \subseteq \tree^0\{y'(\phi,\nu)\}\}.
\end{align}

\begin{proposition}
\label{prop_Sisintersection} 
Recall the definition of  $\gain_{R^{(l)}}\{y'(\phi,\nu), j,s\}$  in \eqref{eq:gain}, and let 
$S_{l,j,s} =  \{\phi :  \gain_{R^{(l-1)}}\{y'(\phi,\nu), j,s\} \leq \gain_{R^{(l-1)}}\{y'(\phi,\nu), j_l,s_l\}\}.
$
Then, $
S_{grow}(\mathcal{B}, \nu) = \bigcap_{l=1}^L \bigcap_{j=1}^p \bigcap_{s=1}^{n-1} S_{l,j,s}.
$
\end{proposition} 

Proposition~\ref{prop_quadratic} says that we can compute $S_{grow}(\mathcal{B}, \nu)$ efficiently. 
\begin{proposition}
\label{prop_quadratic}
The set $S_{l,j,s}$ is defined by a quadratic inequality in $\phi$. Furthermore, we can evaluate all of the sets $S_{l,j,s}$, for $l=1,\ldots,L$, $j=1,\ldots,p$, $s=1,\ldots,n-1$, in $O\left\{npL+np \log(n)\right\}$ operations. Intersecting these sets to obtain $S_{grow}(\mathcal{B}, \nu)$ requires at most $O\left\{npL \times \log(npL)\right\}$ operations, and only $O(npL)$ operations if $\mathcal{B} = \branch\{R_A, \tree^\lambda(y)\}$ and $\nu$ is of the form $\nu_{sib}$ in \eqref{eq_nusib}.
\end{proposition}

Noting that
$S^\lambda(\mathcal{B},\nu) = \left\{ \phi \in S_{grow}(\mathcal{B},\nu) : R^{(L)} \in \tree^\lambda\{y'(\phi,\nu)\} \right\}$,  it remains to characterize the set of $\phi \in  S_{grow}(\mathcal{B},\nu)$ such that $R^{(L)}$ is not removed during pruning.
Recall that $g(\cdot)$ was defined in \eqref{def_g}. 
\begin{proposition}
\label{prop_s.pruning}
There exists a  tree $\tree(\mathcal{B},\nu,\lambda)$ such that
\begin{equation}
\label{eq:mainpruning}
S^\lambda(\mathcal{B},\nu) = S_{grow}(\mathcal{B},\nu) \cap \left( \bigcap_{l=0}^{L-1} \left\{ \phi : g\left\{
R^{(l)}, \tree(\mathcal{B},\nu,\lambda), y'(\phi,\nu)
\right\} \geq \lambda \right\} \right).
\end{equation} 
If $\mathcal{R}(\mathcal{B}) \in \tree^\lambda(y)$, then $\tree(\mathcal{B},\nu,\lambda)=\tree^\lambda(y)$ satisfies \eqref{eq:mainpruning}. Otherwise, given the set $S_{grow}(\mathcal{B},\nu)$, computing a $\tree(\mathcal{B},\nu,\lambda)$ that satisfies \eqref{eq:mainpruning} has a worst-case computational cost of $O(n^2p)$. 
\end{proposition}
We explain how to compute a $\tree(\mathcal{B},\nu,\lambda)$ satisfying \eqref{eq:mainpruning}  when $\mathcal{R}(\mathcal{B}) \not\in \tree^\lambda(y)$ in the supplementary materials. 

\begin{proposition}
\label{prop_pruning.efficient}
The set $\bigcap_{l=0}^{L-1} \left\{ \phi  :  g\left\{
R^{(l)},\tree(\mathcal{B},\nu,\lambda), y'(\phi,\nu)
\right\} \geq \lambda \right\}$ in \eqref{eq:mainpruning} is the intersection of the solution sets of $L$ quadratic inequalities in $\phi$. Given $\tree(\mathcal{B},\nu,\lambda)$, the coefficients of these quadratics can be obtained in $O(nL)$ operations.  After  $S_{grow}(\mathcal{B},\nu)$ has been computed, intersecting it with these quadratic sets to obtain  $S^{\lambda}(\mathcal{B},\nu)$ from \eqref{eq:mainpruning} requires $O\{npL \times \log(npL)\}$ operations in general, and only $O(L)$ operations if $\mathcal{B} = \branch\{ R_A, \tree^\lambda(y) \}$ and $\nu=\nu_{sib}$ from \eqref{eq_nusib}. 
\end{proposition}

The results in this section have relied upon Condition~\ref{cond_nuprop}. Indeed, this condition holds for branches $\mathcal{B}$
and vectors $\nu$ that arise in characterizing the sets $S^\lambda_{sib}(\nu_{sib})$ and $S^\lambda_{reg}(\nu_{reg})$. 
\begin{proposition}
\label{prop_nuprop}
If either
(i) $\mathcal{B} =  \branch\{R_A, \tree^\lambda(y)\}$ and $\nu=\nu_{sib}$ \eqref{eq_nusib}, where $R_A$ and $R_B$ are siblings in $\tree^\lambda(y)$, or (ii) $\mathcal{B}$ is a permutation of $\branch\{R_A, \tree^\lambda(y)\}$ and $\nu=\nu_{reg}$ \eqref{eq_nureg}, where  $R_A \in \tree^\lambda(y)$, then Condition~\ref{cond_nuprop} holds.
\end{proposition}

Combining Lemma~\ref{lemma_orderedsplits}  with Propositions~\ref{prop_Sisintersection}--\ref{prop_nuprop}, we see that  $S_{sib}^\lambda(\nu_{sib})$ can be computed in $O\{npL+np\log(n)\}$ operations.
 However, computing $S_{reg}^\lambda(\nu_{reg})$ is much more computationally intensive: by Lemma~\ref{lemma_permutations}  and Propositions~\ref{prop_Sisintersection}--\ref{prop_nuprop}, it  requires computing \\ $S^\lambda(\pi\left[\branch\{R_A,\tree^\lambda(y)\}\right],\nu_{reg})$  for all $L!$ permutations $\pi \in \Pi$, for a total of $O\left[ L! \left\{ n^2 p L log(p L)\right\} \right]$ operations. In Section~\ref{subsec_actuallycomputeSreg}, we discuss ways to avoid these calculations.

\subsection{A Computationally-Efficient Alternative to $S^\lambda_{reg}(\nu_{reg})$} 
\label{subsec_actuallycomputeSreg}

Lemma~\ref{lemma_permutations} suggests that carrying out
inference on a single region requires computing \\$S^\lambda(\pi\left[\branch\{R_A,\tree^\lambda(y)\}\right],\nu_{reg})$ for every $\pi \in \Pi$. We now present a less computationally demanding alternative.

\begin{proposition}
\label{prop:subsetpermutation}
Let $Q$ be a subset of the $L!$ permutations in $\Pi$, i.e. $Q \subseteq \Pi$. 
Define 
\footnotesize
\begin{equation}
\nonumber
p_{reg}^{Q}(y) = pr_{H_0}\left\{ |\nu_{reg}^\T Y - c | \geq |\nu_{reg}^\T y - c| \mid \bigcup_{\pi \in Q} \left( \mathcal{R}(\pi[\branch\{R_A, \tree^\lambda(y)\}]) \subseteq \tree^{\lambda}(Y) \right), \mathcal{P}_{\nu_{reg}}^\perp Y = \mathcal{P}_{\nu_{reg}}^\perp {y} \right\}.
\end{equation}
\normalsize
The test based on $p_{reg}^{Q}(y)$ controls the selective Type 1 error rate \eqref{eq_st1e} for
$H_0: \nu_{reg}^\T \mu = c$. Furthermore, \small 
$p_{reg}^{Q}(y) = pr\left\{ |\phi- c | \geq |\nu_{reg}^\T y - c| \mid \phi \in \bigcup_{\pi \in Q} S^\lambda\left( \pi[\branch\{R_A, \tree^\lambda(y)\}], \nu_{reg}\right) \right\}$,
\normalsize 
where $\phi \sim N(c, \|\nu_{reg}\|_2^2 \sigma^2)$. 
\end{proposition}

Using the notation in Proposition~\ref{prop:subsetpermutation}, $p_{reg}(y)$ introduced in \eqref{eq:pvalreg} equals $p_{reg}^\Pi(y)$.   
If we take $Q =\{\mathcal{I}\}$, where $\mathcal{I}$ is the identity permutation, then 
we arrive at 
\begin{equation}
	p_{reg}^\mathcal{I}(y) =  P\left(\left|\phi - c\right| \geq |\nu_{reg}^\T {y} - c| \mid \phi \in S^\lambda[\branch\{R_A, \tree^\lambda(y)\}, \nu_{reg}] \right),
\label{eq_identitypvalue}
\end{equation}
 where $\phi \sim N(c, \|\nu_{reg}\|_2^2 \sigma^2)$. 
  The set $S^\lambda[\branch\{R_A, \tree^\lambda(y)\}, \nu_{reg}]$ can be easily computed by Proposition~\ref{prop_s.pruning}. 
  
Compared to~\eqref{eq:pvalreg}, \eqref{eq_identitypvalue} conditions on an extra piece of information: the ancestors of $R_A$. Thus, while \eqref{eq_identitypvalue} controls the selective Type 1 error rate, it may have lower power than \eqref{eq:pvalreg} \citep{fithian2014optimal}. Similarly, inverting \eqref{eq_identitypvalue} to form a confidence interval provides correct selective coverage, but may yield intervals that are wider than those in Proposition~\ref{prop_CIreg}. Proposition~\ref{prop:subsetpermutation} is motivated by a proposal by \cite{lee2016exact} to condition on both the selected model (necessary information) and the signs of the selected variables (extra information) in the lasso setting, to gain computational efficiency at the possible expense of precision and power.

In Appendix~\ref{appendix:permutation_sims}, we show through simulation that the loss in power associated with using \eqref{eq_identitypvalue} rather than \eqref{eq:pvalreg} is negligible. Thus, in practice, we suggest using \eqref{eq_identitypvalue} for its computational efficiency. We use \eqref{eq_identitypvalue} for the remainder of this paper. 

Furthermore, we can consider computing confidence intervals of the form $[L_{S_{reg}^\mathcal{I}}(y), U_{S_{reg}^\mathcal{I}}(y)]$ rather than \eqref{eq_mainCIreg}, where $L_{S_{reg}^\mathcal{I}}(y)$ and $U_{S_{reg}^\mathcal{I}}(y)$  satisfy
\begin{align}
\label{eq:branchCI}   
 & F \left(\nu_{reg}^\T y ; L_{S_{reg}^\mathcal{I}}(y), \sigma^2\|\nu_{reg}\|_2^2, S^\lambda\left[\branch\{R_A, \tree^\lambda(y)\}, \nu_{reg}\right] \right) = 1 - \frac{\alpha}{2},  \nonumber \\ 
 & F \left(\nu_{reg}^\T y ; U_{S_{reg}^\mathcal{I}}(y), \sigma^2\|\nu_{reg}\|_2^2, S^\lambda\left[\branch\{R_A, \tree^\lambda(y)\}, \nu_{reg}\right] \right)  = \frac{\alpha}{2}. 
\end{align}
In Appendix~\ref{appendix:permutation_sims}, we show that the confidence intervals resulting from \eqref{eq:branchCI} are not much wider than those resulting from \eqref{eq_mainCIreg}. We therefore make use of confidence intervals of the form \eqref{eq:branchCI} in the remainder of this paper. 

\section{Simulation Study}
\label{section_simstudy}

\subsection{Data Generating Mechanism}
\label{subsubsec_datagen}

We simulate ${X} \in \mathbb{R}^{n \times p}$ with $n=200, p=10,$ $X_{ij} \overset{i.i.d.}{\sim} N(0,1)$, and $y \sim N_n(\mu, \sigma^2 I_n)$ with $\sigma=5$ and
$\mu_i = b \times \left[ {1}_{\left(x_{i,1} \leq 0\right)} \times \{ 1 + a {1}_{\left(x_{i,2} > 0\right)}+ {1}_{\left( x_{i,3}  \times x_{i,2}  > 0 \right)} \} \right]$. This $\mu$ vector defines a three-level tree, shown in Figure~\ref{fig:truetree} for three values of $a \in \mathbb{R}$. 
\begin{figure}
\includegraphics[width=\textwidth]{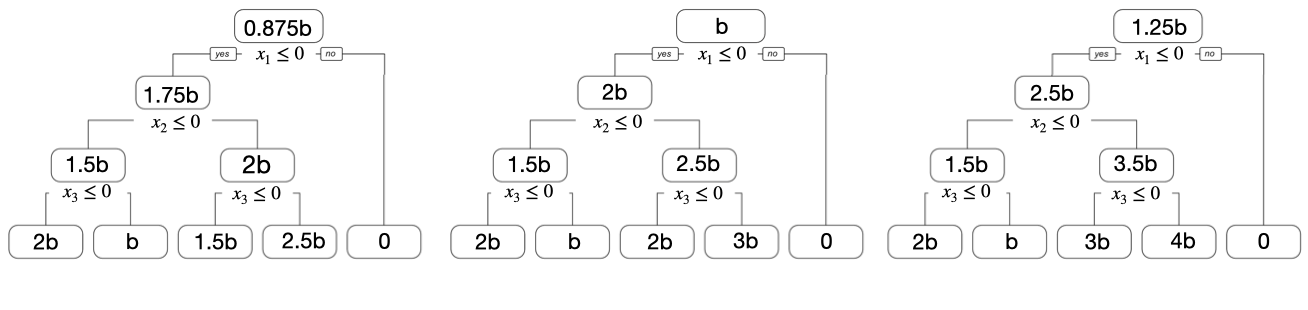}
\caption{The true mean model in Section~\ref{section_simstudy}, for $a=0.5$ (left), $a=1$ (center), and $a=2$ (right). The difference in means between the sibling nodes at level two in the tree is $ab$, while the difference in means between the sibling nodes at level three is $b$.}
\label{fig:truetree}
\end{figure}

\subsection{Methods for Comparison}
\label{subsubsec_methods}

All CART trees are fit using the \texttt{R} package \texttt{rpart} \citep{therneau2015package} with $\lambda=200$, a maximum level of three, and a minimum node size of one. We compare three approaches for conducting inference. 
(i) \emph{Selective $Z$-methods:} Fit a CART tree to the data. For each split, test for a difference in means between the two sibling regions using \eqref{def_mainpval}, and compute the corresponding confidence interval in \eqref{eq_mainCI}. Compute the confidence interval for the mean of each region using \eqref{eq:branchCI}. 
(ii) \emph{Naive $Z$-methods:} Fit a CART tree to  the data. For each split, conduct a naive $Z$-test for the difference in means between the two sibling regions, and compute the corresponding naive $Z$-interval. Compute a naive $Z$-interval for each region's mean.
(iii) \emph{Sample splitting:} Split the data into equally-sized training and test sets. Fit a CART tree to the training set. On the test set, conduct a naive $Z$-test for each split and compute a naive $Z$-interval for each split and each region. If a region has no test set observations, then we fail to reject the null hypothesis and fail to cover the parameter. 

The conditional inference tree (CTree) framework of \citet{hothorn2006unbiased} uses a different criterion than CART to perform binary splits. Within a region, it  tests for linear association between each covariate and the response. The covariate with the smallest p-value for this linear association is selected as the split variable, and a Bonferroni corrected p-value that accounts for the number of covariates is reported in the final tree. Then, the split point is selected. If, after accounting for multiple testing, no variable has a p-value below a pre-specified significance level $\alpha$, then the recursion stops. While CTree's p-values assess linear association and thus are not directly comparable to the p-values in (i)--(iii) above, it is the most popular framework currently available for determining if a regression tree split is statistically significant. Thus, we also evaluate the performance of (iv) \emph{CTree:} Fit a CTree to all of the data using the \texttt{R} package \texttt{partykit} \citep{hothorn2015partykit} with $\alpha=0.05$. For each split, record the p-value reported by \texttt{partykit}. 

In Sections~\ref{subsubsec_Type1}--\ref{subsubsec_width}, we assume that $\sigma$ is known. We consider the case of unknown $\sigma$ in Section~\ref{subsec:unknownvar}.

\subsection{Uniform p-values under a Global Null}
\label{subsubsec_Type1}

We generate $5,000$ datasets
with $a=b=0$, so that $H_0: \nu_{sib}^\T \mu = 0$ holds for all splits in all trees. Figure~\ref{fig_mainType1} displays the distributions of p-values across all splits in all fitted trees for the naive $Z$-test, sample splitting, and the selective $Z$-test. The selective $Z$-test and sample splitting  achieve uniform p-values under the null, 
while the naive $Z$-test (which does not account for the fact that $\nu_{sib}$ was obtained by applying CART to the same data used for testing) does not. CTree is omitted from the comparison: it creates a split only if the p-value is less than $\alpha=0.05$, and thus its p-values over the splits do not follow a Uniform(0,1) distribution.

\begin{figure}[!h]
\centering
\includegraphics[width=0.65\textwidth]{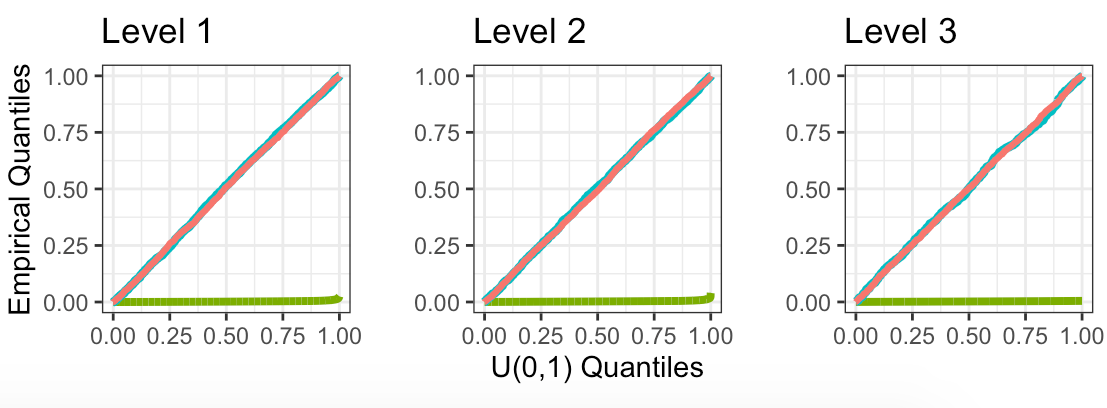}
\caption{Quantile-quantile plots of the p-values for testing $H_0: \nu_{sib}^\T \mu = 0$, as described in Section~\ref{subsubsec_Type1}. A naive Z-test (green), sample splitting (blue), and selective Z-test (pink) were performed; see Section \ref{subsubsec_methods}. The p-values are stratified by the level of the regions in the fitted tree.
}
\label{fig_mainType1}	
\end{figure}

\subsection{Power}
\label{subsubsec_Power}

We generate 500 datasets for each $(a,b) \in \{0.5,1,2\} \times  \{1,\ldots,10\}$, and evaluate the power of selective $Z$-tests, sample splitting, and CTree to reject the null hypothesis $H_0: \nu_{sib}^\T \mu = 0$. As naive $Z$-tests do not control the Type 1 error rate (Figure~\ref{fig_mainType1}), we do not evaluate their power. We consider two aspects of power: the probability that we \emph{detect} a true split, and the probability that we \emph{reject} the null hypothesis corresponding to a true split. 

\begin{table}
\def~{\hphantom{0}}
\tbl{
A $3 \times 3$ contingency table indicating an observation's involvement in a given true split and estimated split. The adjusted Rand index is computed using only the shaded cells }{
\begin{tabular}{clccc}
\multirow{2}{*}{} & & \multicolumn{3}{c}{Estimated Split} \\
&  & In left region & In right region & In neither \\[5pt]
\multirow{3}{*}{True Split} & In left region &  \cellcolor{gray!50}$t_1$ &  \cellcolor{gray!50}$t_2$ &  \cellcolor{gray!50}$t_3$ \\
  & In right region &  \cellcolor{gray!50}$u_1$ &  \cellcolor{gray!50}$u_2$ &  \cellcolor{gray!50}$u_3$ \\
    & In neither & $v_1$ & $v_2$ & $v_3$
\end{tabular}}
\label{table:contingencies}
\end{table}

Given a true split in Figure~\ref{fig:truetree} and an estimated split, we construct the $3 \times 3$ contingency table in Table~\ref{table:contingencies}, which indicates whether an observation is on the left-hand side, right-hand side, or not involved in the true split (rows) and the estimated split (columns). To quantify the agreement between the true and estimated splits, we compute the adjusted Rand index \citep{hubert1985comparing} associated with the $2 \times 3$ contingency table corresponding to the shaded region in Table~\ref{table:contingencies}. For each true split, we identify the estimated split for which the adjusted Rand index is largest; if this index exceeds $0.75$ then  this true split is ``detected".  Given that a true split is detected, the associated null hypothesis is rejected if the corresponding $p$-value is below $0.05$. Figure~\ref{fig_mainPowerres}  displays the proportion of true splits that are detected and rejected by each method.

As sample splitting fits a tree using only half of the data, it detects fewer true splits, and thus rejects the null hypothesis for fewer true  splits, than the selective $Z$-test. 

When $a$ is small, the difference in means between sibling regions at level two is small. Because  CTree makes a split only if there is strong evidence of association at that  level, it tends to build one-level trees, and thus fails to detect many true splits; by contrast, the selective Z-test (based on CART) successfully builds more three-level trees. Thus, when $a$ is small, the selective Z-test detects (and rejects) more true differences than CTree between regions at levels two and three.

\begin{figure}[!h]
\centering
\includegraphics[width=0.77\textwidth]{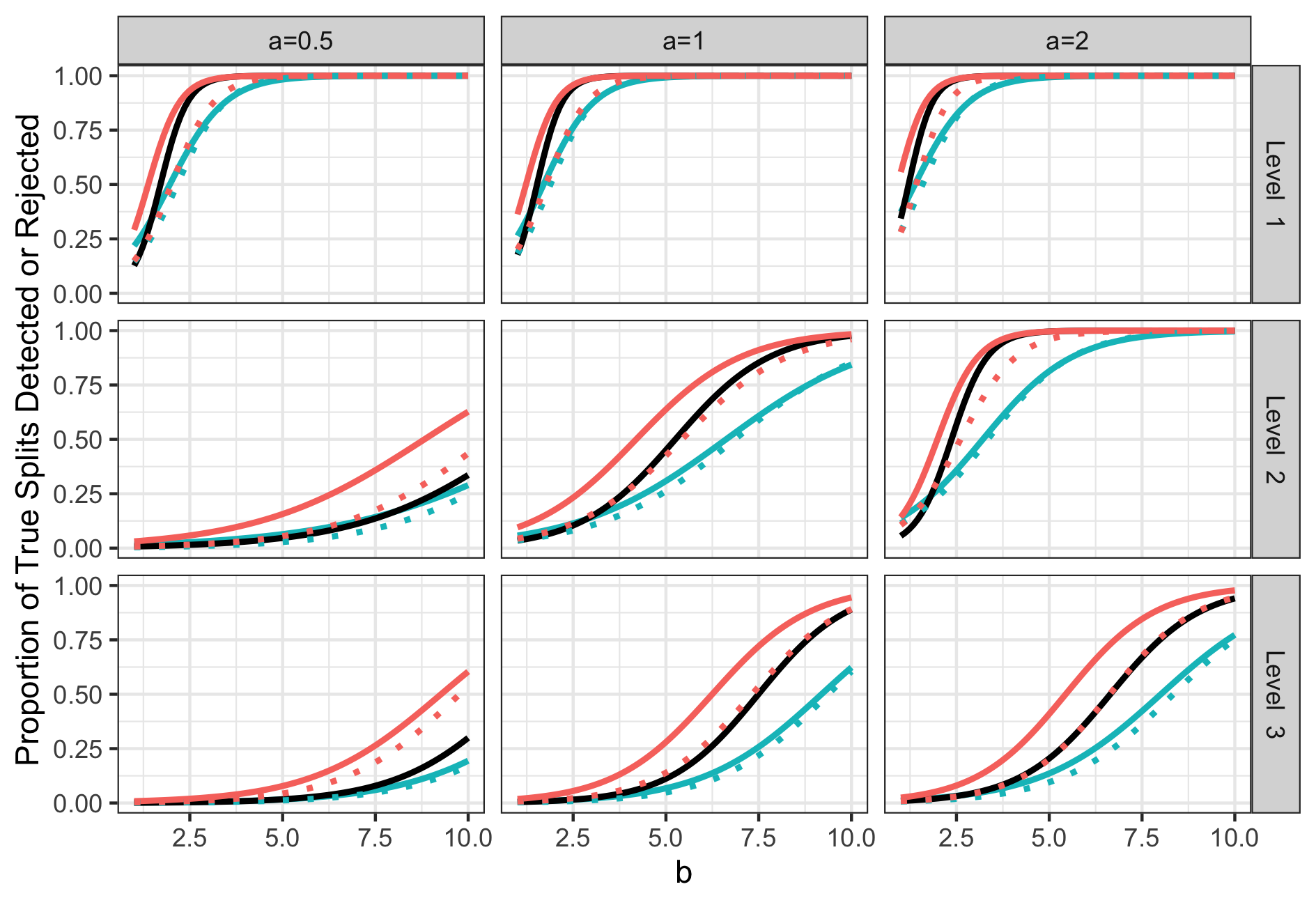}	
\caption{Proportion of true splits detected (solid lines) and rejected (dotted lines) for CART with selective $Z$-tests (pink), CTree (black), and CART with sample splitting (blue) across different settings of the data generating mechanism, stratified by level in tree. As CTree only makes a split if the p-value is less than 0.05, the proportion of detections equals the proportion of rejections. }

\label{fig_mainPowerres}
\end{figure}

\subsection{Coverage of Confidence Intervals for $\nu_{sib}^\T \mu$ and $\nu_{reg}^\T \mu$}
\label{subsubsec_CIcompare}

We generate 500 datasets for each $(a,b) \in \{0.5,1,2\} \times  \{0,\ldots,10\}$ to evaluate the 
coverage of 95\% confidence intervals constructed using naive $Z$-methods, selective $Z$-methods, and sample splitting. CTree is omitted from these comparisons because it does not provide confidence intervals. We say that the interval covers the truth if it contains $\nu^\T \mu$, where  $\nu$ is defined as in \eqref{eq_nusib} (for a particular split) or \eqref{eq_nureg} (for a particular region). Table~\ref{table_maincoverageresults} shows the proportion of each type of interval that covers the truth, aggregated across values of $a$ and $b$. The selective $Z$-intervals attain correct coverage of 95\%, while the naive $Z$-intervals do not.

  It may come as a surprise that sample splitting does not attain correct coverage. Recall that $\nu$ from \eqref{eq_nusib} or \eqref{eq_nureg} is an $n$-vector that contains entries for all observations in both the training set and the test set.  Thus,  $\nu^\T \mu$ involves the true mean among both training and test set observations in a given region or pair of regions. 
   By contrast, sample splitting attains correct coverage for a different parameter involving the true means of only the test observations that fall within a given region or pair of regions.

\subsection{Width of Confidence Intervals}
\label{subsubsec_width}

Figure~\ref{fig:widththings}(a) illustrates that our selective $Z$-intervals for $\nu_{reg}^\T \mu$ can be extremely wide when $b$ is small, particularly for regions located at deeper levels in the tree. For each tree that we build and for levels 1, 2, and 3, we compute the adjusted Rand Index \citep{hubert1985comparing} between the true tree (truncated at the appropriate level) and the estimated tree (truncated at the same level). Figure~\ref{fig:widththings}(b) shows that our selective confidence intervals can be extremely wide when this adjusted Rand Index is small, particularly at deeper levels of the tree. 

When $b$ is small and the adjusted Rand Index is small, the trees built by CART tend to be unstable, in the sense that small perturbations to the data affect the fitted tree. In this setting, the sample statistics $\nu_{reg}^\T y$ fall very close to the boundary of the truncation set. See \cite{kivaranovic2020length} for a discussion of why wide confidence intervals can arise in these settings. The great width of our confidence intervals reflects the uncertainty about the mean response within each region due to the instability of the tree-fitting procedure.

\begin{table}
\def~{\hphantom{0}}
\tbl{
Proportion of 95\% confidence intervals containing the true parameter, aggregated over all trees fit to the 5,500 datasets generated with $(a,b) \in \{0.5,1,2\} \times \{1,\ldots,10\}$
}{
\begin{tabular}{ccccccccc}
&& \multicolumn{3}{c}{Parameter $\nu^\T_{reg}{\mu}$} && \multicolumn{3}{c}{Parameter $\nu^\T_{sib}{\mu}$}  \\[5pt]
Level  && Selective $Z$ & Naive $Z$& Sample Splitting && Selective $Z$ & Naive $Z$ & Sample Splitting \\[5pt]
  1 && 0.951  & 0.889 & 0.918 && 0.948  & 0.834 & 0.915 \\
 2 && 0.950  & 0.645 &  0.921 && 0.951  & 0.410 & 0.917 \\
 3 && 0.951 & 0.711 & 0.921  && 0.950  & 0.550 & 0.921 
\end{tabular}}
\label{table_maincoverageresults}
\end{table}

\begin{figure}[!h]
(a) \hspace{290pt} (b) \\\includegraphics[width=10.5cm, height=3.7cm]{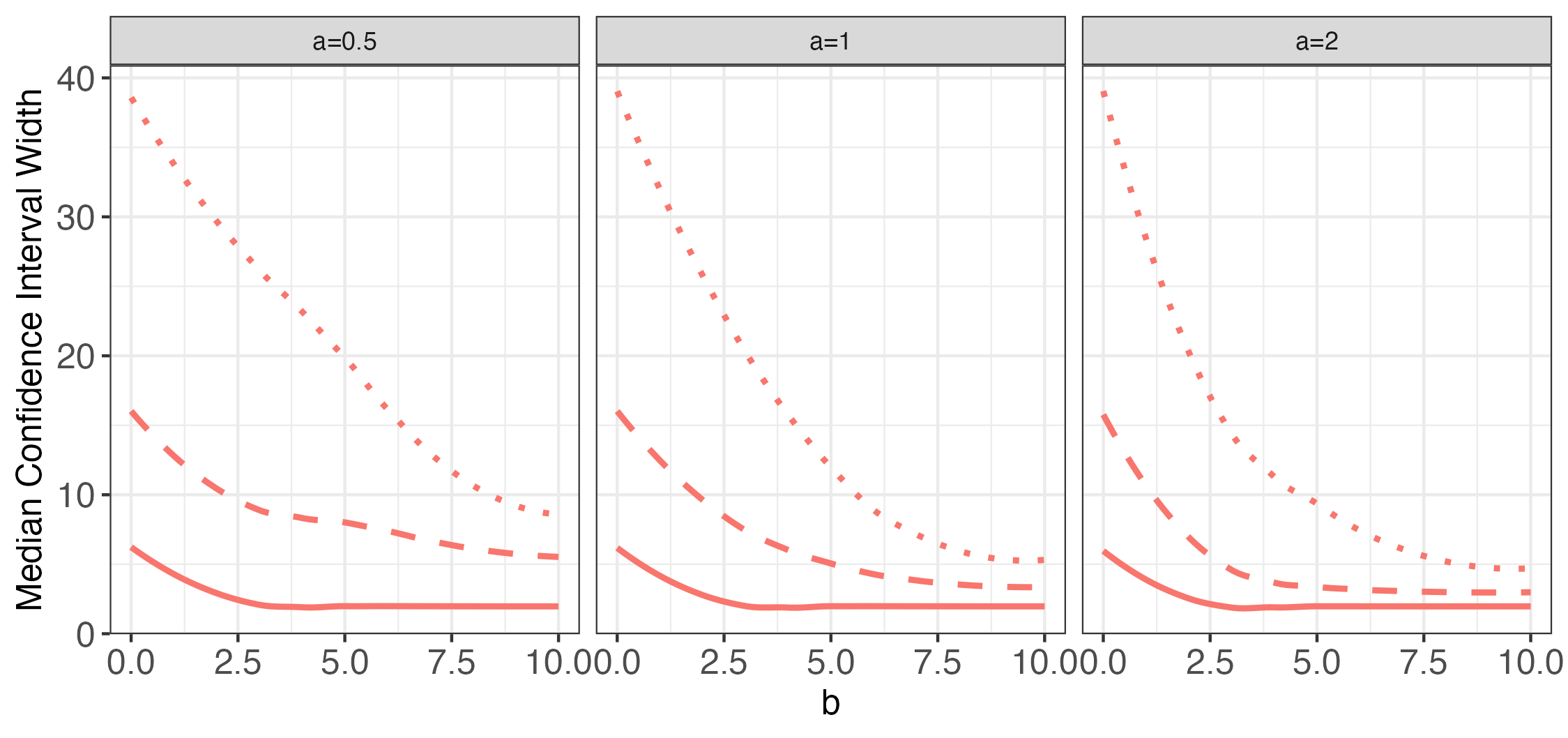}
\includegraphics[width=4cm, height=3.5cm]{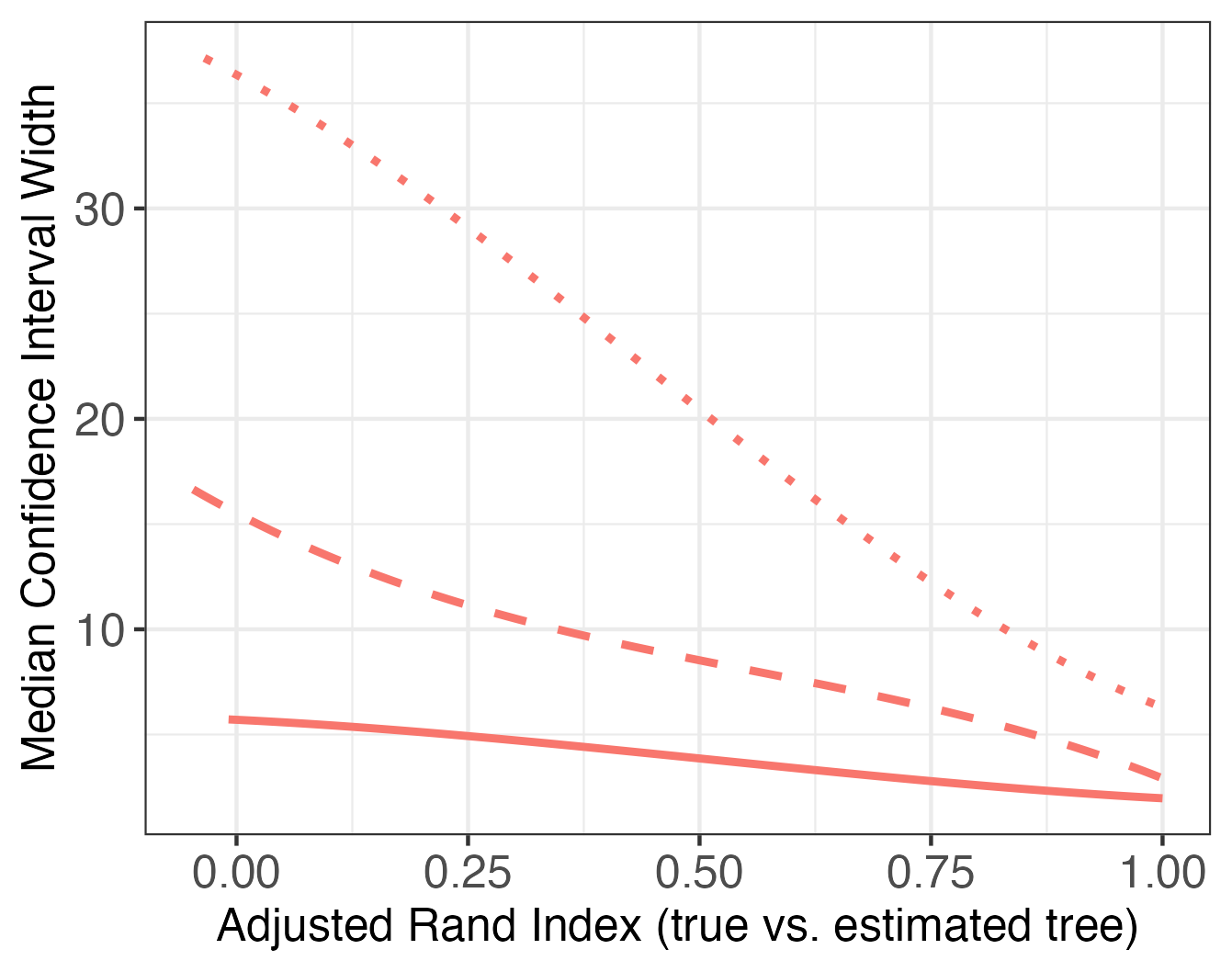}
\caption{The median width of the selective $Z$-intervals for parameter $\nu_{reg}^\T \mu$ for regions at levels one (solid), two (dashed), and three (dotted) of the tree. Similar results hold for parameter $\nu_{sib}^\T \mu$. Panel (a) breaks results down by the parameters $a$ and $b$, whereas panel (b) aggregates results across values of parameters $a$ and $b$, and displays them as a function of the adjusted Rand Index between the true and estimated trees. 
}
\label{fig:widththings}	
\end{figure}

\subsection{Results with Unknown $\sigma$}
\label{subsec:unknownvar}

Thus far, we have assumed that $\sigma$ is known. In this section, we compare the following three versions of the selective $Z$-methods that plug different values of $\sigma$ into the truncated normal CDF when computing p-values and confidence intervals:
\begin{enumerate}
\item $\sigma$: We plug in the true value of $\sigma$, as in Sections~\ref{subsubsec_Type1}--\ref{subsubsec_width}. 
\item $\hat{\sigma}_{\text{cons}}$: We plug in $\hat{\sigma}_{\text{cons}} = \sqrt{(n-1)^{-1} \sum \limits_{i=1}^n (y_i - \bar{y})^2}$, where $\bar{y} = n^{-1} \sum_{i=1}^n y_i$. 
\item $\hat{\sigma}_{\text{SSE}}:$ Let $\mathcal{T} = \left| \term\left( \mathbb{R}^p, \tree^\lambda(y) \right) \right|$ be the number of terminal regions in $\tree^\lambda(y)$. We plug in $\hat{\sigma}_{\text{SSE}} =\sqrt{\left(n -  \mathcal{T}\right)^{-1}\sum_{i=1}^n (y_i - \hat{y}_i)^2}$, where $\hat{y}_i$ is the predicted value for the $i$th observation given by $\textsc{Tree}^\lambda(y)$. 
\end{enumerate}
It is straightforward to show that $E[\hat{\sigma}_{\text{cons}}^2] \geq \sigma^2$, for any value of $E[y] = \mu$. Thus, %as explained in \cite{gao2020selective}, 
we expect this estimate to lead to conservative inference. On the other hand, $\hat{\sigma}_{\text{SSE}}^2$ can be made arbitrarily small by making the fitted tree arbitrarily deep, and so we expect inference based on this estimate to be anti-conservative if the fitted CART tree is large. 

Figure~\ref{fig_type_1_est} shows the distribution of p-values from testing $H_0: \nu_{sib}^T \mu = 0$ with the three versions of the selective $Z$-test under the data generating mechanism described in Section 5.1, with $a = b = 0$. In this setting, $\nu_{sib}^T \mu = 0$ holds for all splits in all trees. We see almost no difference between the three versions of the selective $Z$-test. 
In this global null setting, $E[\hat{\sigma}_{\text{cons}}^2] = \sigma^2$. Furthermore, the empirical bias of $\hat{\sigma}_{\text{SSE}}^2$ is small because the trees we grow are not particularly large; as in the rest of Section~\ref{section_simstudy}, we build trees to a maximum depth of 3 and prune with $\lambda=200$. 

\begin{figure}
\centering 
\includegraphics[width=0.8\textwidth]{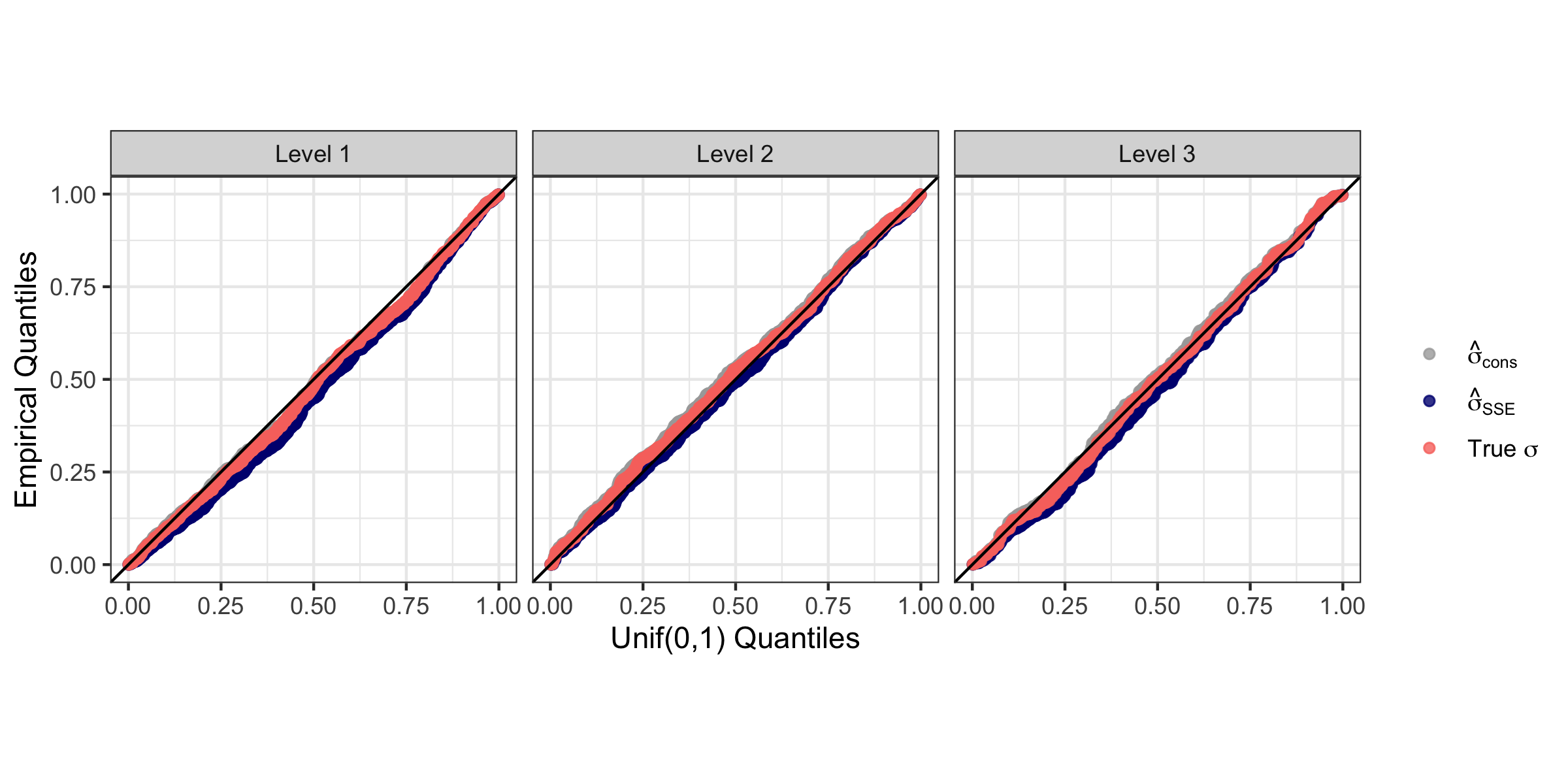}
\caption{
QQ plots of the p-values from testing $H_0: \nu_{sib}^T \mu = 0$ when $\mu = 0_n$ using the selective $Z$-test with three different values plugged in to the truncated normal CDF for $\sigma$. The p-values are stratified by the level of the regions in the fitted tree.
}
\label{fig_type_1_est}
\end{figure}

Figure~\ref{fig_power} displays the proportion of true splits detected and the proportion of true splits detected and rejected, as defined in Section~\ref{subsubsec_Power}, for the three versions of the selective $Z$-test when data is generated as in Section~\ref{subsubsec_Power}. For simplicity, we only show the setting where $a=1$. All three methods detect the same proportion of true splits, because they all perform inference on the same CART trees. The proportion of splits detected and rejected is very similar for $\sigma$ and $\hat{\sigma}_{\text{SSE}}$ because $\hat{\sigma}_{\text{SSE}}$ is a very good estimator for $\sigma$ in this setting. While $\hat{\sigma}_{\text{cons}}$ performs reasonably when $b$ is small, it severely overestimates $\sigma$ and thus has low power when $b$ is large.

\begin{figure}[H]
\centering 
\includegraphics[width=0.8\textwidth]{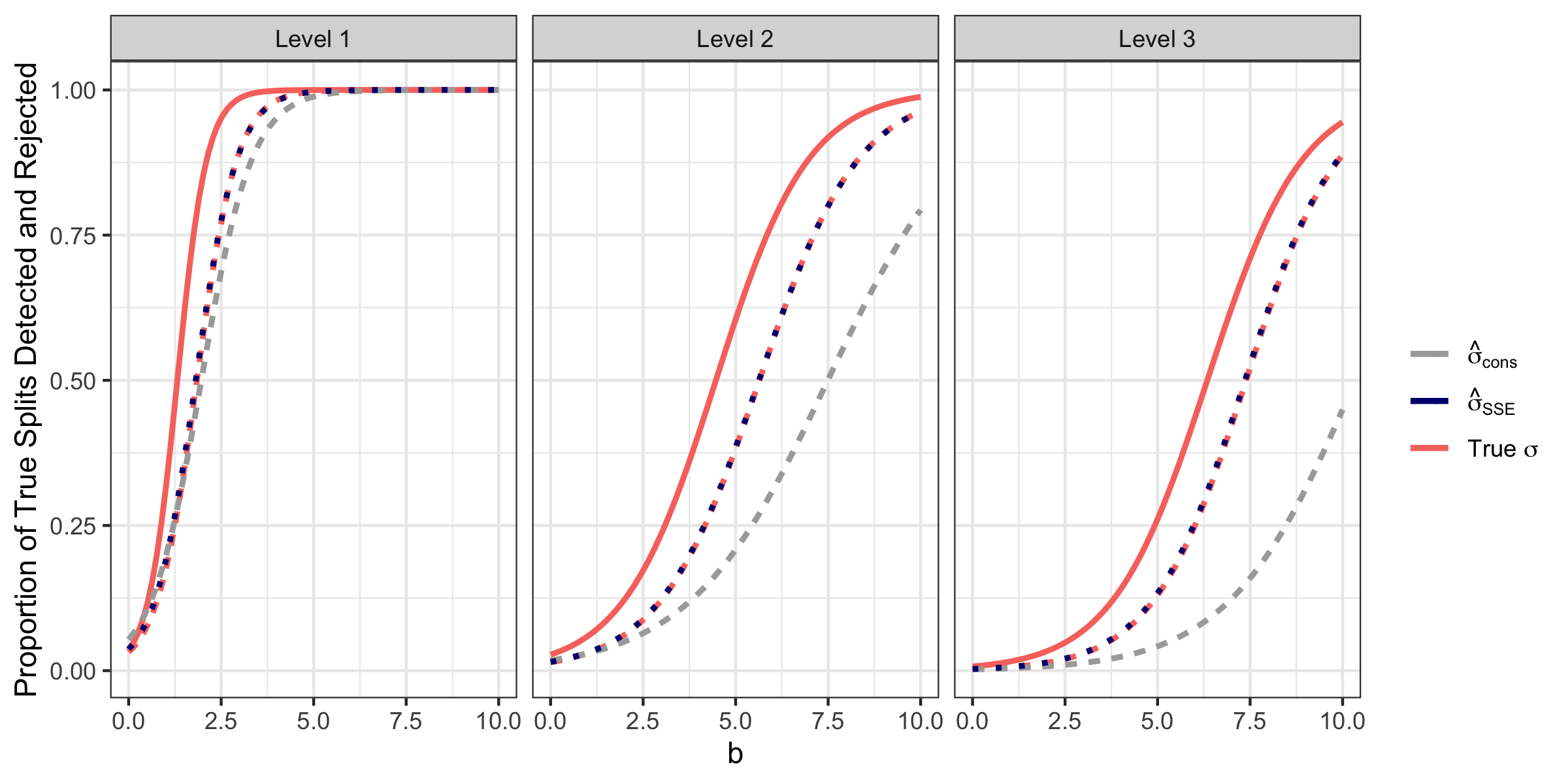}
\caption{Proportion of true splits detected (solid lines) and rejected (dotted lines) for CART with the three versions of the selective Z-test. The results are stratified by level in tree.}
\label{fig_power}
\end{figure}

Table~\ref{tab_coverage} displays confidence intervals for $\nu_{sib}^\T \mu$ and $\nu_{reg}^\T \mu$ for the three versions of the selective $Z$-intervals, where data is generated as in Section~\ref{subsubsec_CIcompare}. As expected, $\hat{\sigma}_{\text{cons}}$ leads to slight over-coverage and $\hat{\sigma}_{\text{SSE}}$ leads to slight under-coverage. 

\begin{table}[H]
\center
\begin{tabular}{ccccccccc}
&& \multicolumn{3}{c}{Parameter $\nu^\T_{reg}{\mu}$} && \multicolumn{3}{c}{Parameter $\nu^\T_{sib}{\mu}$}  \\[5pt]
Level  && $\sigma$ & $\hat{\sigma}_{\text{cons}}$ & $\hat{\sigma}_{\text{SSE}}$ && $\sigma$ & $\hat{\sigma}_{\text{cons}}$ & $\hat{\sigma}_{\text{SSE}}$ \\[5pt]
1 && 0.95 & 0.98 & 0.95 && 0.95 & 0.98  & 0.94 \\
2 && 0.95 & 0.97 & 0.94 && 0.95 & 0.97 & 0.94 \\
3 && 0.95 & 0.96 & 0.94 & & 0.95 & 0.96 & 0.94 \\
\end{tabular}
\caption{Proportion of 95\% confidence intervals containing the true parameter, aggregated over all trees fit to the 5,500 datasets generated with $(a,b) \in \{0.5,1,2\} \times \{1,\ldots,10\}$}
\label{tab_coverage}	
\end{table}

In this section, we have seen that when trees are not grown overly large, plugging in $\hat{\sigma}_{\text{SSE}}$ leads to approximate selective Type 1 error control, approximately correct selective coverage, and good power. Unfortunately, providing theoretical guarantees for our procedures when using $\hat{\sigma}_{\text{SSE}}$ would be quite difficult, as the estimator is anti-conservative and depends on the output of CART. Providing theoretical guarantees for our procedures under $\hat{\sigma}_{\text{cons}}$ is more straightforward, using ideas from \cite{gao2020selective}, \cite{chen2022selective}, and \cite{tibshirani2018uniform}. However, as shown in Figure~\ref{fig_power}, selective $Z$-tests based on $\hat{\sigma}_{\text{cons}}$ can have very low power. One promising avenue of future work involves 
providing theoretical guarantees in the regression tree setting for estimators that are less conservative than $\hat{\sigma}_{\text{cons}}$.

\begin{figure}[!h]
\begin{minipage}[t]{0.6\textwidth}
\subcaption*{}
\includegraphics[height=7cm]{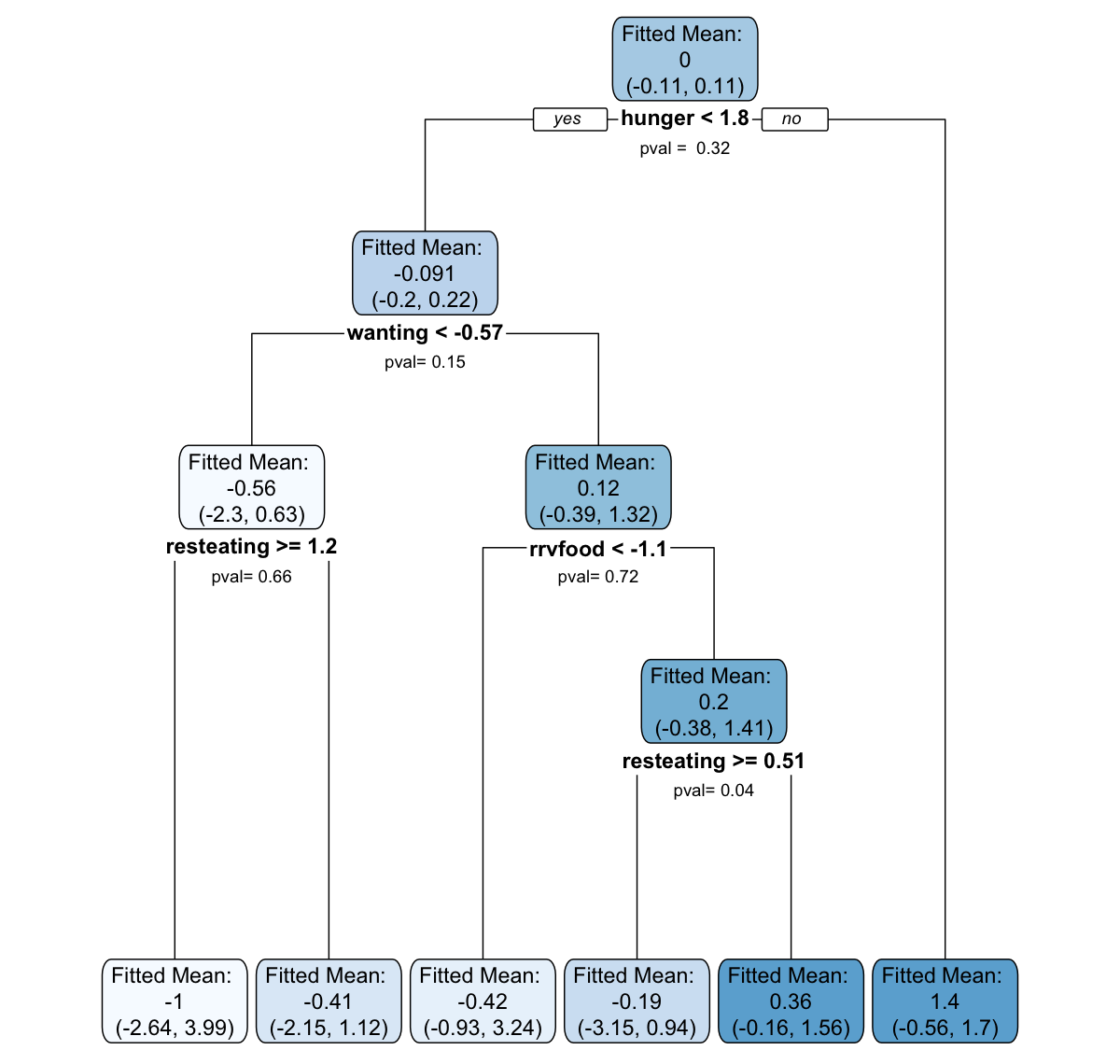}
\end{minipage}\hfill
\begin{minipage}[t]{6.5cm}
\subcaption*{}
\includegraphics[height=3.8cm, width=6.3cm]{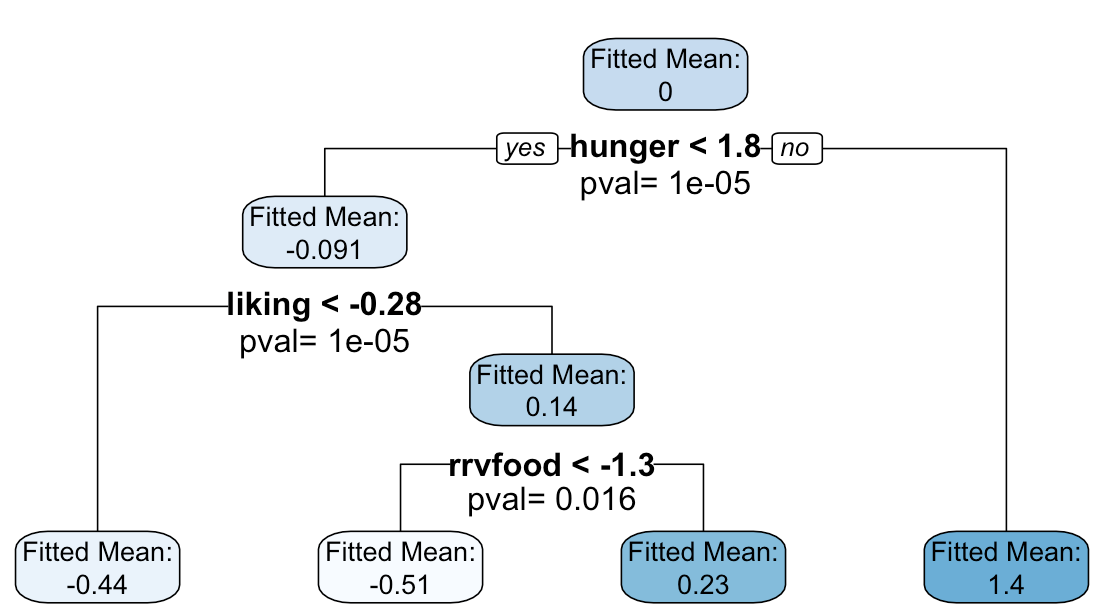}
\vspace{-2mm}
\subcaption*{}
\includegraphics[height=3.5cm, width=6.5cm]{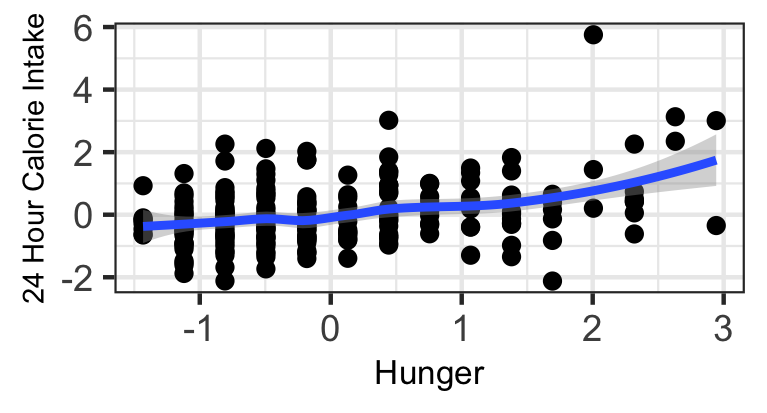}
\end{minipage}
\caption{\emph{Left: }A CART tree fit to the Box Lunch Study data. Each split has been labeled with a p-value \eqref{def_mainpval}, and each region has been labeled with a confidence interval \eqref{eq:branchCI}. The shading of the nodes indicates the average response values (white indicates a very small value and dark blue  a very large value). \emph{Top right:} A CTree fit to the Box Lunch Study data. \emph{Bottom right:} A scatterplot showing the relationship between the covariate hunger and the response.}
\label{fig_BLStrees}
\end{figure}

\section{An Application to the Box Lunch Study}
\label{section_realData}

\cite{venkatasubramaniam2017decision} compare CART and CTree \citep{hothorn2006unbiased} within the context of epidemiological studies. They conclude that CTree is preferable to CART because it provides p-values for each split, even though CART  has higher predictive accuracy. Since our framework provides p-values for each split in a CART tree, we revisit their analysis of the Box Lunch Study, a clinical trial studying the impact of portion control interventions on 24-hour caloric intake. We consider 
identifying subgroups of study participants with baseline differences in 24-hour caloric intake on the basis of scores from an 
assessment that quantifies constructs such as hunger, liking, the relative reinforcement of food (\verb+rrvfood+), and restraint (\verb+resteating+). 

We exactly reproduce the trees presented in Figures 1 and 2 of \cite{venkatasubramaniam2017decision} by building a CTree using \texttt{partykit} and a CART tree using \texttt{rpart} on the Box Lunch Study data provided in the \verb=R= package \texttt{visTree} \citep{venkat2018package}. We apply our selective inference framework to compute p-values  \eqref{def_mainpval} for each split in CART,  and confidence intervals  \eqref{eq:branchCI} for each region. In this section, we use $\hat{\sigma}_{\text{SSE}}$, defined in Section~\ref{subsec:unknownvar}, to estimate the error variance. The results are shown in Figure~\ref{fig_BLStrees}. 

Both CART and CTree choose 
\texttt{hunger$<$1.8} as the first split. 
For this split, our selective $Z$-test reports a large p-value of 0.44, while CTree reports a p-value less than 0.001. The conflicting p-values are explained by the 
difference in null hypotheses. 
CTree finds strong evidence against the null of no linear association between \texttt{hunger} and caloric intake. By contrast, our selective framework for CART does not find strong evidence for a difference between  mean caloric intake of participants with \texttt{hunger$<$1.8} and those with \texttt{hunger$\geq$1.8}. We see from the bottom right of Figure~\ref{fig_BLStrees} that while there is evidence of a linear relationship between \texttt{hunger} and caloric intake, there is less evidence of a difference in means across the particular split $\texttt{hunger$=$1.8}$. Given that the goal of \cite{venkatasubramaniam2017decision} is to ``identify population subgroups that are relatively homogeneous with respect to an outcome", the p-value resulting from our selective framework is more natural than the p-value output by CTree, since the former relates directly to the subgroups formed by the split, whereas the latter does not take into account the location of the split point. In general, the left-hand panel of Figure~\ref{fig_BLStrees} shows that the subgroups of patients identified by CART are not significantly different from one another. This is an important finding that would be missed without our selective inference framework. Furthermore, unlike CTree, our framework provides confidence intervals for the mean response in each subgroup. 

An alternative analysis using $\hat{\sigma}_{\text{cons}}$, defined in Section~\ref{subsec:unknownvar}, is provided in Appendix~\ref{appendix_bls}, and leads to similar findings.

\section{Discussion} 
\label{section_disc}

Our framework relies on the assumption that $Y \sim N_n(\mu, \sigma^2 I)$, with $\sigma^2$ known. In Section~\ref{subsec:unknownvar}, we showed strong empirical performance when the variance is unknown and $\sigma^2$ is estimated. In this section, we briefly comment on the assumptions of spherical variance and normally distributed data.

It natural to wonder whether the assumption that $Y \sim N_n(\mu, \sigma^2 I)$ can be relaxed to the assumption that $Y \sim N_n(\mu, \Sigma)$, with $\Sigma$ known. Following the work of \cite{lee2016exact}, the results in Section~\ref{section_frame} extend to the setting where $Y \sim N_n(\mu, \Sigma)$ if we:
\begin{enumerate}
\item Modify \eqref{def_mainpval} and \eqref{eq:pvalreg} to condition on the event 
$
\left \{\left(I_n - \frac{\Sigma \nu \nu^T}{\nu^T \Sigma \nu}\right) Y = \left(I_n - \frac{\Sigma \nu \nu^T}{\nu^T \Sigma \nu}\right) y \right \}$
rather than the event $\{ \mathcal{P}_\nu^\perp Y =  \mathcal{P}_\nu^\perp y \}$, where $\nu = \nu_{sib}$ in the case of (8) and $\nu = \nu_{reg}$ in the case of (15). 
\item Replace all instances of the perturbation $y'(\phi,\nu)$, defined in Theorem~\ref{theorem_1jewell}, with the perturbation 
$y''(\phi,\nu) = \left(I_n - \frac{\Sigma \nu \nu^T}{\nu^T \Sigma \nu}\right) y + \frac{\Sigma \nu}{\nu^T \Sigma \nu} \phi.$
\end{enumerate}
Unfortunately, the modified perturbation $y''(\phi, \nu)$ does not satisfy Condition~\ref{cond_nuprop} in Section~\ref{subsec_growingcalculations} when $\Sigma \neq \sigma^2 I_n$, and so many of the results of Section~\ref{section_computing} do not extend to this non-spherical setting. Future work could explore how to efficiently compute the conditioning set in this non-spherical setting.

Furthermore, our framework assumes a normally-distributed response variable.
CART is commonly used for classification, survival \citep{segal1988regression}, and treatment effect estimation in causal inference \citep{athey2016recursive}. While the idea of conditioning on a selection event to control the selective Type 1 error rate applies regardless of the distribution of the response, our Theorem~\ref{theorem_1jewell} and Theorem~\ref{theorem_1modified}, and the resulting computational results, relied on normality of $Y$. In the absence of this assumption, exactly characterizing the conditioning set and the distribution of the test statistic requires further investigation.

We show in Appendix~\ref{appendix:non-normal} that our selective $Z$-tests approximately control the selective Type 1 error when the normality assumption is violated. \cite{tian2017asymptotics} and \cite{tibshirani2018uniform} establish conditions under which selective p-values for linear regression (derived under the assumption of normality) will be asymptotically uniformly distributed under non-normality. Thus suggests the possibility of developing asymptotic theory for our proposed selective $Z$-tests under violations of normality.

A reviewer pointed out similarities between the problem of testing significance of the first split in the tree and significance testing for a single changepoint, as in \cite{bhattacharya1994some}. Building on this connection may provide an avenue for future work.

A software implementation of the methods in this paper is available in the \texttt{R} package \texttt{treevalues}, at \texttt{https://github.com/anna-neufeld/treevalues}.

\section*{Acknowledgements}
Daniela Witten and Anna Neufeld were supported by the National Institutes of Health and the Simons Foundation. Lucy Gao was supported by the Natural Sciences and Engineering Research Council of Canada Discovery Grants program.

\bibliography{tree_values}

\begin{thebibliography}{30}
\providecommand{\natexlab}[1]{#1}
\providecommand{\url}[1]{\texttt{#1}}
\expandafter\ifx\csname urlstyle\endcsname\relax
  \providecommand{\doi}[1]{doi: #1}\else
  \providecommand{\doi}{doi: \begingroup \urlstyle{rm}\Url}\fi

\bibitem[Athey and Imbens(2016)]{athey2016recursive}
Susan Athey and Guido Imbens.
\newblock Recursive partitioning for heterogeneous causal effects.
\newblock \emph{Proceedings of the National Academy of Sciences}, 113\penalty0
  (27):\penalty0 7353--7360, 2016.

\bibitem[Bhattacharya(1994)]{bhattacharya1994some}
PK~Bhattacharya.
\newblock Some aspects of change-point analysis.
\newblock \emph{Lecture Notes-Monograph Series}, pages 28--56, 1994.

\bibitem[Bourgon(2009)]{bourgon2009intervals}
Richard Bourgon.
\newblock \emph{Overview of the intervals package}, 2009.
\newblock {R} Vignette, URL
  \url{https://cran.r-project.org/web/packages/intervals/vignettes/intervals_overview.pdf}.

\bibitem[Breiman et~al.(1984)Breiman, Friedman, Stone, and
  Olshen]{breiman1984classification}
Leo Breiman, Jerome Friedman, Charles~J Stone, and Richard~A Olshen.
\newblock \emph{Classification and regression trees}.
\newblock CRC Press, 1984.

\bibitem[Chen and Bien(2020)]{chen2019valid}
Shuxiao Chen and Jacob Bien.
\newblock Valid inference corrected for outlier removal.
\newblock \emph{Journal of Computational and Graphical Statistics}, 29\penalty0
  (2):\penalty0 323--334, 2020.

\bibitem[Chen and Witten(2022)]{chen2022selective}
Yiqun~T Chen and Daniela~M Witten.
\newblock Selective inference for k-means clustering.
\newblock \emph{arXiv preprint arXiv:2203.15267}, 2022.

\bibitem[Fithian et~al.(2014)Fithian, Sun, and Taylor]{fithian2014optimal}
William Fithian, Dennis Sun, and Jonathan Taylor.
\newblock Optimal inference after model selection.
\newblock \emph{arXiv preprint arXiv:1410.2597}, 2014.

\bibitem[Gao et~al.(2020)Gao, Bien, and Witten]{gao2020selective}
Lucy~L Gao, Jacob Bien, and Daniela Witten.
\newblock Selective inference for hierarchical clustering.
\newblock \emph{arXiv preprint arXiv:2012.02936}, 2020.

\bibitem[Hothorn and Zeileis(2015)]{hothorn2015partykit}
Torsten Hothorn and Achim Zeileis.
\newblock {partykit: A modular toolkit for recursive partytioning in R}.
\newblock \emph{The Journal of Machine Learning Research}, 16\penalty0
  (1):\penalty0 3905--3909, 2015.

\bibitem[Hothorn et~al.(2006)Hothorn, Hornik, and Zeileis]{hothorn2006unbiased}
Torsten Hothorn, Kurt Hornik, and Achim Zeileis.
\newblock Unbiased recursive partitioning: A conditional inference framework.
\newblock \emph{Journal of Computational and Graphical Statistics}, 15\penalty0
  (3):\penalty0 651--674, 2006.

\bibitem[Hubert and Arabie(1985)]{hubert1985comparing}
Lawrence Hubert and Phipps Arabie.
\newblock Comparing partitions.
\newblock \emph{Journal of Classification}, 2\penalty0 (1):\penalty0 193--218,
  1985.

\bibitem[Hyun et~al.(2021)Hyun, Lin, G'Sell, and Tibshirani]{hyun2018post}
Sangwon Hyun, Kevin~Z Lin, Max G'Sell, and Ryan~J Tibshirani.
\newblock Post-selection inference for changepoint detection algorithms with
  application to copy number variation data.
\newblock \emph{Biometrics}, pages 1--13, 2021.

\bibitem[Jewell et~al.(2022)Jewell, Fearnhead, and Witten]{jewell2019testing}
Sean Jewell, Paul Fearnhead, and Daniela Witten.
\newblock Testing for a change in mean after changepoint detection.
\newblock \emph{Journal of the Royal Statistical Society, Series B}, 2022.

\bibitem[Kivaranovic and Leeb(2021)]{kivaranovic2020length}
Danijel Kivaranovic and Hannes Leeb.
\newblock On the length of post-model-selection confidence intervals
  conditional on polyhedral constraints.
\newblock \emph{Journal of the American Statistical Association}, 116\penalty0
  (534):\penalty0 845--857, 2021.

\bibitem[Lee et~al.(2016)Lee, Sun, Sun, Taylor, et~al.]{lee2016exact}
Jason~D Lee, Dennis~L Sun, Yuekai Sun, Jonathan~E Taylor, et~al.
\newblock Exact post-selection inference, with application to the lasso.
\newblock \emph{The Annals of Statistics}, 44\penalty0 (3):\penalty0 907--927,
  2016.

\bibitem[Liu et~al.(2018)Liu, Markovic, and Tibshirani]{liu2018more}
Keli Liu, Jelena Markovic, and Robert Tibshirani.
\newblock More powerful post-selection inference, with application to the
  lasso.
\newblock \emph{arXiv preprint arXiv:1801.09037}, 2018.

\bibitem[Loh(2014)]{loh2014fifty}
Wei-Yin Loh.
\newblock Fifty years of classification and regression trees.
\newblock \emph{International Statistical Review}, 82\penalty0 (3):\penalty0
  329--348, 2014.

\bibitem[Loh et~al.(2016)Loh, Fu, Man, Champion, and Yu]{loh2016identification}
Wei-Yin Loh, Haoda Fu, Michael Man, Victoria Champion, and Menggang Yu.
\newblock Identification of subgroups with differential treatment effects for
  longitudinal and multiresponse variables.
\newblock \emph{Statistics in medicine}, 35\penalty0 (26):\penalty0 4837--4855,
  2016.

\bibitem[Loh et~al.(2019)Loh, Man, and Wang]{loh2019subgroups}
Wei-Yin Loh, Michael Man, and Shuaicheng Wang.
\newblock Subgroups from regression trees with adjustment for prognostic
  effects and postselection inference.
\newblock \emph{Statistics in Medicine}, 38\penalty0 (4):\penalty0 545--557,
  2019.

\bibitem[Ripley(1996)]{ripley1996pattern}
Brian~D Ripley.
\newblock \emph{Pattern recognition and neural networks}.
\newblock Cambridge University Press, 1996.

\bibitem[Segal(1988)]{segal1988regression}
Mark~Robert Segal.
\newblock Regression trees for censored data.
\newblock \emph{Biometrics}, 44\penalty0 (1):\penalty0 35--47, 1988.

\bibitem[Taylor and Tibshirani(2015)]{taylor2015statistical}
Jonathan Taylor and Robert~J Tibshirani.
\newblock Statistical learning and selective inference.
\newblock \emph{Proceedings of the National Academy of Sciences}, 112\penalty0
  (25):\penalty0 7629--7634, 2015.

\bibitem[Therneau and Atkinson(2019)]{therneau2015package}
Terry Therneau and Beth Atkinson.
\newblock \emph{rpart: Recursive Partitioning and Regression Trees}, 2019.
\newblock R package version 4.1-15, available on CRAN.

\bibitem[Tian and Taylor(2017)]{tian2017asymptotics}
Xiaoying Tian and Jonathan Taylor.
\newblock Asymptotics of selective inference.
\newblock \emph{Scandinavian Journal of Statistics}, 44\penalty0 (2):\penalty0
  480--499, 2017.

\bibitem[Tian and Taylor(2018)]{tian2018selective}
Xiaoying Tian and Jonathan Taylor.
\newblock Selective inference with a randomized response.
\newblock \emph{The Annals of Statistics}, 46\penalty0 (2):\penalty0 679--710,
  2018.

\bibitem[Tibshirani et~al.(2016)Tibshirani, Taylor, Lockhart, and
  Tibshirani]{tibshirani2016exact}
Ryan~J Tibshirani, Jonathan Taylor, Richard Lockhart, and Robert Tibshirani.
\newblock Exact post-selection inference for sequential regression procedures.
\newblock \emph{Journal of the American Statistical Association}, 111\penalty0
  (514):\penalty0 600--620, 2016.

\bibitem[Tibshirani et~al.(2018)Tibshirani, Rinaldo, Tibshirani, and
  Wasserman]{tibshirani2018uniform}
Ryan~J Tibshirani, Alessandro Rinaldo, Rob Tibshirani, and Larry Wasserman.
\newblock Uniform asymptotic inference and the bootstrap after model selection.
\newblock \emph{The Annals of Statistics}, 46\penalty0 (3):\penalty0
  1255--1287, 2018.

\bibitem[Venkatasubramaniam and Wolfson(2018)]{venkat2018package}
Ashwini Venkatasubramaniam and Julian Wolfson.
\newblock \emph{visTree: Visualization of Subgroups for Decision Trees}, 2018.
\newblock R package version 0.8.1, available on CRAN.

\bibitem[Venkatasubramaniam et~al.(2017)Venkatasubramaniam, Wolfson, Mitchell,
  Barnes, JaKa, and French]{venkatasubramaniam2017decision}
Ashwini Venkatasubramaniam, Julian Wolfson, Nathan Mitchell, Timothy Barnes,
  Meghan JaKa, and Simone French.
\newblock Decision trees in epidemiological research.
\newblock \emph{{Emerging Themes in Epidemiology}}, 14\penalty0 (1):\penalty0
  11, 2017.

\bibitem[Wager and Walther(2015)]{wager2015adaptive}
Stefan Wager and Guenther Walther.
\newblock Adaptive concentration of regression trees, with application to
  random forests.
\newblock \emph{arXiv preprint arXiv:1503.06388}, 2015.

\end{thebibliography}

\appendix
\section{Comparison to \cite{loh2019subgroups}}
\label{appendix:loh}

\cite{loh2016identification} and \cite{loh2019subgroups} use regression trees to find subgroups of patients with similar treatment effects in clinical trials. 
They grow trees based on patient characteristics using a different algorithm than CART. Furthermore, they are interested in the mean treatment effect (which is a linear regression coefficient) within each terminal region of the tree, rather than the mean response within each region. 

One of the goals of \cite{loh2019subgroups} is to construct valid post-selection confidence intervals for the treatment effect within each terminal node. In this appendix, we show that their approach, when adapted to the setting of this paper, does not yield confidence intervals with nominal coverage. 

The basic idea of \cite{loh2019subgroups}, instantiated to our setting, is as follows. Suppose that $R_A \in \textsc{Tree}^\lambda(y)$, and define the vector $\nu_{reg}$ such that $(\nu_{reg})_i = 1_{(x_i \in R_A)}/\left\{ \sum_{i'=1}^n 1_{(x_{i'} \in R_A)}\right\}$, as in \eqref{eq_nureg}.  We know that the ``naive" Z-interval does not achieve nominal coverage, meaning that
\begin{equation}
\label{eq_naive}
\mathrm{Pr}\left( \nu_{reg}^\top \mu \in \left[ \nu_{reg}^\top y - z_{\alpha/2} \frac{\sigma}{\sqrt{\sum_{i=1}^n 1_{(x_i \in R_A)}}}, \nu_{reg}^\top y + z_{\alpha/2} \frac{\sigma}{\sqrt{\sum_{i=1}^n 1_{(x_i \in R_A)}}}\right] \right) < 1-\alpha. 
\end{equation}
This is because the ``multiplier" for the naive confidence interval, $z_{\alpha/2}$, is derived under the assumption that the region $R_A$ (or, equivalently, the vector $\nu_{reg}$), is fixed, rather than a function of the data. 

\cite{loh2019subgroups} observe that there exists some $\alpha' < \alpha$ such that 
\begin{equation}
\label{eq_alphaprime}
\mathrm{Pr}\left( \nu_{reg}^\top \mu \in \left[ \nu_{reg}^\top y - z_{\alpha'/2} \frac{\sigma}{\sum_{i=1}^n 1_{(x_i \in R_A)}}, \nu_{reg}^\top y + z_{\alpha'/2} \frac{\sigma}{\sum_{i=1}^n 1_{(x_i \in R_A)}}\right] \right) = 1-\alpha. 
\end{equation}
The value for $\alpha'$ for a given tree will depend on the number of split covariates $p$, the number of data points $n$, the depth of the tree, and the value of $\lambda$ used for tree pruning, among other considerations.
If we know how the data was generated, then we can check whether some value $\alpha'$ satisfies \eqref{eq_alphaprime} as follows:
\begin{enumerate}
\item Draw $B$ different simulated datasets $\{(X_b, y_b)\}_{b=1}^B$ from the same distribution (call this $F$) as the original data. For $b=1,\ldots,B$:
\begin{enumerate}
\item Build a tree using the simulated data $(X_b,y_b)$, using the same procedure and the same settings as in Step 1, and denote it $\tree^\lambda(y_b)$. 
\item For each terminal region $R \in \term\left(\tree^\lambda(y_b), \mathbb{R}^p\right)$ in the tree:
\begin{enumerate}
\item Construct a $(1-\alpha')$ naive $Z$-interval for the mean response in the region using $(X_b,y_b)$.
\item Check if each interval contains $\bar{\mu}_R$, the true mean for this region $R$. 
\end{enumerate}
\end{enumerate}
\item Compute the fraction of intervals in 1(b) that contain $\bar{\mu}_R$. 
\item If this value is $1-\alpha$, then we have found the correct value of $\alpha'$. If not, then we try a larger or smaller value of $\alpha'$. 
\end{enumerate}

We test this procedure in a very simple simulation study. We generate data $X_{ij} \sim N(0,1)$ and $y_i \sim N(0,1)$ for $i=1,\ldots,100$ and $j=1,\ldots,p$. In this simple setting, the true mean response for every region in every fitted tree is $0$. We carry out the procedure outlined above with $\alpha=0.1$. For two values of $p$ and for CART trees with 1, 2, and 3 levels, we create 1000 datasets and 1000 trees and report the empirical coverage of the intervals obtained using this ideal method, averaged over all nodes in all trees. The results, shown in Table~\ref{table_loh_results}, show that this ideal procedure procedure leads to intervals that achieve nominal coverage. 

Unfortunately, this ideal procedure is practically infeasible, as it requires the user to know the true distribution of the data $F$. Thus, in practice, \cite{loh2019subgroups} propose replacing $F$ by $\hat{F}$, the empirical distribution of the original data. This amounts to replacing the simulated datasets in Step 2 with bootstrapped datasets, and checking whether the naive $Z$-intervals in Step 1(b) contain $\bar{y}_R$ rather than $\bar{\mu}_R$. 

We can see why this is problematic in a very simple setting where we fit a tree with depth 1. If all observations have mean $0$, then a CART tree fit to a bootstrap sample of the data will nevertheless find regions $R_L$ and $R_R$ such that the sample mean value of $y_b$ within $R_L$ is negative and the sample mean value of $y_b$ within $R_R$ is positive. As $y_b$ and $y$ contain many overlapping observations, it is likely that the sample mean value of $y$ within $R_L$ is also negative and the sample mean value of $y$ within $R_R$ is also positive. In other words, because of the overlap between $y$ and $y_b$, the within-region sample means of $y_b$ are closer to the within-region sample means of $y$ than they are to the within-region population means. Thus, when we calibrate $\alpha'$ to cover the mean values of $y$ within various regions, we end up with under-coverage of the true population mean, as shown in Table~\ref{table_loh_results}. 

We see in Table~\ref{table_loh_results} that our selective inference framework approach enables valid inference in this setting, whereas a bootstrap procedure modeled after \cite{loh2019subgroups} does not. 

\begin{table}[H]
\centering
\begin{tabular}{cc|cc|cc|c}
&& \multicolumn{2}{c|}{Loh (ideal)} & \multicolumn{2}{c|}{Loh (bootstrap)} & Selective CIs  \\
p & Tree depth & Coverage & Average $\alpha'$ & Coverage & Average $\alpha'$ & Coverage \\
\hline
& 1 & 0.902 & 0.008 & 0.749 & 0.037 & 0.890  \\
2 & 2 & 0.905 & 0.004 & 0.695 & 0.038 & 0.904 \\
& 3 & 0.900 & 0.005 & 0.660 & 0.047 & 0.895 \\
\hline 
& 1 & 0.883 & 0.001 & 0.601 & 0.016 & 0.908 \\
20 & 2 & 0.900 & 0.00025 & 0.549 & 0.016 & 0.901 \\
 & 3 & 0.904 & 0.00015 & 0.543 & 0.022  & 0.905  \\
\end{tabular}	
\caption{Coverage of 90\% confidence intervals computed using three methods for the simple setting where $y_i \sim N(0,1)$ and $X_{ij} \sim N(0,1)$ for  for $i=1,\ldots,100$ and $j=1,\ldots,p$. Note that the ``Loh (ideal)" method can never be used in practice, as it requires knowledge of the true parameter. }
\label{table_loh_results}
\end{table}

\section{Proofs for Section~\ref{section_frame}}
\label{appendix:section3proofs}

\subsection{Proof of Theorem~\ref{theorem_1jewell}}
\label{appendix:theorem1proof}

Let $0 \leq \alpha \leq 1$. We start by proving the first statement in Theorem~\ref{theorem_1jewell}:
$$pr_{H_0}\left\{ p_{sib}(Y) \leq \alpha \mid R_A,R_B \text{ are siblings in } \tree^\lambda(Y) \right\} = \alpha.$$
This is a special case of Proposition 3 from \cite{fithian2014optimal}. It follows from the definition of $p_{sib}(Y)$ in \eqref{def_mainpval} that 
$$pr_{H_0} \left\{p_{sib}(Y) \leq \alpha \mid R_A,R_B \text{ are siblings in } \tree^\lambda(Y), \mathcal{P}_{\nu_{sib}}^\perp Y = \mathcal{P}_{\nu_{sib}}^\perp y \right\} = \alpha.$$
Therefore, applying the law of total expectation yields
\begin{align*}
&pr_{H_0}\left\{ p_{sib}(Y) \leq \alpha \mid R_A,R_B \text{ siblings in } \tree^\lambda(Y) \right\} \\
&= E_{H_0}\left[ 1_{\left\{p_{sib}(Y) < \alpha\right\}} \mid R_A,R_B \text{ are siblings in } \tree^\lambda(Y) \right] \\
&= E_{H_0}\biggl( E_{H_0}\left[ 1_{\{p_{sib}(Y) < \alpha\}} \mid R_A,R_B \text{ are siblings in } \tree^\lambda(Y), \mathcal{P}_{\nu_{sib}}^\perp Y = \mathcal{P}_{\nu_{sib}}^\perp y \right]  \\
& \hspace{40mm} \mid R_A,R_B \text{ are siblings in } \tree^\lambda(Y) \biggr) \\
&= E_{H_0}\left\{ \alpha \mid R_A,R_B \text{ are siblings in } \tree^\lambda(Y) \right\} = \alpha.
\end{align*}

The second statement of Theorem~\ref{theorem_1jewell} follows directly from the following result. 
\begin{lemma}
\label{lemma_tn}
If $Y \sim N_n\left(\mu, \sigma^2 I_n\right)$, then
the random variable $\nu_{sib}^\T Y$ has the following conditional distribution: 
\footnotesize
\begin{equation}
\label{eq:cond-dist}
\nu_{sib}^\T Y \mid  \{ R_A, R_B \text{ are siblings in } \tree^\lambda(Y), \mathcal{P}_{\nu_{sib}}^\perp Y = \mathcal{P}_{\nu_{sib}}^\perp y \} \sim \mathcal{T}\mathcal{N}\left\{ \nu_{sib}^\T {\mu}, \sigma^2 \|\nu_{sib}\|_2^2; S_{sib}^\lambda(\nu_{sib}) \right\},
\end{equation}
\normalsize
where $S^\lambda_{sib}(\nu_{sib})$ is defined in \eqref{def_S} and $\mathcal{T}\mathcal{N}\left( \mu,\sigma, S \right)$ denotes the $N(\mu,\sigma^2)$ distribution truncated to the set $S$.
\end{lemma}
\begin{proof}

The following holds for any $\nu \in \mathbb{R}^n$.  
\begin{align*}	
& pr\left\{ \nu^\T Y > c \mid  R_A, R_B \text{ are siblings in } \textsc{tree}^\lambda\left(Y\right), \mathcal{P}_\nu^\perp Y = \mathcal{P}_\nu^\perp y \right\} \\
&= pr\left\{ \nu^\T Y > c \mid R_A, R_B \text{ are siblings in } \textsc{tree}^\lambda\left(\mathcal{P}_\nu^\perp Y+\frac{\nu \nu^\T}{\|\nu\|_2^2} Y\right), \mathcal{P}_\nu^\perp Y = \mathcal{P}_\nu^\perp y \right\} \\
&= pr\left\{ \nu^\T Y > c \mid   R_A, R_B \text{ are siblings in } \textsc{tree}^\lambda\left(\mathcal{P}_\nu^\perp y+\frac{\nu}{\|\nu\|_2^2} \nu^\T Y\right), \mathcal{P}_\nu^\perp Y = \mathcal{P}_\nu^\perp y  \right\} \\
&= pr\left\{ \nu^\T Y > c \mid   R_A, R_B \text{ are siblings in } \textsc{tree}^\lambda\left(\mathcal{P}_\nu^\perp y+\frac{\nu}{\|\nu\|_2^2} \nu^\T Y\right) \right\} \\
&= pr\left\{ \phi > c \mid  \phi \in S_{sib}^{\lambda}\left(\nu\right)\right\},
\end{align*}
where $\phi = \nu^\T Y.$ 
In the fourth line, the condition $\mathcal{P}_\nu^\perp Y = \mathcal{P}_\nu^\perp y$ can be dropped because when $Y \sim N_n\left(\mu, \sigma^2 I_n\right)$, $\mathcal{P}_\nu^\perp Y$ is independent of  $\nu^\top Y$.
Finally, since $\phi \sim N\left(\nu^\T \mu, \sigma^2 \|\nu\|_2^2\right)$, \eqref{eq:cond-dist} holds. 
\end{proof}

\subsection{Proof of Proposition~\ref{prop_CI}}
\label{appendix:propCIproof}

Theorem 6.1 from \cite{lee2016exact} says that the truncated normal distribution has monotone likelihood ratio in the mean parameter. This guarantees that $L(y)$ and $U(y)$ in \eqref{eq_mainCI} are unique. Then, for $L(\cdot)$ and $U(\cdot)$ in \eqref{eq_mainCI}, \eqref{eq:cond-dist} in Lemma~\ref{lemma_tn} guarantees that 
\small
\begin{equation}
\label{eq_mainCoverage}
pr\left\{ \nu^\T {\mu} \in \left[ L(Y),U(Y) \right]  \mid R_A,R_B \text{ are siblings in } \textsc{tree}^\lambda(Y) ,\mathcal{P}_{\nu}^\perp Y = {\mathcal{P}}_{\nu}^\perp y \right\} =1-\alpha.
\end{equation}
\normalsize
Finally, we need to prove that \eqref{eq_mainCoverage} implies $(1 - \alpha)$--\emph{selective coverage} as defined in \eqref{eq:selcov}. Following Proposition 3 from \cite{fithian2014optimal}, let $\eta$ be the random variable ${\mathcal{P}}_{\nu}^\perp Y$ and let $f(\cdot)$ be its density. Then,
\begin{align*}
	&pr\left\{ \nu^\T {\mu} \in \left[ L(Y),U(Y) \right]  \mid R_A,R_B \text{ are siblings in } \textsc{tree}^\lambda(Y) \right\} \\
	&= \int pr\left\{\nu^\T {\mu} \in \left[ L(Y),U(Y) \right]  \mid R_A,R_B \text{ are siblings in } \textsc{tree}^\lambda(Y) ,\mathcal{P}_{\nu}^\perp Y = {\mathcal{P}}_{\nu}^\perp y \right\} f(\eta) d\eta \\
	&= \int (1-\alpha) f(\eta) d\eta = 1-\alpha.
\end{align*}

\subsection{Proof of Theorem~\ref{theorem_1modified}}
\label{appendix:theorem2proof}

We omit the proof of the first statement of Theorem~\ref{theorem_1modified}, as it is similar to the proof of the first statement of Theorem~\ref{theorem_1jewell} in Appendix \ref{appendix:theorem1proof}. 

The second statement in Theorem~\ref{theorem_1modified} follows directly from the following result. 
\begin{lemma}
\label{lemma_tn_reg}
The random variable $\nu_{reg}^\T Y$ has the conditional distribution
\begin{equation}
\label{eq:cond-dist-reg}
\nu_{reg}^\T Y \mid  \{ R_A \in \textsc{tree}^\lambda(Y), \mathcal{P}_{\nu_{reg}}^\perp Y = \mathcal{P}_{\nu_{reg}}^\perp y \} \sim \mathcal{T}\mathcal{N}\left\{ \nu_{reg}^\T {\mu}, \sigma^2 \|\nu_{reg}\|_2^2; S_{reg}^\lambda(\nu_{reg}) \right\},
\end{equation}
where $S^\lambda_{reg}(\nu_{reg})$ was defined in \eqref{def_S2}. 
\end{lemma}
We omit the proof of Lemma \ref{lemma_tn_reg}, as it is similar to the proof of Lemma \ref{lemma_tn}. 

\subsection{Proof of Proposition~\ref{prop_CIreg}}
\label{appendix:propCIregproof}

The proof largely follows the proof of Proposition~\ref{prop_CI}. The fact that the truncated normal distribution has monotone likelihood ratio (Theorem 6.1 of \citealt{lee2016exact}) ensures that $L(y)$ and $U(y)$ defined in \eqref{eq_mainCIreg} are unique, and \eqref{eq:cond-dist-reg} in Lemma~\ref{lemma_tn_reg} implies that 
$$pr\left\{\nu_{reg}^\T {\mu} \in \left[ L(Y),U(Y) \right]  \mid R_A \in  \textsc{tree}^\lambda(Y) ,\mathcal{P}_{\nu_{reg}}^\perp Y = {\mathcal{P}}_{\nu_{reg}}^\perp y \right\}=1-\alpha.$$ 
The rest of the argument is as in the proof of Proposition~\ref{prop_CI}.

\section{Proofs for Section~\ref{subsec_branches}}
\label{appendix_lemma1proof}

\subsection{Proof of Lemma~\ref{lemma_orderedsplits}}
\label{proof_lemma1}

We first state and prove the following lemma.
\begin{lemma}
\label{lemma_sameGain}
Let $R_A$ and $R_B$ be the regions in the definition of $\nu_{sib}$ in \eqref{eq_nusib}. For an arbitrary region $R$ and for any $j \in \{1,\ldots,p\}$ and $s \in \{1, \ldots,n-1\}$, recall that the potential children of $R$ (the ones that CART will consider adding to the tree when applying Algorithm~\ref{alg_growing} to region $R$) are given by $R \cap \chi_{j,s,0}$ and $R \cap \chi_{j,s,1}$, where $\chi_{j,s,0}$ and $\chi_{j,s,l}$ were defined in \eqref{eq:halfspace}. If $(R_A \cup R_B) \subseteq R \cap \chi_{j,s,0}$ or $(R_A \cup R_B) \subseteq R \cap \chi_{j,s,1}$, then $Gain_R\{y'(\phi,\nu_{sib}), j,s\} =Gain_R\{y, j,s\}$ for all $\phi$.
\end{lemma}
\begin{proof}
It follows from algebra that for $Gain_R(y, j, s)$ defined in \eqref{eq:gain}, 
\scriptsize
\begin{align} 
Gain_R(y, j, s) = -\left\{ \sum_{i=1}^n 1_{(x_i \in R)} \right\} \left(\overline{y}_{R}\right)^2 + \left\{\sum_{i=1}^n 1_{(x_i \in R \cap \chi_{j,s,0})} \right\} \left( \overline{y}_{R \cap \chi_{j,s,0}}\right)^2 + \left\{\sum_{i=1}^n 1_{(x_i \in R \cap \chi_{j,s,1})} \right\}\left(\overline{y}_{R \cap \chi_{j,s,1}}\right)^2, \label{eq:gain2}
\end{align}
\normalsize
where $\bar{y}_R = \left( \sum_{i \in R} y_i\right)/\left\{ {\sum_{i=1}^n 1_{(x_i \in R)}}\right\}$. It follows from \eqref{eq:gain2} that to prove Lemma \ref{lemma_sameGain}, it suffices to show that 
$\overline{y}_{T} = \overline{y'(\phi, \nu_{sib})}_{T}$ for $T \in \{R, R \cap \chi_{j,s,0}, R \cap \chi_{j,s,1}\}$. Recall from Section~\ref{subsec_intuition} that
$
\{y'(\phi, \nu_{sib})\}_i = y_i + \Delta_i
$, where 
$$
\Delta_i = \begin{cases}(\phi - \nu_{sib}^\T y) \frac{ \sum_{i'=1}^n 1_{(x_i' \in R_B)} }{\sum_{i'=1}^n 1_{(x_i' \in R_A \cup R_B)} } &\text{ if } i \in R_A \\
-(\phi - \nu_{sib}^\T y)\frac{\sum_{i'=1}^n 1_{(x_i' \in R_A)} }{ \sum_{i'=1}^n 1_{(x_i' \in R_A \cup R_B)} }&\text{ if } i \in R_B \\
 0 &\text{ otherwise. }	
 \end{cases}
$$
Without loss of generality, assume that $(R_A \cup R_B) \subseteq R \cap \chi_{j,s,0}$. For any $T \in \{R, R \cap \chi_{j,s,0}\}$, $R_A \cup R_B \subseteq T$. 
Thus, 
\small 
\begin{align*}
\overline{y'(\phi,\nu_{sib})}_{T} &=  \frac{1}{\sum_{i=1}^n 1_{(x_i \in T)}} \left\{ \sum_{i \in T \setminus (R_A \cup R_B)} y_i + \sum_{i \in R_A} (y_i + \Delta_i) + \sum_{i \in R_B} (y_i + \Delta_i) \right\} \\
&= \bar{y}_{T} + \frac{\sum_{i \in R_A} \Delta_i + \sum_{i \in R_B} \Delta_i}{\sum_{i=1}^n 1_{(x_i \in T)}}  \\
&=  \bar{y}_T + \frac{ \left\{ \sum_{i=1}^n 1_{(x_i \in R_A)} \right\} (\phi - \nu_{sib}^\T y) \frac{\sum_{i=1}^n 1_{(x_i \in R_B)}}{\sum_{i=1}^n 1_{(x_i \in R_A \cup R_B)}} - \left\{ \sum_{i=1}^n 1_{(x_i \in R_B)} \right\}  (\phi - \nu_{sib}^\T y) \frac{\sum_{i=1}^n 1_{(x_i \in R_A)}}{\sum_{i=1}^n 1_{(x_i \in R_A \cup R_B)}}}{\sum_{i=1}^n 1_{(x_i \in T)}} \\
&= \bar{y}_T + 0  = \bar{y}_T.
\end{align*}
\normalsize
Furthermore,  
\small 
$$
\overline{y'(\phi,\nu_{sib})}_{R \cap \chi_{j,s,1}} = \frac{1}{\sum_{i=1}^n 1_{(x_i \in  R \cap \chi_{j,s,1})}} \sum_{i \in R \cap \chi_{j,s,1}} (y_i + \Delta_i) =  \frac{1}{\sum_{i=1}^n 1_{(x_i \in  R \cap \chi_{j,s,1})}} \sum_{i \in R \cap \chi_{j,s,1}} (y_i + 0) = \overline{y}_{R \cap \chi_{j,s,1}}.
$$
\normalsize
\end{proof}

We will now prove Lemma~\ref{lemma_orderedsplits}. 

It follows from Definition~\ref{def:branch} that if $\mathcal{R}\left[\branch\{R, \tree^\lambda(y)\}\right] \subseteq \textsc{tree}^\lambda\{y'(\phi,\nu_{sib})\}$, then $R_A$ and $R_B$ are siblings in $\textsc{tree}^\lambda\{y'(\phi,\nu_{sib})\}$. This establishes the $(\Leftarrow)$ direction. 

We will prove the $(\Rightarrow)$ direction by contradiction. Suppose that $R_A$ and $R_B$ are siblings in  $\textsc{tree}^\lambda\{y'(\phi,\nu_{sib})\}$. Define $\branch\{R_A, \tree^\lambda(y)\} = \left((j_1, s_1, e_1), \ldots, (j_L, s_L, e_L) \right)$, and define $R^{(l')} = \bigcap \limits_{l=1}^{l'} \chi_{j_l, s_l, e_l}$ for $l'=1, \ldots, L$. Assume that there exists $l \in \{0,\ldots,L-2\}$ such that $R^{(l)} \in \textsc{tree}^\lambda\{y'(\phi,\nu_{sib})\}$ and $R^{(l+1)} \not \in \textsc{tree}^\lambda\{y'(\phi,\nu_{sib})\}$. We assume that any ties between splits that occur at Step 2 of Algorithm~\ref{alg_growing} are broken in the same way for $y$ and $y'(\phi,\nu_{sib})$, and so this implies that  
there exists $(\tilde{j}, \tilde{s}) \neq (j_{l+1}, s_{l+1})$ such that $(\tilde{j}, \tilde{s}) \in \argmax_{j,s} Gain_{R^{(l)}}\{y'(\phi,\nu_{sib}),j,s\}$ and
\begin{equation}
\label{eq_forcontradition}
Gain_{R^{(l)}}\{y'(\phi, \nu_{sib}),j_{l+1},s_{l+1}\} < Gain_{R^{(l)}}\{y'(\phi, \nu_{sib}),\tilde{j},\tilde{s}\}.
\end{equation}

Since $R_A$ and $R_B$ are siblings in $\textsc{tree}^{\lambda}\{y'(\phi, \nu_{sib})\}$, it follows from Lemma~\ref{lemma_sameGain} that  \\$Gain_{R^{(l)}}\{y'(\phi,\nu_{sib}), \tilde{j}, \tilde{s}\}=Gain_{R^{(l)}}(y, \tilde{j}, \tilde{s})$. Also, since $R_A$ and $R_B$ are siblings in $\textsc{tree}^{\lambda}(y)$, it follows from Lemma~\ref{lemma_sameGain} that $Gain_{R^{(l)}}\{y'(\phi,\nu_{sib}), j_{l+1}, s_{l+1}\}=Gain_{R^{(l)}}(y,  j_{l+1}, s_{l+1})$. Applying these facts to \eqref{eq_forcontradition} yields
\begin{equation}
\label{eq_forcontradition2}
Gain_{R^{(l)}}(y ,j_{l+1},s_{l+1}) < Gain_{R^{(l)}}(y,\tilde{j},\tilde{s}).
\end{equation}
But since $R^{(l)}$ and  $R^{(l+1)}$ both appeared in $\textsc{tree}^\lambda(y)$, 
$$
(j_{l+1}, s_{l+1}) \in \argmax_{j,s} Gain_{R^{(l)}}(y,j,s).$$ 
This contradicts \eqref{eq_forcontradition2}. Therefore, for any $l \in \{0, \ldots, L-2\},$ if $R^{(l)} \in \textsc{tree}^\lambda\{y'(\phi,\nu_{sib})\}$, then $R^{(l+1)} \in \textsc{tree}^\lambda\{y'(\phi,\nu_{sib})\}$. Since $R^{(0)} \in \textsc{tree}^\lambda\{y'(\phi, \nu_{sib})\}$, the proof follows by induction. 

\subsection{Proof of Lemma~\ref{lemma_permutations}}
\label{proof_lemma2}

Let $R_A \in \tree^\lambda(y)$ with $\branch\{R_A,\tree^\lambda(y)\} = ((j_1,s_1,e_1),\ldots,(j_L,s_L,e_L))$ such that $R_A = \bigcap_{l=1}^L \chi_{j_l,s_l,e_l}$. Since Algorithm \ref{alg_growing} creates regions by intersecting halfspaces and set intersections are invariant to the order of intersection, it follows that $R_A = \bigcap_{l=1}^L \chi_{j_l,s_l,e_l} \in \textsc{tree}^\lambda\{y'(\phi, \nu_{reg})\}$ if and only if there exists $\pi \in \Pi$ such that 
$$\left \{ \bigcap \limits_{l=1}^{l'} \chi_{j_{\pi(l)}, s_{\pi(l)}, e_{\pi(l)}} \right \}_{l'=1}^L \subseteq  \textsc{tree}^\lambda\{y'(\phi, \nu_{reg})\}.$$ 
By Definitions \ref{def:branch} and \ref{def:branch_perm}, 
$$
\mathcal{R}(\pi[\branch\{R_A, \tree^\lambda(y)\}]) = \left \{ \bigcap \limits_{l=1}^{l'} \chi_{j_{\pi(l)}, s_{\pi(l)}, e_{\pi(l)}} \right \}_{l'=1}^L.
$$
Thus,
\begin{align*} 
S_{reg}^\lambda &= \left\{ \phi: R_A \in \textsc{tree}^\lambda\{y'(\phi, \nu_{reg})\} \right\} \\ 
&= \bigcup \limits_{\pi \in \Pi} \left \{ \phi: \mathcal{R}(\pi[\branch\{R_A, \tree^\lambda(y)\}]) \subseteq \textsc{tree}^\lambda
\left\{y'\left(\phi, \nu_{reg}\right) \right\} \right \} \\ 
&= \bigcup \limits_{\pi \in \Pi}  S^{\lambda}\left(\pi\left[ \branch\{R_A,\tree^\lambda(y)\}\right] , \nu_{reg}\right),
\end{align*} 
where the third equality follows from the definition of $S^\lambda(\mathcal{B}, \nu)$ in \eqref{eq:sbnu}.

\section{Proofs for Section~\ref{subsec_growingcalculations}}

\subsection{Proof of Proposition~\ref{prop_Sisintersection}}

Recall that $\mathcal{B} = \left( (j_1,s_1,e_1),\ldots,(j_L,s_L,e_L)\right)$ and $\mathcal{R}(\mathcal{B}) = \{R^{(0)},\ldots, R^{(L)}\}$. Recall from \eqref{eq:grow} that $S_{grow}(\mathcal{B},\nu) = \{ \phi: \mathcal{R}(\mathcal{B}) \subseteq \tree^0\{y'(\phi,\nu)\}\}$, and that we define $S_{l,j,s} = \bigl\{ \phi : \gain_{R^{(l-1)}}\{y'(\phi,\nu),j,s\} \leq \gain_{R^{(l-1)}}\{y'(\phi,\nu), j_l, s_l \} \bigr\}$. 

For $l=1,\ldots,L$, 
\small 
\begin{equation}
\label{eq_iff}
R^{(l-1)} \in \tree^0\{y'(\phi,\nu)\} \text{ and } \phi \in \cap_{s=1}^{n-1}\cap_{j=1}^p S_{l,j,s} \iff \{R^{(l-1)},R^{(l)}\} \subseteq   \tree^0\{y'(\phi,\nu)\},
\end{equation}
\normalsize
because, given that $R^{(l-1)} \in \tree^0\{y'(\phi,\nu)\}$, $R^{(l)} \in  \tree^0\{y'(\phi,\nu)\}$ if and only if \\
$
 (j_l,s_l)  \in
\argmax_{(j,s) : s \in \{ 1,\ldots,n-1\}, j \in \{1,  \ldots, p\}}
\textsc{gain}_{R^{(l-1)}}\left(y'(\phi,\nu),j,s\right).
$
Combining \eqref{eq_iff} with the fact that $\{\phi : R^{(0)} \in \tree^0\{y'(\phi,\nu)\} \} = \mathbb{R}$ yields
\begin{align*}
	\bigcap_{l=1}^L \bigcap_{j=1}^{p} \bigcap_{s=1}^{n-1} S_{l,j,s} =\bigcap_{l=1}^L   \left\{ \phi : \{R^{(l-1)},R^{(l)}\} \subseteq   \tree^0\{y'(\phi,\nu)\} \right\} = \{ \phi : \mathcal{R}(\mathcal{B}) \subseteq \tree^0\{y'(\phi,\nu)\} \} .
\end{align*}

\subsection{Proof of Proposition~\ref{prop_quadratic}}
\label{appendix_branch_computation_proofs}

Given a region $R$, let $\mathbb{1}(R)$ denote the vector in $\mathbb{R}^n$ such that the $i$th element is $1_{(x_i \in R)}$. Let 
$\mathcal{P}_{\mathbb{1}{(R)}} = \mathbb{1}{(R)} \left\{ \mathbb{1}{(R)}^\T \mathbb{1}{(R)} \right\}^{-1}\mathbb{1}{(R)}^\T$  denote the orthogonal projection matrix onto the vector $\mathbb{1}{(R)}$.
\begin{lemma}
\label{lemma_matrixM}
For any region $R$,
$
\textsc{gain}_{R}(y,j,s) = y^\T M_{R,j,s} y,
$ 
where 
\begin{equation}
\label{eq_mdef}
M_{R,j,s} = \mathcal{P}_{\mathbb{1}({R \cap \chi_{j,s,1})}} + \mathcal{P}_{\mathbb{1}({R \cap \chi_{j,s,0})}} - \mathcal{P}_{\mathbb{1}{(R)}}.
\end{equation}
Furthermore, the matrix $M_{R,j,s}$ is positive semidefinite. 
\end{lemma}
\begin{proof}
For any region $R$,
$
\sum_{i \in R} (y_i - \bar{y}_R)^2 
= \sum_{i\in R} y_i^2 - y^\T \mathcal{P}_{\mathbb{1}_{(R)}} y.
$ Thus, from \eqref{eq:gain},
\small 
\begin{align*}
\textsc{gain}_{R}(y,j,s) &= \sum_{i \in R} (y_i - \bar{y}_{R})^2 - \sum_{i \in R \cap \chi_{j,s,1}} (y_i - \bar{y}_{R \cap \chi_{j,s,1}})^2 - \sum_{i \in R \cap \chi_{j,s,0}} (y_i - \bar{y}_{R \cap \chi_{j,s,0}})^2  \\
&=\sum_{i \in R} y_i^2-y^\T\mathcal{P}_{\mathbb{1}_{(R)}}y -\sum_{i \in R \cap \chi_{j,s,1}} y_i^2 + y^\T\mathcal{P}_{\mathbb{1}(R \cap \chi_{j,s,1} )}y  - \sum_{i \in R \cap \chi_{j,s,0}} y_i^2
+ y^\T\mathcal{P}_{\mathbb{1}(R \cap \chi_{j,s,0})}y  \\
&= {y^\T \left\{\mathcal{P}_{\mathbb{1}({R \cap \chi_{j,s,1})}} + \mathcal{P}_{\mathbb{1}({R \cap \chi_{j,s,0})}}-\mathcal{P}_{\mathbb{1}{(R)}} \right\} y = y^\T M_{R,j,s} y .}
\end{align*}
\normalsize
To see that $M_{R,j,s}$ is positive semidefinite, observe that, for any vector $v$, 
\small
\begin{align}
\nonumber v^\T M_{R,j,s} v = \textsc{gain}_R(v,j,s) &= \sum_{i \in R} (v_i - \bar{v}_R)^2 -  \min_{a_1, a_2} \left\{ \sum_{i \in R \cap \chi_{j,s,1}} (v_i - a_1)^2 + \sum_{i \in R \cap \chi_{j,s,0}} (v_i - a_2)^2\right\} \\
\nonumber &\geq \sum_{i \in R} (v_i - \bar{v}_R)^2 -   \left\{\sum_{i \in R \cap \chi_{j,s,1}} (v_i - \bar{v}_R)^2 + \sum_{i \in R \cap \chi_{j,s,0}} (v_i - \bar{v}_R)^2\right\} = 0. 
\end{align}
\end{proof}
\normalsize

It follows from  Lemma~\ref{lemma_matrixM} that we can express each set $S_{l,j,s}$ from Proposition~\ref{prop_Sisintersection} as
\begin{align}
\nonumber  S_{l,j,s} 
&= \left\{ \phi : \textsc{gain}_{R^{(l-1)}}\{y'(\phi,\nu), j,s\} \leq \textsc{gain}_{R^{(l-1)}}\{y'(\phi,\nu), j_l, s_l\} \right\} \\
\label{eq_newSnotation}
&= \left\{ \phi : y'(\phi,\nu)^\T M_{R^{(l-1)}, j,s} y'(\phi,\nu) \leq y'(\phi,\nu)^\T M_{R^{(l-1)}, j_l,s_l} y'(\phi,\nu)  \right\}.
\end{align}
We now use \eqref{eq_newSnotation} to prove the first statement of  Proposition~\ref{prop_quadratic}. 
\begin{lemma}
\label{lemma_SquadraticSIMPLE}
Each set $S_{l,j,s}$ is defined by a quadratic inequality in $\phi$. 
\end{lemma}
\begin{proof}
The definition of $y'(\phi,\nu)$ in \eqref{def_S} implies that
\begin{align}
    y'(\phi,\nu)^\T M_{R,j,s}y'(\phi,\nu)
     \nonumber&= \left( \mathcal{P}_\nu^\perp y + \frac{\nu \phi}{\|\nu\|_2^2} \right)^\T M_{R,j,s} \left( \mathcal{P}_\nu^\perp y + \frac{\nu \phi}{\|\nu\|_2^2} \right)  \\
     \nonumber &= \frac{\nu^\T M_{R,j,s} \nu}{\|\nu\|_2^4} \phi^2 + \frac{2 \nu^\T M_{R,j,s} \mathcal{P}_\nu^\perp y}{\|\nu\|_2^2} \phi + y^\T P_{\nu}^\perp M_{R,j,s} P_\nu^\perp y \\
     \label{eq_abc}
     &\equiv a(R,j,s) \phi^2 + b(R,j,s) \phi + c(R,j,s).
    \end{align}
   Therefore, by \eqref{eq_newSnotation},
  \begin{align}
  \label{eq_mainquadratic}
\nonumber S_{l,j,s} = \bigg\{ \phi : &\left[ a\left\{R^{(l-1)},j,s\right\} -a\left\{R^{(l-1)},j_l,s_l \right\}\right]\phi^2  + \left[ b\left\{R^{(l-1)},j,s\right\} -b\left\{R^{(l-1)},j_l,s_l \right\} \right]\phi  \\
&+\left[ c\left\{R^{(l-1)},j,s\right\} -c\left\{R^{(l-1)},j_l,s_l \right\}\right] \leq 0 \bigg\}.
  \end{align}
\end{proof}

Proposition~\ref{prop_Sisintersection} indicates that to compute $S_{grow}(\mathcal{B},\nu)$ from \eqref{eq:grow}, we need to compute the coefficients of the quadratic for each $S_{l,j,s}$, where $l=1, \ldots, L$, $j=1,\ldots,p$, and $s=1,\ldots,n-1$. 

\begin{lemma}
\label{lemma_compcost1}
We can compute the coefficients $a\left\{R^{(l-1)},j,s\right\}, b\left\{R^{(l-1)},j,s\right\}$ and $c\left\{R^{(l-1)},j,s\right\}$, defined in Lemma~\ref{lemma_SquadraticSIMPLE}, for $l=1, \ldots, L$, $j=1,\ldots,p$, and $s=1,\ldots,n-1$, in $O\left\{nplog(n)+npL\right\}$ operations.  
\end{lemma}
\begin{proof}

Using the definitions in Lemmas~\ref{lemma_matrixM} and \ref{lemma_SquadraticSIMPLE} and algebra, we have that
\footnotesize
\begin{equation}
\label{eq_adef}
\|\nu\|_2^4 a\{R^{(l-1)},j,s\} = \frac{\left[\nu^\T \mathbb{1}\{ R^{(l-1)} \cap \chi_{j,s,1}\}\right]^2}{\sum_{i=1}^n 1_{\{i \in {R^{(l-1)} \cap \chi_{j,s,1}}\}}} + \frac{\left[\nu^\T \mathbb{1}\{R^{(l-1)} \cap \chi_{j,s,0}\}\right]^2}{\sum_{i=1}^n 1_{\{i \in {R^{(l-1)} \cap \chi_{j,s,0}}\}}} - \frac{\left[\nu^\T \mathbb{1}\{R^{(l-1)}\}\right]^2}{\sum_{i=1}^n 1_{\{i \in R^{(l-1)}\}}},
\end{equation}
\begin{align}
\label{eq_bdef}
\frac{1}{2}\|\nu\|_2^2 b\{R^{(l-1)},j,s\} &= \frac{
\nu^\T \mathbb{1}\{R^{(l-1)} \cap \chi_{j,s,1}\}\left(\mathcal{P}_\nu^\perp y \right)^\T \mathbb{1}\{R^{(l-1)} \cap \chi_{j,s,1}\}}{\sum_{i=1}^n 1_{\{i \in {R^{(l-1)} \cap \chi_{j,s,1}}\}}} +  \\
\nonumber &\frac{\nu^\T \mathbb{1}\{R^{(l-1)} \cap \chi_{j,s,0}\}\left(\mathcal{P}_\nu^\perp y \right)^\T \mathbb{1}\{R^{(l-1)} \cap \chi_{j,s,0}\}}{\sum_{i=1}^n 1_{\{i \in {R^{(l-1)} \cap \chi_{j,s,0}}\}}} -
\frac{\nu^\T \mathbb{1}\{R^{(l-1)}\}\left(\mathcal{P}_\nu^\perp y \right)^\T \mathbb{1}\{R^{(l-1)}\}}{\sum_{i=1}^n 1_{\{i \in {R^{(l-1)}}\}}},   \end{align}
\begin{equation}
\label{eq_cdef}
c\{R^{(l-1)},j,s\} = \frac{\left[(\mathcal{P}_\nu^\perp y)^\T \mathbb{1}\{{R^{(l-1)} \cap \chi_{j,s,1}\}}\right]^2}{\sum_{i=1}^n 1_{\{i \in {R^{(l-1)} \cap \chi_{j,s,1}}\}}} + \frac{\left[(\mathcal{P}_\nu^\perp y)^\T\mathbb{1}\{ R^{(l-1)} \cap \chi_{j,s,0}\}\right]^2}{\sum_{i=1}^n 1_{\{ i \in {R^{(l-1)} \cap \chi_{j,s,0}}\}}} - \frac{\left[(\mathcal{P}_\nu^\perp y)^\T \mathbb{1}\{R^{(l-1)}\}\right]^2}{\sum_{i=1}^n 1_{\{i \in {R^{(l-1)}}\}}}.
\end{equation}
\normalsize

We compute the scalar $\|\nu\|_2^2$ and the vector $\mathcal{P}_\nu^\perp y$ in $O(n)$ operations once at the start of the algorithm. We also sort each feature in $O\left[nlog(n)\right]$ operations per feature.  We will now show that for the $l$th level and the $j$th feature, we can compute 
$a\{R^{(l-1)},j,s\}$,
$b\{R^{(l-1)},j,s\},$ and $c\{R^{(l-1)},j,s\}$ for all $n-1$ values of $s$ in $O(n)$ operations. 

The index $s$ appears in \eqref{eq_adef}--\eqref{eq_cdef} only through $\sum_{i=1}^n 1_{\{i \in {R^{(l-1)} \cap \chi_{j,s,1}}\}}$, $\sum_{i=1}^n 1_{\{i \in {R^{(l-1)} \cap \chi_{j,s,0}}\}}$, and through inner products of vectors $\nu$ and $\mathcal{P}_\nu^\perp y$ with indicator vectors $\mathbb{1}\left\{R^{(l-1)} \cap \chi_{j,s,1}\right\}$ and $\mathbb{1}\left\{R^{(l-1)} \cap \chi_{j,s,0}\right\}$. For simplicity, we assume that covariate $x_j$ is continuous, and thus the order statistics are unique. 

Let $m_1$ be the index corresponding to the smallest value of $x_j$. Then $\nu^\T \mathbb{1}\left\{R^{(l-1)} \cap \chi_{j,1,1}\right\} = \nu_{m_1}$ if observation $m_1$ is in $R^{(l-1)}$, and is $0$ otherwise. Similarly,
 $\sum_{i=1}^n 1_{\{i \in {R^{(l-1)} \cap \chi_{j,1,1}}\}} = 1$ if observation $m_1$ is in $R^{(l-1)}$, and is $0$ otherwise. Next, let $m_2$ be the index corresponding to the second smallest value of $x_j$. Then $\nu^\T \mathbb{1}\left\{R^{(l-1)} \cap \chi_{j,2,1}\right\} = \mathbb{1}\left\{R^{(l-1)} \cap \chi_{j,1,1}\right\} + \nu_{m_2}$ if observation $m_2$ is in $R^{(l-1)}$, and is equal to $\mathbb{1}\left\{R^{(l-1)} \cap \chi_{j,1,1}\right\}$ otherwise. 
 We compute $\sum_{i=1}^n 1_{\{i \in {R^{(l-1)} \cap \chi_{j,2,1}}\}}$ in the same manner. Each update is done in constant time. Continuing in this manner, computing the full set of $n-1$ quantities $\nu^\T \mathbb{1}\left\{R^{(l-1)} \cap \chi_{j,s,1}\right\}$ and $\sum_{i=1}^n 1_{\{i \in {R^{(l-1)} \cap \chi_{j,s,1}}\}}$ for $s=1,\ldots,n-1$ requires a single forward pass through the sorted values of $x_j$, which takes $O(n)$ operations. The same ideas can be applied to compute $\left(\mathcal{P}_\nu^\perp y\right)^\T \mathbb{1}\left\{R^{(l-1)} \cap \chi_{j,1,1}\right\}$, 
$\nu^\T \mathbb{1}\left\{R^{(l-1)} \cap \chi_{j,s,0} \right\}$, $\left(\mathcal{P}_\nu^\perp y \right)^\T \mathbb{1}\left\{R^{(l-1)} \cap \chi_{j,s,0} \right\}$, and $\sum_{i=1}^n 1_{\{i \in {R^{(l-1)} \cap \chi_{j,s,0}}\}}$ using constant time updates for each value of $s$. 

Thus, we can obtain all components of coefficients $a\{R^{(l-1)},j,s\}$,
$b\{R^{(l-1)},j,s\},$ and $c\{R^{(l-1)},j,s\}$ for a fixed $j$ and $l$, and for all $s=1, \ldots, n-1$, in $O(n)$ operations. These scalar components can be combined to obtain the coefficients in $O(n)$ operations. Therefore, given the sorted features, we compute the $(n-1)pL$ coefficients in $O(npL)$ operations. 
\end{proof}

Once the coefficients on the right hand side of \eqref{eq_mainquadratic} have been computed, we can compute $S_{l,j,s}$ in constant time via the quadratic equation: it is either a single interval or the union of two intervals. Finally, in general we can intersect $(n-1)pL$ intervals in $O\left\{npL \times log(npL)\right\}$ operations \citep{bourgon2009intervals}. The final claim of Proposition~\ref{prop_quadratic} involves the special case where $\nu=\nu_{sib}$ and $\mathcal{B} = \branch\{R_A, \tree^\lambda(y)\}$. 
\begin{lemma}
\label{lemma_nusibfaster}
Suppose that $\nu= \nu_{sib}$ from \eqref{eq_nusib} and $\mathcal{B} = \branch\{R_A, \tree^\lambda(y)\}$. (i) 
If $l < L$, then for all $j$ and $s$, there exist $a, b \in [-\infty,\infty]$ such that $a \leq \nu^\T y \leq b$ and $S_{l,j,s} = (a, b)$. (ii) If $l=L$, then for all $j$ and $s$, there exist $c, d \in \mathbb{R}$ such that $c \leq 0 \leq d$ and $S_{l,j, s} = (-\infty, c] \cup [d,\infty)$. (iii) We can intersect all $(n-1)pL$ sets of the form $S_{l,j,s}$ in $O(npL)$ operations. 
\end{lemma}
\begin{proof}
This proof relies on the form of $S_{l,j,s}$ given in \eqref{eq_mainquadratic}.
To prove (i), note that when $l < L$,
\begin{align*}
\| \nu_{sib} \|_2^4 \left[ a\left\{R^{(l-1)},j,s\right\} - a\left\{R^{(l-1)},j_l,s_l\right\} \right] &= \nu_{sib}^\T M_{R^{(l-1)},j,s} \nu_{sib} - \nu_{sib}^\T M_{R^{(l-1)},j_l,s_l} \nu_{sib} \\
&= \nu_{sib}^\T M_{R^{(l-1)},j,s} \nu_{sib} \geq 0.
\end{align*}
\normalsize
The first equality follows directly from the definition of $a(R,j,s)$ in \eqref{eq_abc}. To see why the second equality holds, 
observe that $R_A\cup R_B \subseteq R^{(l-1)}$, and without loss of generality assume that $R_A\cup R_B \subseteq R^{(l-1)} \cap \chi_{j_l, s_l, 1}$. Recall that the $i$th element of $\nu_{sib}$ is non-zero if and only if $i \in R_A \cup R_B$, and that the non-zero elements of $\nu_{sib}$ sum to $0$. Thus, $\mathbb{1}\left\{R^{(l-1)}\right\}^\T \nu_{sib} = 0$ and $\mathbb{1}\left\{R^{(l-1)} \cap \chi_{j_l, s_l, 1}\right\}^\T\nu_{sib} = 0$. Furthermore, the supports of
$R^{(l-1)} \cap \chi_{j_l, s_l, 0}$ and $\nu_{sib}$ are non-overlapping, and so $\mathbb{1}\left\{R^{(l-1)} \cap \chi_{j_l, s_l, 0}\right\}^\T \nu_{sib} = 0$. Thus, 
\small
$$
M_{R^{(l-1)},j_l,s_l} \nu_{sib} = \left[  \mathcal{P}_{\mathbb{1}\left\{R^{(l-1)} \cap \chi_{j_l, s_l, 1}\right\}}+ \mathcal{P}_{\mathbb{1}\left\{R^{(l-1)} \cap \chi_{j_l, s_l, 0}\right\}} -  \mathcal{P}_{\mathbb{1}\left\{R^{(l-1)}\right\}} \right]\nu_{sib} = 0.
$$
\normalsize
The final inequality follows because $M_{R^{(l-1)},j,s}$ is positive semidefinite (Lemma~\ref{lemma_matrixM}). 

Thus, when $l < L$, $S_{l, j, s}$ is defined in \eqref{eq_mainquadratic} by a quadratic inequality with a non-negative quadratic coefficient. Thus, $S_{l,j,s}$ must be a single interval of the form $(a,b)$. Furthermore, since $\mathcal{B} = \branch\{R_A, \tree^\lambda(y)\}$, we know that $\mathcal{R}(\mathcal{B}) \subseteq \tree^0(y) = \tree^0\{y'(\nu^\T y,\nu)\}$. Therefore,  
$\nu^\T y \in S_{grow}\left(\mathcal{B},\nu_{sib}\right) = \cap_{l=1}^L \cap_{j=1}^ p \cap_{s=1}^{n-1} S_{l, j, s}$, and so we conclude $a \leq \nu^\T y \leq b$. This completes the proof of (i).

To prove (ii), we first prove that when $l=L$ the quadratic equation in $\phi$ defined in \eqref{eq_mainquadratic} has a non-positive quadratic coefficient.
To see this, note that
\begin{align}
\nonumber
& \| \nu_{sib} \|_2^4 \left[ a\left\{ R^{(L-1)},j,s\right\} - a\left\{R^{(L-1)},j_L,s_L\right\} \right] = \nu_{sib}^\T M_{R^{(L-1)},j,s} \nu_{sib} - \nu_{sib}^\T M_{R^{(L-1)},j_L,s_L} \nu_{sib} \\
\nonumber
&= \nu_{sib}^\T \left[
\mathcal{P}_{\mathbb{1}\{{R^{(L-1)} \cap \chi_{j,s,1}\}}} +\mathcal{P}_{\mathbb{1}\{{R^{(L-1)} \cap \chi_{j,s,0}\}}} \right] \nu_{sib} 
- \nu_{sib}^\T \left[
\mathcal{P}_{\mathbb{1}\{{R^{(L-1)} \cap \chi_{j_L,s_L,1}\}}} +\mathcal{P}_{\mathbb{1}\{{R^{(L-1)} \cap \chi_{j_L,s_L,0}\}}} \right]\nu_{sib}  \\
\nonumber
&= \nu_{sib}^\T \left[\mathcal{P}_{\mathbb{1}\{{R^{(L-1)} \cap \chi_{j,s,1}\}}} +\mathcal{P}_{\mathbb{1}\{{R^{(L-1)} \cap \chi_{j,s,0}\}}} \right] \nu_{sib} - \nu_{sib}^\T \nu_{sib}  \\
\label{eq_ineq}
&= \left\| \left[ \mathcal{P}_{\mathbb{1}\{{R^{(L-1)} \cap \chi_{j,s,1}\}}} +\mathcal{P}_{\mathbb{1}\{{R^{(L-1)} \cap \chi_{j,s,0}\}}} \right] \nu_{sib} \right\|_2^2 - \| \nu_{sib} \|_2^2 \leq 0.
\end{align}
The first equality follows from \eqref{eq_abc}. The second follows from the definition of $M_{R,j,s}$ given in \eqref{eq_mdef} and from the fact that $\mathcal{P}_{\mathbb{1}\{{R^{(L-1)}\}}} \nu_{sib} = 0$ because $\mathbb{1}\{{R^{(L-1)}\}}^\T \nu_{sib}$ sums up all of the non-zero elements of $\nu_{sib}$, which sum to $0$. The third equality follows because $\nu_{sib}$ lies in $span\left[
\mathbb{1}\left\{R^{(L-1)} \cap \chi_{j_L,s_L,1} \right\},
\mathbb{1}\left\{R^{(L-1)} \cap \chi_{j_L,s_L,0}\right\}
\right]$; projecting it onto this span yields itself. 
Noting that $\left[ \mathcal{P}_{\mathbb{1}\{{R^{(L-1)} \cap \chi_{j,s,1}\}}} +\mathcal{P}_{\mathbb{1}\{{R^{(L-1)} \cap \chi_{j,s,0}\}}} \right]$ is itself a projection matrix, the fourth equality follows from the idempotence of projection matrices, and the inequality follows from the fact that $\|\nu_{sib}\|_2 \geq \|Q \nu_{sib}\|_2$ for any projection matrix $Q$. Thus, when $l=L$, the quadratic that defines $S_{l,j,s}$ has a non-positive quadratic coefficient. 

Equality is attained in \eqref{eq_ineq} if and only if $\nu_{sib} \in span\left[\mathbb{1}\{{R^{(L-1)} \cap \chi_{j,s,1}\}},\mathbb{1}\{{R^{(L-1)} \cap \chi_{j,s,0}\}}\right]$. This can only happen if splitting $R^{(L-1)}$ on $j, s$ yields an identical partition of the data to splitting on $j_L, s_L$. If this is the case, then $S_{L,j,s} = (-\infty,0] \cup [0, \infty)$ from the definition of $S_{l,j,s}$ in Proposition~\ref{prop_Sisintersection}, and so (ii) is satisfied with $c=d=0$. 

We now proceed to the setting where the inequality in \eqref{eq_ineq} is strict. In this case, \eqref{eq_mainquadratic} implies that 
$S_{L,j,s} = (-\infty,c] \cup [d, \infty)$ for $c \leq d$ and $c,d \in \mathbb{R}$. To complete the proof of (ii), we must argue that $c \leq 0$ and $d \geq 0$. Recall that the quadratic in \eqref{eq_newSnotation} has the form 
$\textsc{gain}_{R^{(L-1)}}\{y'(\phi,\nu), j,s\} - \textsc{gain}_{R^{(L-1)}}\{y'(\phi,\nu), j_L,s_L\}$. When $\phi = 0$, $\textsc{gain}_{R^{(L-1)}}\{y'(\phi,\nu), j_L,s_L\} = 0$, because $\phi=0$ eliminates the contrast between $R_A$ and $R_B$, so that the split on $j_L,s_L$ provides zero gain. So, when $\phi = 0$, the quadratic evaluates to $\textsc{gain}_{R^{(L-1)}}\{y'(\phi,\nu), j,s\}$, which is non-negative by Lemma~\ref{lemma_matrixM}. Thus, $S_{l,j,s}$ is defined by a downward facing quadratic that is non-negative when $\phi=0$, and so the set $S_{l,j,s}$ has the form $(-\infty,c] \cup [d, \infty)$ for $c \leq 0 \leq d$.

To prove (iii), observe that (i) implies that 
$\cap_{l=1}^{L-1} \cap_{j=1}^p  \cap_{s=1}^{n-1} S_{l,j,s} = (a_{max}, b_{min})$, where $a_{max}$ is the maximum over all of the $a$'s, and $b_{min}$ is the minimum over all of the $b$'s. This can be computed in $np(L-1)$ steps. Furthermore, (ii) implies that $\cap_{j=1}^p  \cap_{s=1}^{n-1} S_{L,j,s} = (-\infty,c_{min}] \cap [d_{max}, \infty)$, where 
$c_{min}$ and $d_{max}$ are the minimum over all of the $c$'s and the maximum over all the $d$'s, respectively. This can be computed in $np$ steps. Thus, we can compute $ \cap_{l=1}^{L} \cap_{j=1}^p  \cap_{s=1}^{n-1} S_{l,j,s}$ in $O(npL)$ operations. 
\end{proof}

\subsection{Proof of Proposition \ref{prop_s.pruning}}
\label{appendix_pruningproof}

To prove Proposition~\ref{prop_s.pruning}, we first propose a particular method of constructing an example of $\tree(\mathcal{B},\nu, \lambda)$. We then show that 
\eqref{eq:mainpruning} holds for this particular choice for $\tree(\mathcal{B},\nu, \lambda)$. We conclude by evaluating the computational cost of computing such an example of $\tree(\mathcal{B},\nu, \lambda)$, and by arguing that in the special case where $\mathcal{R}(\mathcal{B}) \in \tree^\lambda(y)$, our example is equal to $\tree^\lambda(y)$.

When Algorithm~\ref{alg_pruning} is called with parameters $\tree,y,\lambda,\mathcal{O}$, where $\mathcal{O}$ is a bottom-up ordering of the $K$ nodes in $\tree$, it computes a sequence of intermediate trees, $\tree_0, \ldots, \tree_K$. We use the notation $\tree_k(\tree,y,\lambda,\mathcal{O})$, for $k=0,\ldots,K$, to denote the $k$th of these intermediate trees. The following lemma helps build up to our proposed example of $\tree(\mathcal{B},\nu, \lambda)$. 

\begin{lemma}
\label{lemma-firstsequence}
Let $\phi_1 \in S_{grow}(\mathcal{B},\nu)$ and $\phi_2 \in S_{grow}(\mathcal{B},\nu)$. Then $\tree^0\{y'(\phi_1,\nu)\}=\tree^0\{y'(\phi_2,\nu)\}$. Let $\mathcal{O}$ be a bottom-up ordering of the $K$ regions in $\tree^0\{y'(\phi_1,\nu)\}$ such that the last $L$ regions in the ordering are $R^{(L-1)}, \ldots, R^{(0)}$. Then $\tree_{K-L}[\tree^0\{y'(\phi_1,\nu)\},y'(\phi_1,\nu),\lambda,\mathcal{O}]=\tree_{K-L}[\tree^0\{y'(\phi_2,\nu) \},y'(\phi_2,\nu),\lambda,\mathcal{O}]$. 
\end{lemma}

\begin{proof}
We first prove that $\tree^0\{y'(\phi_1,\nu)\} \subseteq \tree^0\{y'(\phi_2,\nu)\}$, where $\phi_1,\phi_2 \in S_{grow}(\mathcal{B}, \nu)$. The fact that $\phi_1 \in S_{grow}(\mathcal{B}, \nu)$ and $\phi_2 \in S_{grow}(\mathcal{B}, \nu)$
implies two properties: 
\begin{list}{}{}
\item \underline{\emph{Property 1:}} $R^{(l)} \in \tree^0\{y'(\phi_1,\nu)\}$ and $R^{(l)} \in \tree^0\{y'(\phi_2,\nu)\}$ for $l \in \{0,\ldots,L\}$ by the definition of  $S_{grow}(\mathcal{B}, \nu)$.  
\item \underline{\emph{Property 2:}} $R^{(l)}_{sib} \in \tree^0\{y'(\phi_1,\nu)\}$ and $R_{sib}^{(l)} \in \tree^0\{y'(\phi_2,\nu)\}$ for $l \in \{1,\ldots,L\}$, where $R^{(l)}_{sib} \equiv R^{(l-1)} \cap \chi_{j_{l}, s_l, 1-e_l}$. This follows from Property 1 and Definition~\ref{def_tree}.
\end{list}

Suppose that $R \in \tree^0\{y'(\phi_1,\nu)\}$. Then $R$ must belong to one of these three cases, illustrated in Figure~\ref{fig_proofhelp}(a):
\begin{list}{}{}
\item \underline{\emph{Case 1: $\exists \ l \in \{0, \ldots L\}$ such that $R=R^{(l)}$.}} By Property 1, $R \in \tree^0\{y'(\phi_2,\nu)\}$.
\item \underline{\emph{Case 2: $\exists \ l \in \{1,\ldots,L\}$ such that  $R = R^{(l)}_{sib}$.}} By Property 2, $R \in  \tree^0\{y'(\phi_2,\nu)\}$. 
\item \underline{\emph{Case 3: $R \in \desc[R', \tree^0\{y'(\phi_1, \nu)\}]$, where either  $R' = R^{(l)}_{sib}$ for  some $l \in \{1,\ldots,L\}$, }}
\underline{\emph{or else  $R' = R^{(L)}$.}} By Properties 1 and 2, $R' \in  
\tree^0\{y'(\phi_2, \nu)\}$. Condition~\ref{cond_nuprop} ensures that, for all $i \in R'$ and for some constants $c$ and $d$,
$$
\{y'(\phi_2,\nu)\}_i = 
\begin{cases}
 	\{y'(\phi_1,\nu)\}_i &\text{ if } R' = R_{sib}^{(l)} \text{ for some  } l \in \{1,\ldots,L-1\}, \\ 
    \{y'(\phi_1,\nu)\}_i + c &\text{ if } R' = R_{sib}^{(L)}, \\ 
   \{y'(\phi_1,\nu)\}_i + d &\text{ if } R' = R^{(L)}. \\ 
 \end{cases}
 $$
 As constant shifts preserve within-node sums of squared errors, in each of these three scenarios, $\desc[R', \tree^0\{y'(\phi_1, \nu)\}]= \desc[R', \tree^0\{y'(\phi_2, \nu)\}]$. Thus, $R \in \tree^0\{y'(\phi_2,\nu)\}$.
\end{list}

Thus, if $R \in \tree^0\{y'(\phi_1,\nu)\}$, then $R \in \tree^0\{y'(\phi_2,\nu)\}$. This completes the argument that $\tree^0\{y'(\phi_1,\nu)\} \subseteq \tree^0\{y'(\phi_2,\nu)\}$. Swapping the roles of $\phi_1$ and $\phi_2$ in this argument, we see that 
$\tree^0\{y'(\phi_2,\nu)\} \subseteq \tree^0\{y'(\phi_1,\nu)\}$. This concludes the proof that $\tree^0\{y'(\phi_1,\nu)\} = \tree^0\{y'(\phi_2,\nu)\}$.

Because $\tree^0\{y'(\phi_1,\nu)\} = \tree^0\{y'(\phi_2,\nu)\}$, it follows that any bottom-up ordering of the regions in $\tree^0\{y'(\phi_1,\nu)\}$ is also a bottom-up ordering for the regions in $\tree^0\{y'(\phi_2,\nu)\}$. We next prove by induction that, if we choose a bottom-up ordering $\mathcal{O}$ that places the regions in $\mathcal{R}(\mathcal{B})$ at the end of the ordering, then 
\small
\begin{equation}
\label{eq_inductionpruning}
\tree_{k}[\tree^0\{y'(\phi_1,\nu)\}, y'(\phi_1,\nu), \lambda, \mathcal{O}] = \tree_{k}[\tree^0\{y'(\phi_2,\nu)\}, y'(\phi_2,\nu), \lambda, \mathcal{O}],	
\end{equation}
\normalsize
for $k=0,\ldots,K-L$. It follows immediately from Algorithm~\ref{alg_pruning} and the argument above that 
$\tree_{0}[\tree^0\{y'(\phi_1,\nu)\},y'(\phi_1,\nu),\lambda,\mathcal{O}]=\tree_{0}[\tree^0\{y'(\phi_2,\nu)\},y'(\phi_2,\nu),\lambda,\mathcal{O}]$. Next, suppose that for some $k \in \{1,\ldots,K-L\}$, 
\small
\begin{equation}
\label{eq_induction_kminus1}
\tree_{k-1}[\tree^0\{y'(\phi_1,\nu)\}, y'(\phi_1,\nu), \lambda, \mathcal{O}] = \tree_{k-1}[\tree^0\{y'(\phi_2,\nu)\}, y'(\phi_2,\nu), \lambda, \mathcal{O}],
\end{equation}
\normalsize
and denote this tree with $\tree_{k-1}$ for brevity.
We must prove 
that \eqref{eq_inductionpruning} holds. 
Let $R$ be the $k$th region in $\mathcal{O}$ and recall the assumption that the last $L$ regions in $\mathcal{O}$ are $\{R^{(L-1)}, \ldots, R^{(0)}\}$. Since $k \leq K-L$, this implies that $R \not\in \{R^{(L-1)}, \ldots, R^{(0)}\}$.
This means that either $R \in \desc\left[R_{sib}^{(l)}, \tree^0\{y'(\phi_1,\nu)\}\right]$ for $l \in \{1,\ldots,L\}$ or $R \in \desc\left[R^{(L)}, \tree^0\{y'(\phi_1,\nu)\}\right]$, meaning that $R$ is a black region in Figure \ref{fig_proofhelp}(b). From Condition~\ref{cond_nuprop},
$$
\{y'(\phi_2,\nu)\}_i = 
\begin{cases}
 	\{y'(\phi_1,\nu)\}_i &\text{ if } R \in \desc\left[R_{sib}^{(l)}, \tree^0\{y'(\phi_1,\nu)\}\right] \text{ for } l \in \{1,\ldots,L-1\}, \\ 
    \{y'(\phi_1,\nu)\}_i + c &\text{ if } R \in \desc\left[R_{sib}^{(L)}, \tree^0\{y'(\phi_1,\nu)\}\right], \\ 
   \{y'(\phi_1,\nu)\}_i + d &\text{ if }  R \in \desc\left[R^{(L)},\tree^0\{y'(\phi_1,\nu)\}\right]. \\ 
 \end{cases}
 $$
 In any of the three cases illustrated in Figure~\ref{fig_proofhelp}, for $g(\cdot)$ defined in \eqref{def_g}, $g\{R, \tree_{k-1},y'(\phi_1,\nu)\} = g\{R, \tree_{k-1},y'(\phi_2,\nu)\}$. Combining this with 
\eqref{eq_induction_kminus1} and Step 2(b) of Algorithm~\ref{alg_pruning} yields \eqref{eq_inductionpruning}. This completes the proof by induction. 

\begin{figure}[!h]
\centering
\begin{subfigure}{0.45\textwidth}
\subcaption{}
\includegraphics[width=\textwidth]{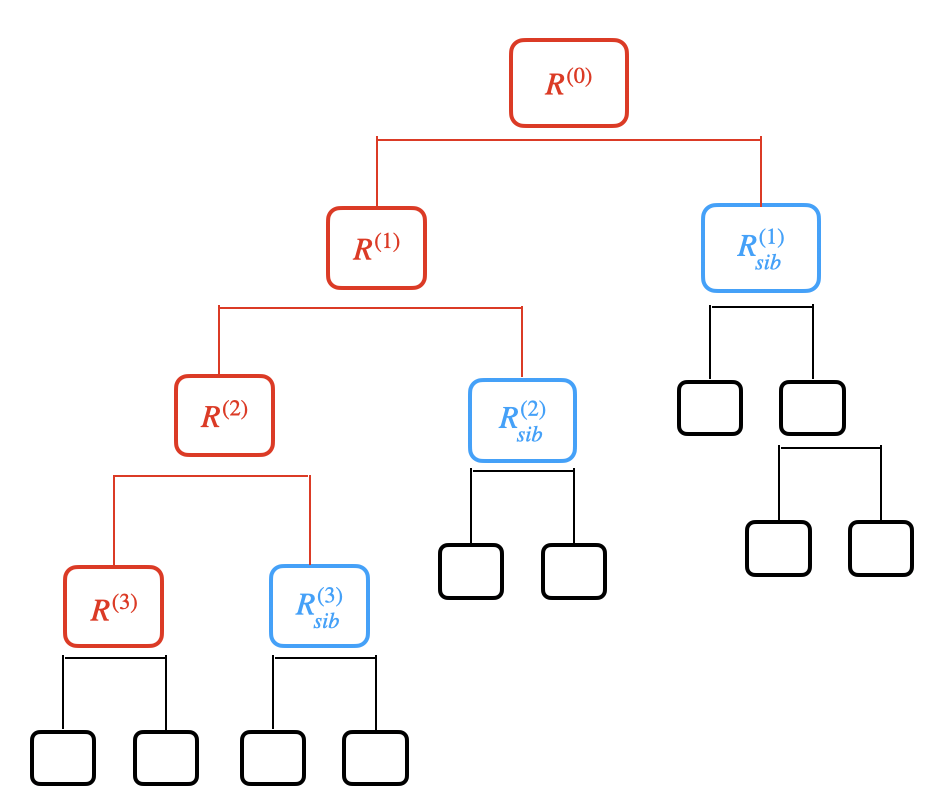}
\end{subfigure}	\hspace{7mm}
\begin{subfigure}{0.45\textwidth}
\subcaption{}
\includegraphics[width=\textwidth]{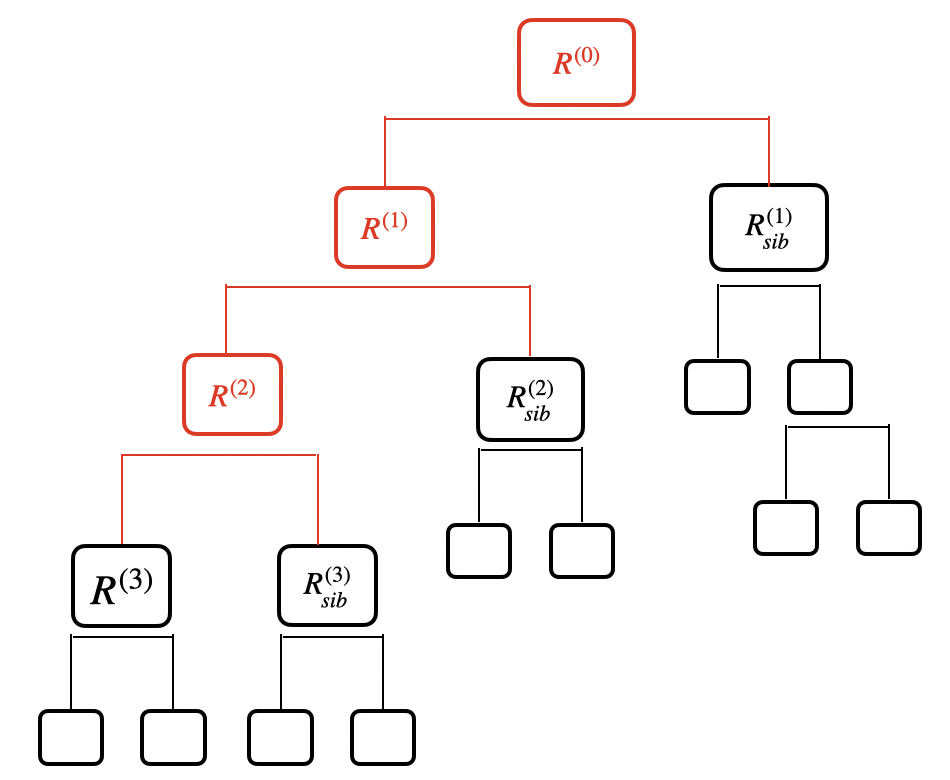}	
\end{subfigure}
\caption{\emph{(a).}  An illustration of Case 1 (red), Case 2 (blue), and Case 3 (black)  for a region $R \in \tree^0\{y'(\phi_1,\nu)\}$ in the base case of the proof of Lemma~\ref{lemma-firstsequence}, where $\mathcal{R}(\mathcal{B}) = \{R^{(0)}, \ldots, R^{(3)} \}$. \emph{(b.)} The black regions show the possible cases for $R \in \tree_{k-1}$
in the inductive step of the proof of Lemma~\ref{lemma-firstsequence}. 
}
\label{fig_proofhelp} 
\end{figure}
\end{proof}

Since Lemma~\ref{lemma-firstsequence} guarantees that each $\phi \in S_{grow}(\mathcal{B},\nu)$ leads to the same $\tree^0\{y'(\phi,\nu)\}$, we will refer to this tree as $\tree^0$, will let $K$ be the number of regions in this tree, and will let $\mathcal{O}$ be a bottom-up ordering of these regions that places $R^{(L-1)},\ldots,R^{(0)}$ in the last $L$ spots. We will further denote $\tree_{K-L}\{\tree^0, y'(\phi,\nu), \lambda, \mathcal{O}\}$ for any $\phi \in S_{grow}(\mathcal{B},\nu)$ as $\tree_{K-L}$, since Lemma~\ref{lemma-firstsequence} further tells us that this is the same for all $\phi \in S_{grow}(\mathcal{B},\nu)$. In what follows, we argue that if we let $\tree(\mathcal{B},\nu,\lambda) = \tree_{K-L}$, where 
$\tree(\mathcal{B},\nu,\lambda)$ appears in the statement of Proposition~\ref{prop_s.pruning}, then \eqref{eq:mainpruning} holds. In other words, we prove that $\tree_{K-L}$, which always exists and is well-defined, is a valid example of $\tree(\mathcal{B},\nu,\lambda)$. 

Recall from \eqref{eq:sbnu} that 
$S^\lambda\left(\mathcal{B}, \mathcal{\nu}\right) = \left\{ \phi \in S_{grow}(\mathcal{B}, \mathcal{\nu}): R^{(L)} \in \tree^\lambda\{y'(\phi,\nu)\}\right\}$. Lemma~\ref{lemma-firstsequence} says that for $\phi \in S_{grow}(\mathcal{B},\nu)$, we can rewrite $\tree^\lambda\{y'(\phi,\nu)\}$ as $\tree_{K}\{\tree^0, y'(\phi,\nu), \lambda, \mathcal{O}\}$. So we can rewrite $ S^\lambda\left(\mathcal{B}, \mathcal{\nu}\right)$ as 
 \begin{equation}
 \label{eq_sprunenounion}	
 S^\lambda\left(\mathcal{B}, \mathcal{\nu}\right) = \left\{ \phi \in S_{grow}(\mathcal{B}, \mathcal{\nu}): 
R^{(L)} \in \tree_{K}\left\{\tree^0, y'(\phi,\nu), \lambda, \mathcal{O}\right\} \right\}.
 \end{equation}
Furthermore, since 
$R^{(L-1)},\ldots,R^{(0)}$ (all of which are ancestors of $R^{(L)}$) are the last $L$ nodes in the ordering $\mathcal{O}$, we see that $R^{(L)} \in \tree_{K}\{\tree^0, y'(\phi,\nu), \lambda, \mathcal{O}\}$ if and only if 
no pruning occurs during the last $L$ iterations of Step 2 in  Algorithm~\ref{alg_pruning}. This means that we can characterize \eqref{eq_sprunenounion} as 
\small
 \begin{equation}
 \label{eq_notreeupdate}
\left\{ \phi \in S_{grow}(\mathcal{B}, \mathcal{\nu}): 
\tree_{K-L}\{\tree^0, y'(\phi,\nu), \lambda, \mathcal{O}\} = \tree_{K}\left\{\tree^0, y'(\phi,\nu), \lambda, \mathcal{O}\right\} \right\}.
 \end{equation}
 \normalsize 
\normalsize
 Recall that for $k=K-L+1,\ldots,K$, $R^{(K-k)}$ is the $k$th region in $\mathcal{O}$, and is an ancestor of $R^{(L)}$. We next argue that we can rewrite \eqref{eq_notreeupdate} as
\footnotesize
\begin{align}
\label{eq_sprunealmost}
\bigcap_{k=K-L+1}^K 
\biggl\{ &\phi \in S_{grow}(\mathcal{B},\nu) : g\left\{R^{(K-k)}, \tree_{K-L}, y'(\phi,\nu)\right\}  \geq \lambda \biggr\}.
\end{align}
\normalsize

To begin, suppose that $\phi \in \eqref{eq_sprunealmost}$. As we are talking about a particular $\phi$, for $k=0,\ldots,K$ we will suppress the dependence of $\tree_k\left\{\tree^0, y'(\phi,\nu), \lambda, \mathcal{O} \right\}$ on its arguments and denote it with $\tree_k$. The fact that $\phi \in \eqref{eq_sprunealmost}$ means that 
$g\left\{R^{(L-1)}, \tree_{K-L}, y'(\phi,\nu)\right\}  \geq \lambda$, which ensures that no pruning occurs at step $K-L+1$, which in turn ensures that $\tree_{K-L+1}=\tree_{K-L}$. 
Combined with \eqref{eq_sprunealmost}, this implies that $g\left\{R^{(L-2)}, \tree_{K-L}, y'(\phi,\nu)\right\}  = $ \\ $g\left\{R^{(L-2)}, \tree_{K-L+1}, y'(\phi,\nu)\right\} \geq \lambda$, which ensures that no pruning occurs at step $K-L+2$, which in turn ensures that  $\tree_{K-L+2}=\tree_{K-L+1}=\tree_{K-L}$. Proceeding in this manner, by tracing through the last $L$ iterations of Step 2 of Algorithm~\ref{alg_pruning}, we see that $\phi$ satisfies $\tree_K = \tree_{K-L}$, and so $\phi \in \eqref{eq_notreeupdate}$. 

Next suppose that $\phi \not\in $ \eqref{eq_sprunealmost}. Let 
$$
k' = \min_{k \in \{K-L+1,\ldots,K\}} \{ k : g\left\{R^{(K-k)}, \tree_{K-L}, y'(\phi,\nu)\right\}  < \lambda\}.
$$
As $k'$ is a minimum, we know that no pruning occurred during steps $K-L+1,\ldots,k'-1$, and so $ \tree_{k'-1}\left\{\tree^0(y'(\phi,\nu)), y'(\phi,\nu), \lambda, \mathcal{O}\right\}=\tree_{K-L}$. This implies that
\\
$g\left\{R^{(K-k')}, \tree_{K-L}
, y'(\phi,\nu)\right\}  < \lambda$ can be rewritten as \\
$g\left(R^{(K-k')}, \tree_{k'-1}\left[\tree^0\{y'(\phi,\nu)\}, y'(\phi,\nu), \lambda, \mathcal{O}\right], y'(\phi,\nu)\right)  < \lambda$. It then follows from Algorithm~\ref{alg_pruning} that pruning occurs at step $k'$, which means that $\tree_K$ cannot possibly equal $\tree_{K-L}$. Thus, $\phi \not\in \eqref{eq_notreeupdate}$.

Thus, $\phi \in  \eqref{eq_sprunealmost}$ if and only if $\phi \in \eqref{eq_notreeupdate}$. 

Finally, Proposition~\ref{prop_s.pruning} rewrites \eqref{eq_sprunealmost} with the indexing over $k$ changed to an indexing over $l$, and plugging in $\tree(\mathcal{B},\nu,\lambda) = \tree_{K-L}$. Therefore, $\tree_{K-L}$ is a valid example of $\tree(\mathcal{B},\nu,\lambda)$. 

To compute $\tree(\mathcal{B},\nu,\lambda)$, we first select an arbitrary $\phi \in S_{grow}(\mathcal{B},\nu)$. We then apply Algorithm~\ref{alg_growing} to grow $\tree^0\{y'(\phi,\nu)\}$. We create a bottom-up ordering $\mathcal{O}$ of the $K$ nodes in $\tree^0\{y'(\phi,\nu)\}$ such that $R^{(L-1)},\ldots,R^{(0)}$ are at the end. Finally, we apply the first $K-L$ iterations of Algorithm~\ref{alg_pruning} with arguments $\tree^0\{y'(\phi,\nu)\}, y'(\phi,\nu), \lambda,$ and $\mathcal{O}$ to obtain $\tree(\mathcal{B},\nu,\lambda)$. The worst case computational cost of CART (the combined Algorithm~\ref{alg_growing} and Algorithm~\ref{alg_pruning}) is $O(n^2 p)$. 

In the special case that $\mathcal{R}(\mathcal{B}) \subseteq \tree^\lambda(y)$, we have that $\nu^T y \in S_{grow}(\mathcal{B},\nu)$ because $y=y'(\nu^T y, \nu)$ and $\mathcal{R}(\mathcal{B}) \subseteq \tree^\lambda(y) \subseteq \tree^0(y)$. In this case, suppose that we carry out the process described in the previous paragraph by selecting $\phi = \nu^T y$. As $\mathcal{R}(\mathcal{B}) \subseteq \tree^\lambda(y)$, it is clear from Algorithm~\ref{alg_pruning} that no pruning occurs during the last $L$ iterations of Step 2 in Algorithm~\ref{alg_pruning} applied to arguments $\tree^0(y), y, \lambda,$ and $\mathcal{O}$. Therefore, the process from the previous paragraph returns the optimally pruned $\tree^\lambda(y)$. Thus, 
when $\mathcal{R}(\mathcal{B}) \subseteq \tree^\lambda(y)$, we can 
simply plug in $\tree^\lambda(y)$, which has already been built, for $\tree(\mathcal{B},\nu,\lambda)$ in \eqref{eq:mainpruning}. 
 
\subsection{Proof of Proposition~\ref{prop_pruning.efficient}}
\label{proof_pruning.efficient}

We first show that we can express
\begin{equation}
\label{eq_sprunesinglelevel}	
\left\{ \phi : g\left\{R^{(l)}, \tree(\mathcal{B},\nu,\lambda), y'(\phi,\nu)\right\}  \geq \lambda \right\}
\end{equation}
as the solution set of a quadratic inequality in $\phi$ for $l=0,\ldots,L-1$, where $g(\cdot)$ was defined in \eqref{def_g}. Because only the numerator of $g\{R^{(l)},\tree(\mathcal{B},\nu,\lambda),y'(\phi,\nu)\}$ depends on $\phi$, it will be useful to introduce the following concise notation:
\begin{equation}
\label{eq:tstgain}
h(R, \tree, y) \equiv \underset{i \in R}{\sum} (y_i - \bar{y}_{R})^2 - \underset{r \in \term(R, \tree)}{\sum} \ \ \underset{i \in r}{\sum} (y_i - \bar{y}_r)^2.
\end{equation}
We begin with the following lemma.
\begin{lemma}
\label{lemma_decompose}	
Suppose that region $R$ in $\tree$ has children $R \cap \chi_{j,s,0}$ and $R \cap \chi_{j,s,1}$. Then,
$
h(R, \tree,y) = \gain_{R}(y,j,s) + h(R \cap \chi_{j,s,1}, \tree,y)+h(R \cap \chi_{j,s,0}, \tree,y)
$, where $\gain_{R}(y,j,s)$ is defined in \eqref{eq:gain}.
\end{lemma}
\begin{proof}
The result follows from adding and subtracting $\sum_{i \in R \cap \chi_{j,s,1}} (y_i - \bar{y}_{R \cap \chi_{j,s,1}})^2$ and $\sum_{i \in R \cap \chi_{j,s,0}} (y_i - \bar{y}_{R \cap \chi_{j,s,0}})^2$ in \eqref{eq:tstgain} and noting that $\term(R, \tree) = \term(R \cap \chi_{j,s,1}, \tree) \cup \term(R \cap \chi_{j,s,0})$. 
\end{proof}

Recall that $\mathcal{B} = \left( (j_1,s_1,e_1), \ldots, (j_L,s_L,e_L) \right)$ and that $\mathcal{R}(\mathcal{B}) = \{R^{(0)}, R^{(1)}, \ldots, R^{(L)}\}$. Lemma~\ref{lemma_decompose_specific} follows from Lemma~\ref{lemma_decompose} and the fact that, due to the form of the vector $\nu$, there are many regions $R$ for which $h\left\{R,  \tree(\mathcal{B},\nu,\lambda), y'(\phi,\nu) \right\}$ does not depend on  $\phi$.  

\begin{lemma}
\label{lemma_decompose_specific}
For any $\tilde{\phi} \in S_{grow}(\mathcal{B},\nu)$, we can decompose $h\{R^{(l)}, \tree(\mathcal{B},\nu,\lambda),y'(\phi,\nu)\}$ as 
\begin{align*}
h\left\{R^{(l)}, \tree(\mathcal{B},\nu,\lambda),y'(\phi,\nu) \right\} &= \sum_{l'=l+1}^{L} h\left\{R^{(l')}_{sib},\tree(\mathcal{B},\nu,\lambda),y'(\tilde{\phi},\nu)\right\} \\
&+ h\left\{R^{(L)}, \tree(\mathcal{B},\nu,\lambda),y'(\tilde{\phi},\nu)\right\} \\
&+  \sum_{l'=l}^{L-1} \gain_{R^{(l')}}\{y'(\phi,\nu), j_{l'+1}, s_{l'+1}\},
\end{align*}
where $R^{(l)}_{sib}$ is the sibling of $R^{(l)}$ in $\tree(\mathcal{B},\nu,\lambda)$. 
\end{lemma}

\begin{proof}
Repeatedly applying Lemma~\ref{lemma_decompose} yields 
\begin{align} 
h\left\{R^{(l)}, \tree(\mathcal{B},\nu,\lambda),y'(\phi,\nu) \right\} &= \sum_{l'=l+1}^{L} h\left\{R^{(l')}_{sib}, \tree(\mathcal{B},\nu,\lambda),y'(\phi, \nu) \right\}\nonumber \\
&+ h\left\{R^{(L)}, \tree(\mathcal{B},\nu,\lambda),y'(\phi, \nu) \right\} \nonumber \\
&+  \sum_{l'=l}^{L-1} \gain_{R^{(l')}}\{y'(\phi,\nu), j_{l'+1}, s_{l'+1}\}. \label{eq:part1} 
\end{align} 
For $l' = l+1, \ldots, L-1$, 
\small $h\left\{R^{(l')}_{sib}, \tree(\mathcal{B},\nu,\lambda),y'(\phi, \nu)\right\} =h\left\{R^{(l')}_{sib}, \tree(\mathcal{B},\nu,\lambda),y'(\tilde{\phi}, \nu)\right\}$ \normalsize because $R^{(l')}_{sib}$ only contains observations where $[y'(\phi,\nu)]_i = [y'(\tilde{\phi},\nu)]_i$. Similarly, \\ $h\left\{R^{(L)}, \tree(\mathcal{B},\nu,\lambda), y'(\phi, \nu) \right\} = h\left\{R^{(L)},\tree(\mathcal{B},\nu,\lambda),y'(\tilde \phi, \nu)\right\}$ \normalsize and  \small 
$h\left\{R^{(L)}_{sib}, \tree(\mathcal{B},\nu,\lambda), y'(\phi, \nu) \right\} = h\left\{R^{(L)}_{sib},\tree(\mathcal{B},\nu,\lambda),y'(\tilde \phi, \nu)\right\}$ \normalsize
because for $i \in R^{(L)}$ and $i \in R^{(L)}_{sib} $, $\{y'(\phi,\nu)\}_i$ and $\{y'(\tilde{\phi},\nu)\}_i$ only differ by a constant shift. 
Plugging these two facts into \eqref{eq:part1} completes the proof. 
\end{proof}

We can now write \eqref{eq_sprunesinglelevel} as
\footnotesize
\begin{align}
\nonumber
&\left\{ \phi : g\left\{R^{(l)}, \tree(\mathcal{B},\nu,\lambda), y'(\phi,\nu)\right\}  \geq \lambda \right\} \\
\nonumber 
&= \left\{ \phi :
h\left\{R^{(l)}, \tree(\mathcal{B},\nu,\lambda), y'(\phi,\nu)\right\}  \geq \lambda \left[ |\term\{R^{(l)},\tree(\mathcal{B},\nu,\lambda)\}|-1 \right] \right\} \\
\nonumber
&= \bigg\{ \phi : \sum_{l'=l}^{L-1} \gain_{R^{(l')}}\{y'(\phi,\nu), j_{l'+1},s_{l'+1}\} \geq  \lambda \left[|\term\{ R^{(l)}, \tree(\mathcal{B},\nu,\lambda)\}|-1\right] - \\
\nonumber & \hspace{40mm}  \sum_{l'=l+1}^{L} h\left\{R^{(l')}_{sib},\tree(\mathcal{B},\nu,\lambda),y'(\tilde{\phi}, \nu)\right\} - h\left\{R^{(L)}, \tree(\mathcal{B},\nu,\lambda),y'(\tilde{\phi}, \nu) \right\} \bigg\} \\ 
\label{eq_spruningfinalform}
&=\bigg\{ \phi : \sum_{l'=l}^{L-1} \left[ a\left\{R^{(l')}, j_{l'+1},s_{l'+1}\right\}\phi^2 + b\left\{R^{(l')}, j_{l'+1},s_{l'+1}\right\} \phi +
c\left\{ R^{(l')}, j_{l'+1},s_{l'+1}\right\} \right] \geq \gamma_l \bigg\},
\end{align}
\normalsize
where the functions $a(\cdot), b(\cdot),$ and $c(\cdot)$ were defined in \eqref{eq_abc} in Appendix~\ref{appendix_branch_computation_proofs}, and where 
\footnotesize
\begin{align*}
\gamma_l \equiv \lambda \left[ |\term\{R^{(l)}, \tree(\mathcal{B},\nu,\lambda)\}|-1\right]- \sum_{l'=l+1}^{L} h\left\{ R^{(l')}_{sib},\tree(\mathcal{B},\nu,\lambda),y'(\tilde{\phi}, \nu)\right\} - h\left\{R^{(L)}, \tree(\mathcal{B},\nu,\lambda),y'(\tilde{\phi}, \nu) \right\}
\end{align*}
\normalsize
is a constant that does not depend on $\phi$. The first equality simply applies the definitions of $h(\cdot)$ and $g(\cdot)$. The second equality follows from Lemma~\ref{lemma_decompose_specific} and moving terms that do not depend on $\phi$ to the right-hand-side. The third equality follows from plugging in notation from Appendix~\ref{appendix_branch_computation_proofs} and defining the constant $\gamma_l$ for convenience. Thus, \eqref{eq_sprunesinglelevel} is quadratic inequality in $\phi$. We  now just need to argue that its coefficients can be obtained efficiently. 

We need to compute the coefficients in \eqref{eq_spruningfinalform} for $l=0,\ldots,L-1$. The quantities $a\{R^{(l')},j_{l'},s_{l'+1}\}$,  $b\{R^{(l')},j_{l'+1},s_{l'+1}\}$, and  $c\{R^{(l')},j_{l'+1},s_{l'+1}\}$ for $l'=0,\ldots,L-1$ were already computed while computing $S_{grow}(\mathcal{B},\nu)$. To get the coefficients for the left hand side of \eqref{eq_spruningfinalform} for each $l=0,\ldots,L-1$, we simply need to compute $L$ partial sums of these quantities, which takes $O(L)$ operations. As we are assuming that we have access to $\tree(\mathcal{B},\nu,\lambda)$, 
computing $h\{R, \tree(\mathcal{B},\nu,\lambda), y'(\tilde{\phi},\nu)\}$ requires $O(n)$ operations. Therefore, computing $\gamma_0$ takes $O(nL)$ operations. By storing partial sums during the computation of $\gamma_0$, we can subsequently obtain $\gamma_l$ for $l=1,\ldots,L-1$ in constant time.

We have now seen that we can obtain the coefficients needed to express \eqref{eq_sprunesinglelevel} as a quadratic function of $\phi$ for $l=0,\ldots,L-1$ in $O(n L)$ total operations. Once we have these quantities, we can compute each set of the form \eqref{eq_sprunesinglelevel} in constant time using the quadratic equation. 

It remains to compute
\begin{equation}
\label{eq_restatesprune}
S^\lambda(\mathcal{B},\nu) = S_{grow}(\mathcal{B},\nu) \bigcap \left[ \bigcap_{l=0}^{L-1} \left\{ \phi : g\left\{R^{(l)}, \tree(\mathcal{B},\nu,\lambda), y'(\phi,\nu)\right\}  \geq \lambda \right\}
\right].
\end{equation}
Recall from Proposition~\ref{prop_Sisintersection} that  $S_{grow}(\mathcal{B},\nu)$ is the intersection of $O(npL)$ quadratic sets. Thus, in the worst case, $S_{grow}(\mathcal{B},\nu)$ has $O(npL)$ disjoint components, and so this final intersection involves $O(npL)$ components. Thus, we can compute $S^\lambda(\mathcal{B},\nu)$ in $O\{npL \times log(npL)\}$ operations \citep{bourgon2009intervals}.

The following lemma explains why, similar to Proposition~\ref{prop_quadratic}, computation time can be reduced in the special case where $\nu=\nu_{sib}$ and $\mathcal{B} = \branch\{R_A, \tree^\lambda(y)\}$. 
\begin{lemma}
\label{lemma:nusibprune}	
When $\mathcal{B} = \branch\{R_A, \tree^\lambda(y)\}$, the
set \\
$
\{ \phi : g\{R^{(l)}, \tree(\mathcal{B},\nu,\lambda),y'(\phi,\nu_{sib})\} \geq \lambda \}
$
has the form $(-\infty, a_l) \cup (b_l,\infty)$, where $a_l \leq 0 \leq b_l$. Therefore, we can compute $\bigcap_{l=0}^{L-1} \left\{ \phi : g\left\{R^{(l)}, \tree(\mathcal{B},\nu,\lambda), y'(\phi,\nu)\right\}  \geq \lambda \right\}$ as
$$
\bigcap_{l=0}^{L-1} (-\infty, a_l) \cup (b_l, \infty) = (-\infty, \min_{0 \leq l \leq L-1} a_l) \cup (\max_{0 \leq l \leq L-1} b_{l}, \infty)
$$
in $O(L)$ operations. Furthermore, we can compute \eqref{eq_restatesprune} in constant time. 
\end{lemma}
\begin{proof}

As $\mathcal{R}(\mathcal{B}) \subseteq \tree^\lambda(y)$, we can let $\tilde{\phi} \in S_{grow}(\mathcal{B},\nu)$ from Lemma~\ref{lemma_decompose_specific} be $\nu^T y$ such that $y'(\tilde{\phi},\nu) = y$. We can then apply Lemma~\ref{lemma_sameGain} to note that
\small
$$
 \sum_{l'=l}^{L-1} \gain_{R^{(l')}}\{y'(\phi,\nu_{sib}), j_{l'+1}, s_{l'+1}\} =  \left[ \sum_{l'=l}^{L-2} \gain_{R^{(l')}}\{y, j_{l'+1}, s_{l'+1}\} \right] +  \gain_{R^{(L-1)}}\{y'(\phi,\nu_{sib}), j_{L}, s_{L}\}.
$$
\normalsize

Thus, when $\nu=\nu_{sib}$ and $\mathcal{B} = \branch\{R_A, \tree^\lambda(y)\}$, we can rewrite \eqref{eq_spruningfinalform} with all of the terms corresponding to 
$\gain_{R^{(l')}}\{y'(\phi,\nu_{sib}), j_{l'+1}, s_{l'+1}\}$ for $l' = l+1,\ldots,L-2$ moved into the constant on the right-hand-side. This lets us rewrite \eqref{eq_spruningfinalform} as 
$\{ \phi : Gain_{R^{(L-1)}}\{y'(\phi,\nu_{sib}), j_{L}, s_{L}\} \geq \tilde{\gamma}_l \},
$
where $\tilde{\gamma}_l$ is an updated constant that does not depend on $\phi$.

To prove that $\{ \phi : Gain_{R^{(L-1)}}\{j_{L}, s_{L}, y'(\phi,\nu_{sib})\} \geq \tilde{\gamma}_l \}$ has the form  $(-\infty, a_l) \cup (b_l,\infty)$ for $a_l \leq 0 \leq b_l$, first recall from Lemma~\ref{lemma_matrixM} that $Gain_{R^{(L-1)}}\{j_{L}, s_{L}, y'(\phi,\nu_{sib})\}$ is a quadratic function of $\phi$. It then suffices to show that this quadratic has a non-negative second derivative and achieves its minimum when $\phi=0$. The second derivative of this quadratic is $a\left\{R^{(L-1)}, j_L, s_L\right\} = ||\nu_{sib}||_2^{-4} \nu_{sib}^T M_{R^{(L-1)},j_L,s_L} \nu_{sib}$, which is non-negative by Lemma~\ref{lemma_matrixM}. From Lemma~\ref{lemma_matrixM}, $Gain_{R^{(L-1)}}\left\{j_{L}, s_{L}, y'(\phi,\nu_{sib})\right\}$ is non-negative. It equals $0$ when $\phi=0$, because when $\phi=0$ then $\bar{y}_{R^{(L-1)} \cap \chi_{j_L,s_L,1}} = \bar{y}_{R^{(L-1)} \cap \chi_{j_L,s_L,0}}$. 

Intersecting $L$ sets of the form $(-\infty, a_l) \cup (b_l,\infty)$ for $a_l \leq 0 \leq b_l$ only takes $O(L)$ operations, because we simply need to identify the minimum $a_l$ and maximum $b_l$. Finally, Lemma~\ref{lemma_nusibfaster} ensures that $S_{grow}(\mathcal{B},\nu_{sib})$ has at most two disjoint intervals, and so the final intersection with $S_{grow}(\mathcal{B},\nu_{sib})$ takes only $O(1)$ operations. 
\end{proof}

\subsection{Proof of Proposition~\ref{prop_nuprop}}

Let $\mathcal{B} = \branch\{R_A,\tree^\lambda(y)\}$, let
$\mathcal{R}(\mathcal{B}) = \{R^{(0)},\ldots,R^{(L)}\}$, and let $\nu=\nu_{sib}$ \eqref{eq_nusib}. Applying the expression given for $\{y'(\phi,\nu_{sib})\}_i$ in Section~\ref{subsec_intuition} (which follows from algebra), we immediately see that Condition~\ref{cond_nuprop} holds with $R_A = R^{(L)}$ and $R_B = R^{(L-1)} \cap \chi_{j_L,s_L, 1-e_L}$. 

Let $\mathcal{B} = \pi\left[ \branch\{R_A,\tree^\lambda(y)\}\right]$ and let $\nu=\nu_{reg}$ \eqref{eq_nureg}. Note that $\pi\left[ \branch\{R_A,\tree^\lambda(y)\}\right]$ induces the same region $R^{(L)}$ as the unpermuted $ \branch\{R_A,\tree^\lambda(y)\}$. Regardless of the permutation, the induced $R^{(L)}$ is equal to $R_A$. Applying the expression for $\{y'(\phi,\nu_{reg})\}_i$ given in Section~\ref{subsec_intuition}, Condition~\ref{cond_nuprop} holds with constant $c_2=0$.

\section{Proofs for Section~\ref{subsec_actuallycomputeSreg}}

\subsection{Proof of Proposition~\ref{prop:subsetpermutation}}
\label{proof_subsetpermutations}

As stated in Proposition~\ref{prop:subsetpermutation}, let
\footnotesize
$$
p_{reg}^{Q}(y) = pr_{H_0}\left\{ |\nu_{reg}^\T Y - c | \geq |\nu_{reg}^\T y - c| \mid \bigcup_{\pi \in Q} %\left\{ 
\mathcal{R}(\pi[\branch\{R_A, \tree^\lambda(y)\}]) \subseteq \textsc{tree}^{\lambda}(Y),
%\right\}, 
\mathcal{P}_{\nu_{reg}}^\perp Y = \mathcal{P}_{\nu_{reg}}^\perp {y} \right\}.
$$
\normalsize
First, we will show that the test based on $p_{reg}^Q(y)$ controls the selective Type 1 error rate, defined in \eqref{eq_st1e}; this is a special case of Proposition 3 from \cite{fithian2014optimal}. 

Define $\mathcal{E}_1 = \{Y : R_A \in \textsc{tree}^\lambda(Y)\}$, $\mathcal{E}_2 = \{ Y : \mathcal{P}_{\nu_{reg}}^\perp Y = \mathcal{P}_{\nu_{reg}}^\perp y\}$, and \\
$\mathcal{E}_3 = \left\{Y : 
\bigcup_{\pi \in Q} 
\mathcal{R}\left(\pi\left[\branch\left\{R_A, \tree^\lambda(y)\right\}\right]\right)
 \subseteq \textsc{tree}^\lambda(Y)\right\}$. Recall that the test 
 of $H_0: \nu_{reg}^\T \mu = c$ based on $p_{reg}^Q(y)$ controls the selective Type 1 error rate if, for all $\alpha \in [0,1]$, $pr_{H_0}\{p_{reg}^Q(Y) \leq \alpha \mid \mathcal{E}_1\} \leq \alpha$. By construction, 
$$
pr_{H_0}\left\{ p_{reg}^Q(Y) \leq \alpha \mid \mathcal{E}_2 \cap \mathcal{E}_3 \right\} = E\left[1_{\{pr_{H_0}\left( |\nu_{reg}^\T Y - c | \geq |\nu_{reg}^\T y - c| \mid \mathcal{E}_2 \cap \mathcal{E}_3 \right) \leq \alpha\}}  \mid \mathcal{E}_2 \cap \mathcal{E}_3 \right] = \alpha.
$$
Let $\psi_{R_A}^Q = 1_{\{pr_{H_0}\left( |\nu_{reg}^\T Y - c | \geq |\nu_{reg}^\T y - c| \mid \mathcal{E}_2 \cap \mathcal{E}_3 \right) \leq \alpha\}}$. An argument similar to that of Lemma~\ref{lemma_permutations} indicates that  $\mathcal{E}_3 \subseteq  \mathcal{E}_1$. The law of total expectation then yields
\small
$$
E(\psi_{R_A}^{Q} \mid \mathcal{E}_1) =  E\left\{ E(\psi_{R_A}^{Q}\mid  \mathcal{E}_1 \cap \mathcal{E}_2 \cap \mathcal{E}_3) \mid  \mathcal{E}_1 \right\} = E\left\{E(\psi_{R_A}^{Q}\mid \mathcal{E}_2 \cap \mathcal{E}_3) \mid  \mathcal{E}_1  \right\} = E(\alpha \mid  \mathcal{E}_1) = \alpha.
$$
\normalsize
Thus, the test based on $p_{reg}^Q(y)$ controls the selective Type 1 error rate. We omit the proof that $p_{reg}^{Q}(y)$ can be computed as
$$p_{reg}^{Q}(y)= pr_{H_0}\left\{ |\phi- c | \geq |\nu_{reg}^\T y - c| \mid \phi \in \bigcup_{\pi \in Q} S^\lambda\left( \pi[\branch\{R_A, \tree^\lambda(y)\}], \nu_{reg}\right) \right\}$$ 
for $\phi \sim N(c, \|\nu_{reg}\|_2^2 \sigma^2)$, as the proof is similar to the proof of Theorem~\ref{theorem_1jewell}.

\section{Effect of Using the Computationally Efficient Alternative in Section~\ref{subsec_actuallycomputeSreg}}
\label{appendix:permutation_sims}

In this section, we investigate the effect of using $p_{reg}^{\mathcal{I}}(y)$ from \eqref{eq_identitypvalue} rather than $p_{reg}(y)$ from \eqref{eq:pvalreg} on power. We also investigate the effect of using \eqref{eq:branchCI} rather than \eqref{eq_mainCIreg} on the width of confidence intervals for $\nu_{reg}^\T \mu$. 

We generate data as described in Section~\ref{subsubsec_datagen}, but for simplicity we restrict our attention to the case where $a=1$ (corresponding to the center panel of Figure~\ref{fig:truetree}).

For each tree that we build, we consider (a) testing $H_0: \nu_{reg}^\T \mu = 0$ and (b) constructing a confidence interval for $\nu_{reg}^\T \mu$, for each region appearing at the third level of the tree. For each test,  we compare the test that uses the full conditioning set (i.e. that uses $p_{reg}(y)$ from \eqref{eq:pvalreg}) to the test that uses the identity permutation only (i.e. that uses $p_{reg}^{\mathcal{I}}(y)$ from \eqref{eq_identitypvalue}). For each interval, we compare the method that uses the full conditioning set (i.e. \eqref{eq_mainCIreg}) to the method that uses the identity permutation only (i.e. \eqref{eq:branchCI}). The results are displayed in Figure~\ref{fig_perm}. 

The left panel of Figure~\ref{fig_perm} shows that the power loss resulting from using \eqref{eq_identitypvalue} instead of \eqref{eq:pvalreg} is negligible. In fact, we need to zoom in on the left panel, as shown in the center panel, to see any separation between the power curves. We see in the center panel that power is lower when \eqref{eq_identitypvalue} is used, though we emphasize that the differences in power are extremely small. 

We see in the right panel of Figure~\ref{fig_perm} that the computationally efficient conditioning set has a more noticeable impact on the median width of our confidence intervals. As expected, the confidence intervals are narrower when we use the full conditioning set (i.e. \eqref{eq_mainCIreg}). However, this difference is most noticeable when $b$ is small. When $b$ is small, in which case the confidence intervals are wide even when the full conditioning set is used. Thus, the amount of precision lost overall by constructing confidence intervals using the identity permutation only (i.e. \eqref{eq:branchCI}) is not of practical importance. 

Based on these results, we recommend using the identity permutation in practice, because it is the most computationally efficient choice \emph{and} does not meaningfully reduce power or precision compared to the full conditioning set.

\begin{figure}
\centering
\includegraphics[width=0.8\textwidth]{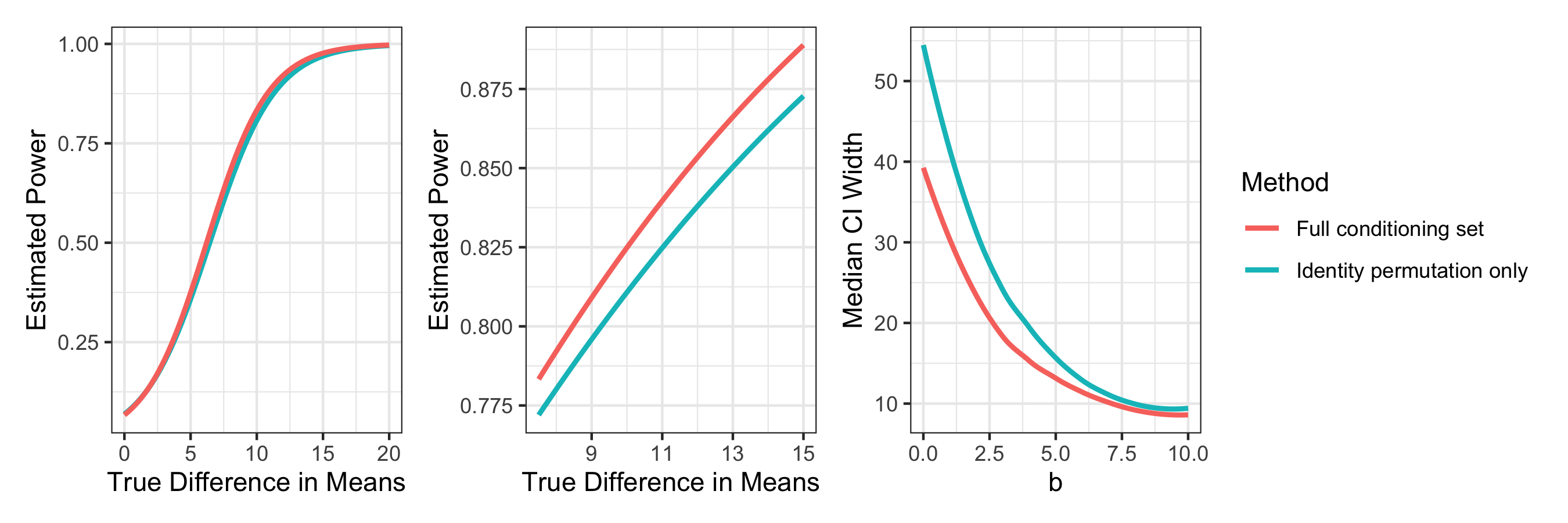}
\caption{
Simulation results comparing inference based on the full conditioning set to inference based on the identity permutation only (see Section~\ref{subsec_actuallycomputeSreg}). The left panel shows power curves. The center panel zooms in on one section of the left panel. The right panel shows median widths of confidence intervals. }	
\label{fig_perm}
\end{figure}

\section{Robustness to Non-Normality}
\label{appendix:non-normal}

In this section, we explore the performance of the selective $Z$-test under a global null when the normality assumption on $Y$ is violated. 
For four choices of cumulative distribution function $F$, we generate $Y_i \overset{\text{i.i.d.}}{\sim} F$ such that all observations have the same expected value. We set $n=200,p=10$, and we grow trees to a maximum depth of $3$. We plug in $(n-1)^{-1} \sum_{i=1}^n (y_i -\bar{y})^2$ as an estimate of $\sigma^2$, as it does not make sense to assume known variance for distributions with a mean-variance relationship. Figure~\ref{fig_nonnormal} displays quantile-quantile plots of the p-values for testing $H_0: \nu_{sib}^\T \mu = 0$, using the test in \eqref{def_mainpval}. 

Figure~\ref{fig_nonnormal} shows that despite the fact that our proposed selective inference framework was derived under a normality assumption, it yields approximately uniformly distributed p-values for Poisson(10), Bernoulli(0.5), and Gamma(1,10) data.  We suspect that this is because $\mathcal{P}_\nu^\perp Y$ is approximately independent of $\nu^\T Y$ for these distributions (see the proof of Theorem~\ref{theorem_1jewell} in Appendix~\ref{appendix:section3proofs}). In the case of Bernoulli(0.1) data, which represents a particularly extreme violation of the normality assumption, the p-values from our selective $Z$-test are not uniformly distributed. As mentioned in Section~\ref{section_disc}, future work could involve characterizing conditions for $F$ under which our selective $Z$-tests will approximately control the selective Type 1 error. 

\begin{figure}
\centering
\includegraphics[width=0.8\textwidth]{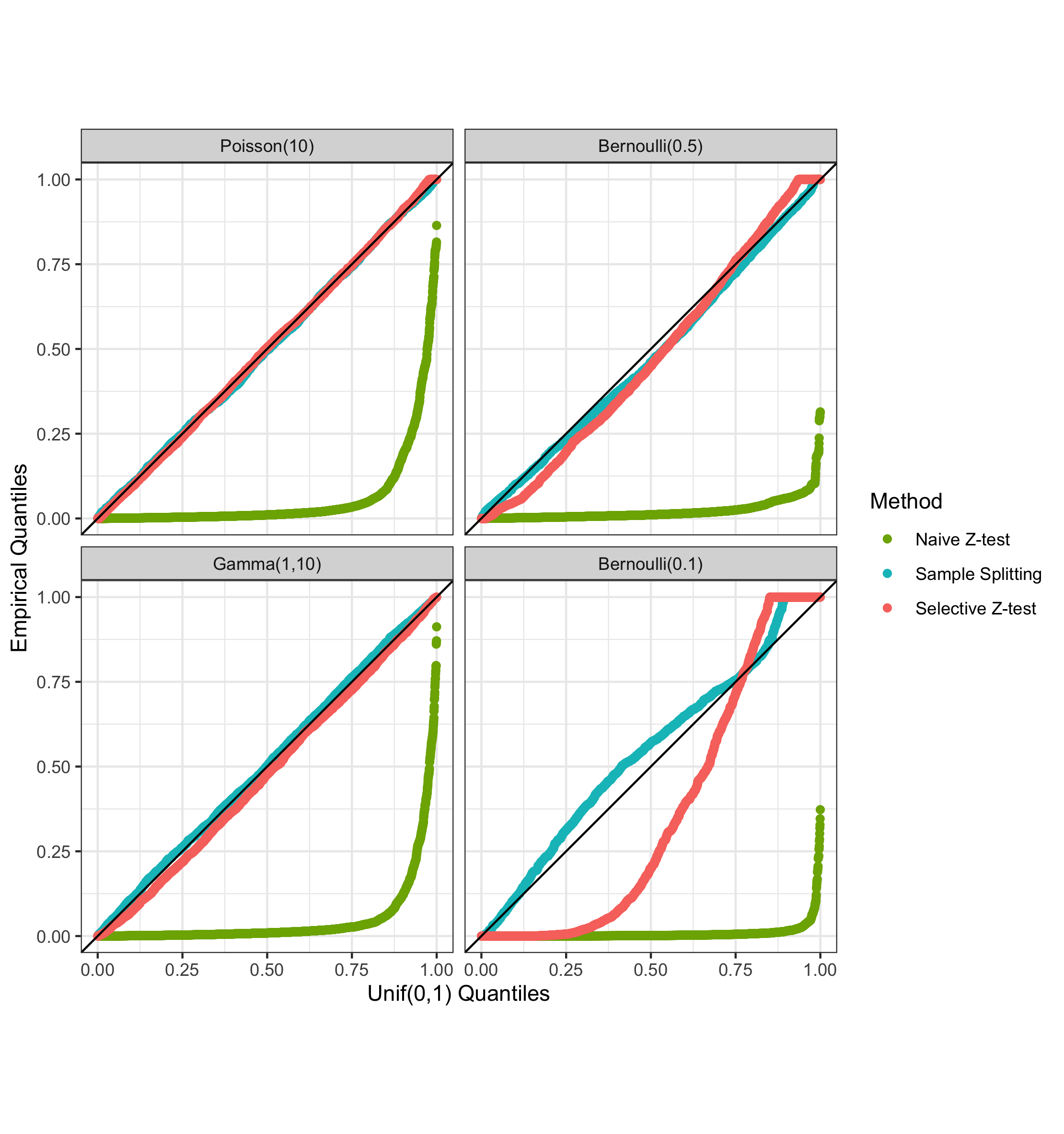}
\caption{Quantile-quantile plots of the p-values for testing $H_0: \nu_{sib}^\T \mu = 0$ under a global null. A naive Z-test (green), sample splitting (blue), and selective Z-test (pink) were performed; see Section \ref{subsubsec_methods}.}	
\label{fig_nonnormal}
\end{figure}

\section{Alternate Box Lunch Study Analysis}
\label{appendix_bls}

Figure~\ref{fig_bls_alt} is the same as the left panel of Figure~\ref{fig_BLStrees}, but the selective $Z$-inference is carried out with $\hat{\sigma}_{\text{cons}}$, from Section~\ref{subsec:unknownvar}, rather than $\hat{\sigma}_{\text{SSE}}$. The takeaways presented in Section~\ref{section_realData} do not change.

\begin{figure}
\centering
\includegraphics[width=0.6\textwidth]{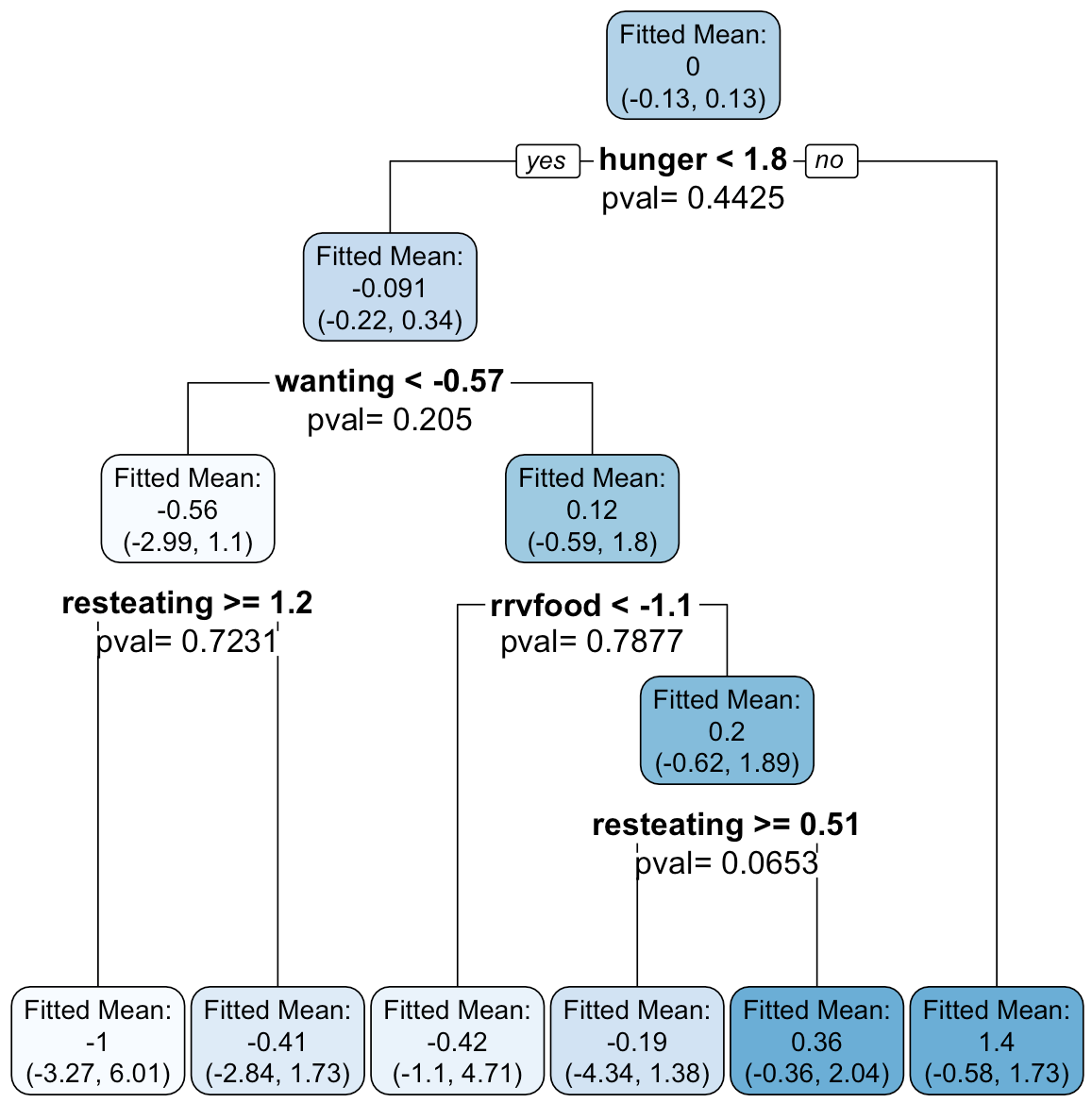}
\caption{A CART tree fit to the Box Lunch Study data. Each split has been labeled with a p-value \eqref{def_mainpval}, and each region has been labeled with a confidence interval \eqref{eq:branchCI}. Inference is carried out by plugging in $\hat{\sigma}_{\text{cons}}$, from Section~\ref{subsec:unknownvar}, as an estimate of $\sigma$. }
\label{fig_bls_alt}	
\end{figure}

\end{document}